\documentclass[aps,prx,onecolumn,superscriptaddress,nofootinbib,floatfix]{revtex4-2}
\usepackage[T1]{fontenc}
\usepackage{lmodern}
\usepackage[english]{babel}
\usepackage{physics}
\usepackage{bbm}
\usepackage{tikz}
\usetikzlibrary{arrows.meta}
\usetikzlibrary{calc}
\usetikzlibrary{fit}
\usetikzlibrary{backgrounds}
\usepackage{graphicx} 
\DeclareMathAlphabet{\mathpzc}{OT1}{pzc}{m}{it}
\usepackage{latexsym}
\usepackage{amssymb}
\usepackage[framemethod=TikZ]{mdframed}
\usepackage{multirow}
\usepackage{booktabs}
\newmdenv[
  topline=false,
  bottomline=false,
  rightline=false,
  linewidth=2pt,
  linecolor=black,
  skipabove=\topsep,
  skipbelow=\topsep,
  leftmargin=10pt,
  innerleftmargin=10pt,
  innertopmargin=0pt,
  innerbottommargin=0pt
]{definition}

\usepackage[framemethod=TikZ]{mdframed}
\usepackage{amsthm} 

\newmdenv[
  leftline=false,
  rightline=false,
  linewidth=1pt,
  skipabove=\topsep,
  skipbelow=\topsep,
  innertopmargin=5pt,
  innerbottommargin=5pt
]{theorem}

\usepackage{amsmath}
\usepackage{amsfonts}
\usepackage{graphicx}
\usepackage{qcircuit}
\usepackage{mathtools}
\usepackage{algorithmic}
\makeatletter
\newcounter{algorithm}
\renewcommand{\thealgorithm}{\arabic{algorithm}}
\newenvironment{algorithm}[1][tbp]{%
  \begin{table}[#1]%
    \def\fnum@table{Algorithm~\thealgorithm}%
    \let\c@table\c@algorithm
    
    \centering\hrule height.8pt\kern3pt\relax
}{%
    \kern3pt\hrule height.8pt
  \end{table}%
}
\makeatother
\usepackage{listings}

\newcommand{\VOCV}{V_{\mathrm{OCV}}}
\newcommand{\dQdV}{\mathrm{d}Q/\mathrm{d}V}
\newcommand{\SOC}{\mathrm{SOC}}
\newcommand{\LMO}{\mathrm{LiMn_2O_4}}
\newcommand{\LFP}{\mathrm{LiFePO_4}}
\newcommand{\FP}{\mathrm{FePO_4}}
\newcommand{\Trec}{T_{\mathrm{rec}}}
\newcommand{\cQFT}{\mathrm{cQFT}}
\newcommand{\icQFT}{\mathrm{cQFT}^{-1}}

\newcommand{\HSpace}{\mathcal{H}}

\newcommand{\jw}{\mathrm{JW}}
\newcommand{\Trsmall}{\operatorname{tr}}

\usepackage[colorlinks=true, allcolors=blue]{hyperref}

\begin{document}

\title{Quantum Algorithm for Open-System Battery Cathodes by Modeling Multiple Strongly Coupled Holstein Polarons with Chain-Mapped Caldeira-Leggett Dynamics}

\author{Joshua M. Courtney}
\affiliation{Department of Physics and Astronomy and the Center for Simulational Physics, University of Georgia, Athens, GA 30602, USA}

\date{\today}

\begin{abstract}
Cathode lithiation occupies a chemical regime of tightly localized orbitals, narrow bandwidths, and strong electron--lattice coupling.
The defining electrochemical observables (open-circuit voltage and differential capacity) are open-system, reservoir-equilibration quantities that closed-Hamiltonian quantum simulation cannot produce, set by exchange with electron, Li$^+$, and phonon baths. 
We present a fault-tolerant quantum algorithm that recovers them through a unitary chain-mapped Caldeira--Leggett embedding, rendering the baths Trotterizable. 
The resulting fourth-order Trotter step has a T-gate count polynomial in system size, validating its open-system dynamics against hierarchical equations of motion (HEOM) at strong coupling and the Lindblad limit at weak coupling. 
For single-carrier olivine LiFePO$_4$, a single voltage anchor on an otherwise DFT-fixed Hamiltonian places the differential-capacity peak within the $\pm5$~mV reproducibility of the experimental plateau. 
For multi-carrier spinel LiMn$_2$O$_4$, whose $1{:}1$ Mn$^{3+}$/Mn$^{4+}$ filling makes the inter-site Coulomb repulsion dynamically active, the same kernel yields a two-plateau voltage curve with a $125$~mV split, within $17\%$ of the observed $150$~mV.
To our knowledge, we deliver the first end-to-end fault-tolerant resource estimate for such a multi-carrier, three-reservoir observable: $368$ logical qubits and $\sim3\times10^5$ T-gates per step, or $\sim1.7\times10^{12}$ T-gates for a full voltage curve (parallelizable over $\sim10^3$ trajectories), leaving the production-scale dynamical run as a milestone for future hardware.
The same kernel reproduces macroscopic quantum coherence, two-band superconductivity, and the Mikheyev--Smirnov--Wolfenstein resonance without modification, placing dynamical battery chemistry and similar Hamiltonians within scope for fault-tolerant quantum simulation.
\end{abstract}

\maketitle

\section{Introduction}\label{sec:Introduction}

Lithium-ion cathodes underpin grid-scale storage and electrified transport~\cite{ellis2010positive}.
Among others, two commercially dominant manganese- and iron-based chemistries (spinel LiMn$_2$O$_4$ (LMO) and olivine $\LFP$) owe their electrochemistry to carrier self-trapping as phonon-dressed small polarons on the transition-metal sites~\cite{padhi1997phospho, thackeray1983w, thackeray1984electrochemical}. 
These polarons move in concert with Li$^+$ insertion through the lattice~\cite{yip2007handbook, maxisch2006ab}. 
Macroscopic quantities and observables that decide a cell's usefulness can be modeled with nonadiabatic, strongly-coupled, finite-temperature dynamics of polarons, exchanging energy and charge with the surrounding reservoirs (thermal baths).
Computing these quantities from the polaron dynamics perspective has remained a materials-design bottleneck where cathode discovery is still largely empirical, motivating a route to predict voltage and rate behavior directly from a DFT-parameterized Hamiltonian~\cite{aydinol1997ab}, shortening the design loop for cobalt-free, manganese-rich chemistries (LMO and its manganese-rich relatives LMNO/LMRO/NMC).

The relevance is not confined to batteries.
Multi-polaron transport involves many carriers self-trapping, repelling one another through Coulomb interaction, and coupling to a shared phonon bath. 
This is the same organizing physics of colossal magnetoresistive manganites, polaronic conduction in transition-metal oxides, charge mobility in organic semiconductors and light-emitting diodes~\cite{coropceanu2007charge}, and normal-state anomalies of doped correlated oxides~\cite{tokura2006critical, yip2007handbook, perroni2002effects}.
The regime approaching computational intractability for each of these fields is also occupied by battery chemistry, where intermediate-to-strong electron--phonon coupling ($S \gtrsim 1$) at finite carrier density cannot be accurately handled by weak-coupling perturbation theory or Born--Oppenheimer approximation, being difficult to scale using tensor network methods.
A quantum simulation primitive resolving nonadiabatic multi-polaron dynamics as an open system therefore has reach well beyond cathode design. 
We choose the battery problem as a concrete, economically consequential instance.

Characterizing a battery cathode requires observables defined by reservoir equilibration, where electron current collector, the Li$^+$ electrolyte, and the phonon thermal bath are the systems the cathode equilibrates with. 
These observables include open-circuit voltage $\VOCV(\SOC)$, differential capacity $\dQdV(V)$, impedance $Z(\omega)$, exchange current $i_0$, chemical diffusion coefficient $D_{\mathrm{Li}}$, and activation energy $E_a$.

Cathode interiors compound difficulty in modeling: multiple charge carriers self-localize into phonon-dressed polarons~\cite{yip2007handbook, ellis2010positive} (Huang--Rhys factor $S \gtrsim 1$, beyond Lang--Firsov), interact through Coulomb correlations, and couple to the reservoirs above. 
This acts as a conjunctive regime (many polarons, strong coupling, multiple baths, DFT-derived parameters) which strains canonical classical methods, such as DMRG bond dimension, HEOM hierarchy depth, ML-MCTDH layer count, and the sign structure of diagrammatic Monte Carlo~\cite{prokof1998polaron}, all growing steeply with carrier count or coupling~\cite{white1992density, wang2003multilayer, tanimura2020numerically, hohenadler2003spectral, alexandrov2008polarons, perroni2002effects}. 
To our knowledge, existing quantum simulation work has treated only the single-polaron~\cite{macridin2018digital, macridin2018electron} or closed-system~\cite{bauer2016hybrid, cade2020strategies} limits. 
Digital and analog Holstein-circuit simulators~\cite{mezzacapo2012digital, stojanovic2012quantum, mei2013analog, stojanovic2014transmon, stojanovic2023extracting, lamata2014efficient}, variational and hybrid electron--phonon solvers~\cite{denner2023hybrid, li2023efficient, backes2023dynamical, kumar2025digital, torabian2025lattice}, and interaction-picture fast-forwarding via the polaron transform~\cite{apel2026quantum} likewise target single-carrier, weak-coupling, or closed/encoding-limited regimes rather than the conjunctive open-system problem treated here. 
We consider two examples, being olivine $\LFP$ in its single-carrier cell, and spinel LMO at its $1{:}1$ Mn$^{3+}$/Mn$^{4+}$ filling, where several Jahn--Teller-active polarons coexist and inter-site Coulomb repulsion becomes dynamically active~\cite{marianetti2001first, rodriguez1998electronic}. 
We use the term ``first-principles'' to describe model construction, being a looser usage than canonical nonadiabatic chemistry literature. 
The inputs used here (couplings, site energies, screening) are DFT-derived, but the LFP/FP two-phase plateau enters at the coarse-grained mean-field level (Section~\ref{ssec:lfp_instantiation}). 
The result is a first-principles-parameterized prediction instead of an ab-initio simulation of the two-phase boundary.

We close an open-system observable gap using first-quantized quantum dynamics on digital quantum circuits.
Using a unitary chain-mapped Caldeira--Leggett embedding~\cite{caldeira1983path, chin2010exact, prior2010efficient} with the Tamascelli finite-temperature extension~\cite{tamascelli2019efficient} produces reservoir-equilibration observables ($\VOCV$, $\dQdV$, $\dots$) with no secular-Markov approximation. 
This is uncontrolled at LFP coupling ($\varepsilon_{\mathrm{Markov}}^{\mathrm{secular}} \approx 0.29$ (Eq.~\ref{eq:born_markov_bounds})). 
Chain mapping circumvents the approximation entirely by rendering the bath unitary and Trotterizable~\cite{tamascelli2018nonperturbative}.
Adapting chain-mapping Caldeira-Leggett embedding to quantum circuits, we compose a closed-form Trotter step with seven gate primitives whose per-step T-count is polynomial in the system size.

The LFP regime lies outside the Born--Markov expansion (quantified in \ref{ssec:chain_construction}), so we represent the bath as an explicit chain of `dark' modes, making available a Trotterizable Clifford$+$T circuit whose only approximation is the finite chain length $K$ ($\Trec \sim \pi K/\Omega_c$), coinciding with HEOM at sufficient depth~\cite{tamascelli2018nonperturbative}.
A single $\LFP$ voltage anchor sets $\VOCV(\SOC{=}0.5) \equiv 3.45$~V, making the $\dQdV$ peak a structural prediction within the $\sim\pm 5$~mV reproducibility of the Yamada plateau~\cite{yamada2006room} (converged FWHM $\approx 12.6$ vs $\sim 20~\mathrm{mV}$), predicting plateau extent and solid-solution shoulders.
Open-system dynamics are cross-checked against numerically-exact HEOM at strong coupling. 
Its multi-polaron companion is spinel LMO (LiMn$_2$O$_4$), for which we deliver the first complete fault-tolerant resource costing of a multi-carrier, active-Coulomb, three-reservoir electrochemical observable.
Specifically, we provide an explicit multi-polaron Trotter step and an exact emitted-circuit operation count (Section~\ref{sec:lmo_spec}), being at per-step parity with the single-carrier case. 
To our knowledge, no prior closed-system battery fault-tolerant study~\cite{delgado2022simulating, zini2023quantum, fomichev2024simulating, kunitsa2025quantum, loaiza2026quantum} costs a dynamical multi-carrier open-system observable.
The two-plateau $\VOCV$ split is a derived equilibrium-tier electrostatic prediction (Fig.~\ref{fig:lmo_twoplateau}), with the full quantum-dynamical production trajectory the declared milestone.

The algorithmic kernel bilinearly couples a few slow degrees of freedom to a chain-mapped continuum, evolves unitarily, and is read out by partial trace (measuring select sub-registers). 
We confirm the recovery of a macroscopic-quantum-coherence double-well~\cite{leggett1984quantum, leggett1987dynamics} and the two-band BCS superconductor~\cite{suhl1959bardeen, choi2002origin} unmodified, with the Mikheyev--Smirnov--Wolfenstein resonance~\cite{wolfenstein2018neutrino, mikheyev1988neutrino} as an algebraic two-channel check to composite extensions of similar Hamiltonians in the same family.

The closed-form gate count presented is of the emitted production circuit (Section~\ref{ssec:scope}), with a production register size of $376$ logical qubits, $\sim3.0\times10^5$ T-gates per Trotter step, and $\sim1.7\times10^{8}$ per trajectory under the production fourth-order Suzuki composition (Section~\ref{ssec:trotter_step}; $\sim3.0\times10^{8}$ at the empirically gate-validated second-order comparison point).
We separate these estimates from the per-anchor estimation cost, requiring $K \approx 1.6\times10^2$ coherent repetitions under the direct measurement protocol (Section~\ref{sec:resources}).
The direct measurement protocol puts a full $\VOCV$ curve at $\sim8\times10^{11}$ T-gates; the finite-difference envelope (quantum amplitude estimation, QAE: $K_{\mathrm{QAE}} \approx 1.8\times10^4$), retained to control worst-case expectations, raises this to $\sim6\times10^{14}$.

\section{Background}\label{sec:Background}

\subsection{Prior Work}

Grid-based quantum simulation assigns a qubit register to each continuous degree of freedom, acting as a first-quantized Schr\"odinger evolution. 
Product-formula approaches alternate non-commuting propagators (e.g. kinetic and potential operators) in a split-operator Trotter decomposition and iterate for a desired evolution time~\cite{wiesner1996simulations, zalka1998simulating,kassal2008polynomial}. 
Ollitrault \emph{et al.}~\cite{ollitrault2020nonadiabatic} demonstrated the approach for nonadiabatic chemical dynamics on digital circuits for a Marcus model with a linearly-approximated nonadiabatic coupling. 
This work was extended to a closed-system, multi-channel form with validated extensions to vibronics, chemical scattering, dissociation, and fluorescence~\cite{courtney2026oracle}, adapting uniformly-controlled-rotation function evaluation~\cite{mottonen2004quantum} and
XOR-class fragmentation of off-diagonal couplings~\cite{courtney2026oracle, motlagh2025quantum}. 
Here we take these primitives as given and extend to the open-system setting.

For batteries specifically, existing fault-tolerant literature targets closed-system ground-state energetics including equilibrium voltage and stability via qubitized phase estimation~\cite{delgado2022simulating, zini2023quantum} and spectroscopic response via time-domain methods~\cite{fomichev2024simulating, fomichev2025fast, kunitsa2025quantum, loaiza2026quantum}.
Noisy intermediate-scale quantum (NISQ) studies of cathode fragments act as the variational counterpart~\cite{farag2022towards}.
Polaron physics on quantum hardware has been treated in the single-carrier or closed-system limits: digital and analog Holstein simulators~\cite{macridin2018digital, macridin2018electron, mezzacapo2012digital, stojanovic2012quantum, mei2013analog, stojanovic2014transmon, stojanovic2023extracting, lamata2014efficient}, variational and hybrid electron--phonon solvers~\cite{denner2023hybrid, li2023efficient, backes2023dynamical, kumar2025digital, torabian2025lattice}, and polaron-transform fast-forwarding~\cite{apel2026quantum}.
Generic open-system simulation on circuits proceeds via Lindblad dilation~\cite{cleve2016efficient, childs2016efficient, ding2024simulating}, presuming the weak-coupling master equation this regime violates (Section~\ref{ssec:problem_hamiltonian}). 
Reservoir-equilibration electrochemical observables at strong coupling and finite temperature remain understudied in the quantum-algorithmic perspective, being the one we take here, as summarized as Algorithm~\ref{alg:REQWIEM} and Figure~\ref{fig:workflow}.

\begin{algorithm}
\caption{Open-system REQWIEM step (chain-mapped Caldeira--Leggett)}
\label{alg:REQWIEM}
\begin{algorithmic}
\REQUIRE $H = H_S + \sum_r \bigl[H^{(r)}_{\mathrm{chain}} + V^{(r)}_{SB}\bigr] + H_{\mathrm{CT}}$; reservoirs $r \in \{e,\,\mathrm{Li},\,\mathrm{ph}\}$; $L$ sites, $K_r$ chain modes, $n$ qubits/mode, $N_{\mathrm{step}}$ steps
\STATE Prepare polaron $+$ Li system state (unitary coupled cluster, UCC); all chain registers in vacuum (Tamascelli doubling encodes finite $T$ exactly)
\FOR{$j = 1$ to $N_{\mathrm{step}}$}
    \STATE $V_{\mathrm{diag}}(\tau/2)$: on-site energies, Coulomb controlled phases, Holstein and Jahn--Teller couplings, counter-term; soft-Coulomb terms via reversible $1/\sqrt{\cdot}$ arithmetic ancilla
    \STATE $V_{SB}(\tau/2)$: bilinear system--chain phases; $T_{\mathrm{chainhop}}(\tau/2)$: bosonic chain hops (inner symmetric split)
    \STATE Kinetic layers $(\tau/2)$: cQFT $\to$ diagonal $p^2$ $\to$ cIQFT per bosonic register
    \STATE $T_{\mathrm{hop}}(\tau)$: Jordan--Wigner carrier hopping, even/odd color classes (Li layer, then polaron layer)
    \STATE Mirror the half-step layers (symmetric Strang product)
\ENDFOR
\STATE Read out by tracing over chain registers: populations, chemical potentials, correlation functions
\end{algorithmic}
\end{algorithm}

\begin{figure}[t]
\centering
\resizebox{\textwidth}{!}{%
\begin{tikzpicture}[
  font=\scriptsize, >=stealth,
  box/.style={draw, rounded corners, align=center, text width=2.45cm,
              minimum height=1.25cm, inner sep=2pt}]
\node[box] (s1) at (0,0) {\textbf{1.\ Physical model}\\[1pt] cathode cell;
   $e^-$/Li$^+$/phonon reservoirs\\[1pt] $H$ (Eq.~\ref{eq:master_H})};
\node[box] (s2) at (3.25,0) {\textbf{2.\ Chain mapping}\\[1pt] TEDOPA $+$
   Tamascelli doubling\\[1pt] $\Trec=\pi K/\Omega_c$, $K{=}10$};
\node[box] (s3) at (6.5,0) {\textbf{3.\ Compilation}\\[1pt] 4th-order Strang
   $\to$ Clifford$+$T\\[1pt] poly($L$); ${\sim}3.0{\times}10^5$ T/step};
\node[box] (s4) at (9.75,0) {\textbf{4.\ Execution}\\[1pt]
   $\max(n_{\mathrm{sh}},K){\approx}160$ reps\\[1pt] within $\Trec$};
\node[box] (s5) at (13.0,0) {\textbf{5.\ Readout}\\[1pt] partial trace;
   $(\mu,\beta)$ sweep\\[1pt] $V{=}\mu/e$ control};
\node[box] (s6) at (16.25,0) {\textbf{6.\ Observables}\\[1pt]
   $\VOCV$, $\dQdV$, $Z(\omega)$};
\foreach \a/\b in {s1/s2,s2/s3,s3/s4,s4/s5,s5/s6}
   \draw[->,very thick] (\a) -- (\b);
\draw[rounded corners, fill=black!6, draw=black!30]
   (-1.35,-2.35) rectangle (17.6,-1.35);
\draw[dashed, black!55] (8.0,-2.3) -- (8.0,-1.4);
\node[align=center, text width=8.2cm] at (3.3,-1.85)
   {\textbf{Validated} on simulable instances: quality control $10^{-16}$ $\cdot$
    ED $6{\times}10^{-11}$ $\cdot$ HEOM $2.07\%$ $\cdot$ Strang ratio $3.99$};
\node[align=center, text width=8.2cm] at (12.9,-1.85)
   {\textbf{Extrapolated} to $292$-qubit production: $3{\times}10^5$ T/step
    $\cdot$ $376$ qubits $\cdot$ $d{=}25$ $\cdot$ $\sim 7$~h ($\LFP$);
    \\ $3{\times}10^5$ T/step $\cdot$ $368$ qubits $\cdot$ $\sim 15$~h ($\LMO$)};
\end{tikzpicture}}
\caption{End-to-end workflow. The open-system cathode Hamiltonian (1) is chain-mapped (2; time-evolving density operator with orthogonal polynomials TEDOPA $+$ Tamascelli finite-$T$ doubling), compiled into a Clifford$+$T fourth-order Trotterization whose per-step T-count is polynomial in system size (3; Theorem~\ref{thm:poly_L}), executed as a trajectory ensemble within the recurrence window (4), and read out by partial trace and an anchor sweep with $V=\mu/e$ set as a control (5) to yield the electrochemical observables (6). The band marks the boundary (Section~\ref{ssec:scope}) between classically-validated instances (left) and the calculated production geometry (right).}
\label{fig:workflow}
\end{figure}
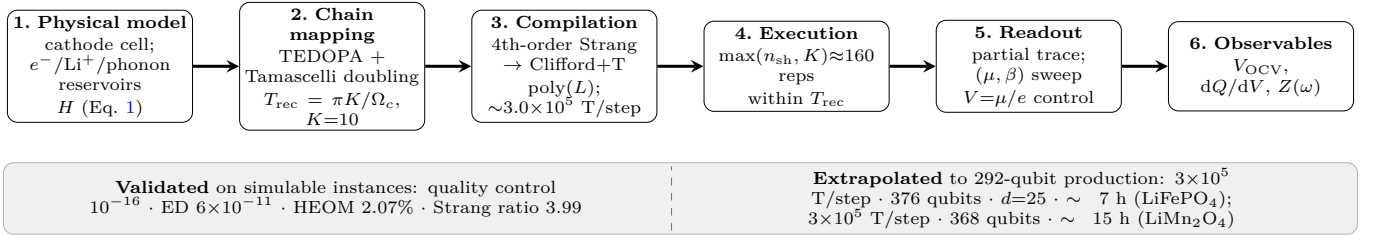

\subsection{The Open-System Holstein--Coulomb Hamiltonian}\label{ssec:problem_hamiltonian}

We consider a system composed by a lattice of $L$ transition-metal sites hosting small polarons, coupled to a migrating Li$^+$ degree of freedom and to local phonons. 
In Jordan--Wigner occupation operators $n_i = c_i^\dagger c_i$,

\begin{equation}\label{eq:master_H}
\begin{split}
H_S &= -\tfrac{1}{2 m_{\mathrm{Li}}}\partial_R^2 + V_{\mathrm{wb}}(R)
    + \sum_i \varepsilon_i\, n_i
    - t_{ij}\sum_{\langle i,j\rangle} \bigl(c_i^\dagger c_j + \mathrm{h.c.}\bigr) \\
    &\quad + \sum_i V_{\mathrm{LiP}}(R - x_i)\, n_i
    + \tfrac{1}{2}\sum_i \omega_{\mathrm{ph}}\bigl(q_i^2 + p_i^2\bigr)
    + g\sum_i n_i\, q_i
    + \sum_{i<j} U_{ij}\, n_i n_j,
\end{split}
\end{equation}

with $R$ the Li coordinate in a washboard potential $V_{\mathrm{wb}}$, $V_{\mathrm{LiP}}(r) = -[\varepsilon_\infty\sqrt{r^2 + \rho_c^2}\,]^{-1}$ the screened Li--polaron attraction, $g$ the Holstein coupling (Huang--Rhys $S = g^2/2\omega_{\mathrm{ph}}^2 \gtrsim 1$ in the cathode regime), and $U_{ij} = U_{\mathrm{nn}}/|i-j|^{\nu}$ the screened inter-site Coulomb repulsion. 
The Hamiltonian is at most $2$-local in occupations $\{n_i\}$ (Section~\ref{sec:REQWIEM}). 
Material specification fixes parameters and carrier encoding. 
So, single-carrier $\LFP$ retains a continuous $R$ register, while multi-carrier $\LMO$ replaces it by a discrete Li-site occupancy register and adds a Jahn--Teller doublet per Mn$^{3+}$ (Sections~\ref{sec:lfp_spec}, \ref{sec:lmo_spec}).

Open-circuit voltage $\VOCV(\SOC)$, differential capacity $\dQdV(V)$, impedance $Z(\omega)$, exchange current $i_0$, diffusion coefficient $D_{\mathrm{Li}}$, and activation energy $E_a$ are defined by exchange with the electron current collector, the Li$^+$ electrolyte, and the phonon thermal bath~\cite{feynman2000theory, breuer2002theory}.
Each reservoir enters as a continuum bosonic (or, for the lead, fermionic) bath in Caldeira--Leggett form~\cite{caldeira1983path},

\begin{equation}\label{eq:H_CL_main}
H_{\mathrm{tot}} = H_S
  + \sum_k \tfrac{1}{2}\bigl(p_k^2 + \omega_k^2 q_k^2\bigr)
  + q\sum_k c_k q_k
  + q^2 \sum_k \frac{c_k^2}{2\omega_k^2},
\end{equation}

The Caldeira--Leggett bath is characterized by its spectral density $J(\omega) = \pi\sum_k (c_k^2/2\omega_k)\,\delta(\omega-\omega_k)$, taken as Ohmic with hard cutoff $\Omega_c$ for the cathode phonon bath.
The final term is the counter-term canceling the bath-induced static renormalization.

We find that at cathode coupling, $\omega_{\mathrm{sys}}/\Omega_c \approx 0.54$, lying outside the Born--Markov small-parameter expansion:

\begin{equation}\label{eq:born_markov_bounds}
\varepsilon_{\mathrm{Born}} \le
  \bigl(\lambda/\omega_{\mathrm{sys}}\bigr)^4\,\Omega_c T_{\mathrm{evol}},
\qquad
\varepsilon_{\mathrm{Markov}} \le
  \bigl(\omega_{\mathrm{sys}}/\Omega_c\bigr)^2 \approx 0.29,
\end{equation}

roughly three times the conventional breakdown threshold~\cite{breuer2002theory} and difficult to control within a Lindblad description.
We therefore state the computational problem: given the DFT-parameterized Hamiltonian (Eq.~\ref{eq:master_H}) coupled to three reservoirs of the form (Eq.~\ref{eq:H_CL_main}) at $T = 300$~K, construct a quantum algorithm to produce the open-system observables above with controlled error without the Born--Markov approximation.

\section{REQWIEM Algorithm}\label{sec:REQWIEM}

We give an extension to the algorithm given in Courtney~\cite{courtney2026oracle}, being the Real-time Evolution of Quantum Wavepackets In Explicit Modularity (REQWIEM) algorithm, adapted for multi-polaron simulation.
We specify a material-agnostic form, including register architecture, chain-mapping reservoirs, Trotter step primitives, locality collapse rendering the on-diagonal potential polynomial in system size, reversible-arithmetic treatment of non-polynomial diagonals, the Jordan--Wigner hopping layer, and a gate-level verification protocol.
The unified application of closed-system multi-channel primitives (uniformly-controlled-rotation synthesis~\cite{mottonen2004quantum} and XOR-class fragmentation~\cite{motlagh2025quantum}) for systems exhibiting duality between their finite basis representation and discrete variable representation~\cite{light1985generalized} is given in Courtney~\cite{courtney2026oracle}, used here as black-box subroutines.
Everything specific to the open-system problem including chain registers, finite-temperature doubling, the counter-term, and locality collapse is constructed below and detailed in Appendix~\ref{app:hilbert_space}.

\subsection{Quantum Registers}\label{ssec:registers}

The full Hilbert space is a tensor product of a carrier-coordinate register, the polaron Jordan--Wigner register, per-site phonon registers, and three dark chain registers (one per reservoir; Figure~\ref{fig:tensor_product}):
\begin{equation}
\label{eq:tensor_product}
\HSpace_{\mathrm{tot}} =
    \HSpace_{X}^{(2^{n_x})}
    \otimes \bigotimes_{i=1}^{L} \HSpace_{\jw, i}^{(2)}
    \otimes \bigotimes_{i=1}^{L} \HSpace_{B, i}^{(2^{n_b})}
    \otimes \bigotimes_{r \in \{e, \mathrm{Li}, \mathrm{ph}\}}
        \bigotimes_{k=0}^{K_r - 1} \HSpace_{D, r, k}^{(2^{n_{d,r}})},
\end{equation}
realized with a fixed axis ordering $(X, \jw_{1..L}, B_{1..L}, D)$ so that the state tensor has shape $(2^{n_x},) + (2,)^L + (2^{n_b},)^L +(2^{n_d},)^{K}$. 
In the discrete-site carrier encoding (Section~\ref{sec:lmo_spec}) the continuous $X$ register is replaced by an $L_{\mathrm{Li}}$-qubit JW occupancy register, and Jahn--Teller-active sites carry an additional two-mode doublet register pair.

\begin{figure}[h]
\centering
\begin{tikzpicture}[
    >=Stealth, font=\small,
    reg/.style={rectangle, rounded corners=2pt, draw=blue!55!black, fill=blue!8, align=center, minimum height=8mm, minimum width=1.5cm},
    dark/.style={rectangle, rounded corners=2pt, draw=black!55, fill=black!5, align=center, minimum height=8mm, minimum width=1.7cm},
  ]
  \node[reg] (X)  at (0,2.0)   {$X$\\\scriptsize $n_x{=}12$};
  \node[reg] (JW) at (1.8,2.0) {$\jw$\\\scriptsize $L$};
  \node[reg] (B1) at (3.6,2.0) {$B_1$\\\scriptsize $n_b{=}9$};
  \node at (4.8,2.0) {$\cdots$};
  \node[reg] (BL) at (6.0,2.0) {$B_L$\\\scriptsize $n_b{=}9$};
  \begin{scope}[on background layer]
    \node[draw=blue!40!black, dashed, rounded corners=4pt, fit=(X)(JW)(B1)(BL),
          inner sep=6pt, fill=blue!3, label=above:{\textbf{Closed-system register} $(X,\jw,B)$}] {};
  \end{scope}
  \node[dark] (De)  at (0.6,-1.0) {$D_{e,k}$\\\scriptsize $K\cdot n_{d}$};
  \node[dark] (DLi) at (3.0,-1.0) {$D_{\mathrm{Li},k}$\\\scriptsize $K\cdot n_{d}$};
  \node[dark] (Dph) at (5.4,-1.0) {$D_{\mathrm{ph},k}$\\\scriptsize $K\cdot n_{d}$};
  \begin{scope}[on background layer]
    \node[draw=black!50, dashed, rounded corners=4pt, fit=(De)(DLi)(Dph),
          inner sep=6pt, fill=black!3, label=below:{\textbf{Dark chain registers} $(D_e,D_{\mathrm{Li}},D_{\mathrm{ph}})$}] {};
  \end{scope}
  \draw[<->, green!45!black, thick] (JW.south) to[bend right=10] (De.north);
  \node[green!45!black, font=\scriptsize, anchor=east] at (1.05,0.45) {boundary $e^-$};
  \draw[<->, orange!80!black, thick] (X.south) to[bend left=8] (DLi.north west);
  \node[orange!70!black, font=\scriptsize, anchor=west] at (1.6,0.55) {boundary Li$^+$};
  \draw[<->, red!60!black, thick] (B1.south) to[bend left=8] (Dph.north);
  \node[red!60!black, font=\scriptsize, anchor=west] at (4.55,0.45) {$H_{SB}^{\mathrm{ph}}$};
  \draw[<->, red!60!black, thick, dashed] (BL.south) to[bend left=14] (Dph.north east);
\end{tikzpicture}
\caption{The tensor-product Hilbert space of the system-plus-three-reservoirs
framework (LFP production geometry shown). The closed-system register
$(X, \jw, B)$ couples to the three dark chain registers
$(D_e, D_{\mathrm{Li}}, D_{\mathrm{ph}})$ through three bilinear cross
gates, one per reservoir.}
\label{fig:tensor_product}
\end{figure}
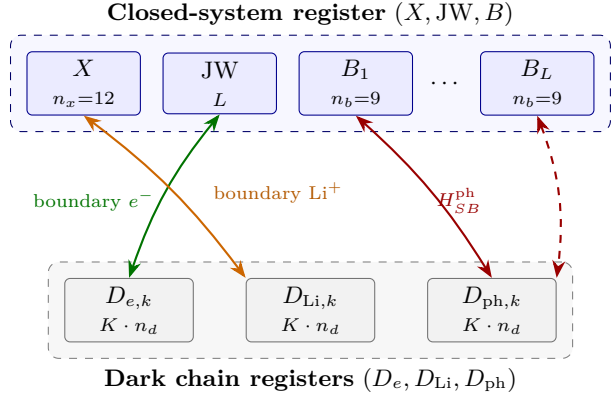

Every bosonic register is a DVR grid on $N = 2^n$ points (the fermionic electron-chain modes are exactly two-level), with each width satisfying progressively tight bounds, being the Fock-tail truncation at $\varepsilon_{\mathrm{tail}} = 10^{-8}$, the discretization bound~\cite{macridin2018digital}, and the Nyquist coherent-propagation bound. 
The dual-space ($c^2$-scaled) Nyquist bound,
\begin{equation}\label{eq:nyquist_HO}
n_{\min}^{\mathrm{Nyq, HO}} =
\Bigl\lceil \log_2\Bigl(\frac{4c^2 N_{\mathrm{ph}}}{\pi}\Bigr)\Bigr\rceil
\approx \log_2(N_{\mathrm{ph}}) + \log_2(4c^2/\pi),
\end{equation}
is typically binding ($\sim \log_2 N_{\mathrm{ph}} + 4.35$ at $c = 4$ containment)~\cite{ollitrault2020nonadiabatic}. 
The Nyquist bound guarantees coherent propagation of the wavefunction and its dual on the qubit grid. 
If the bulk of either leaks past the periodic boundary, the wavefunction wraps to the other side of the simulation box. 
Dynamic observables rely on the coherent propagation of the wavefunction and cannot be extracted in this event, confirmed in first-quantized quantum circuit compositions with explicit convergence analyses~\cite{ollitrault2020nonadiabatic, courtney2026oracle}.
Per-register widths and the binding bound for the LFP production geometry are collected in Section~\ref{sec:lfp_spec}. 
The total for that geometry is
\begin{equation}
\label{eq:total_qubits}
n_{\mathrm{tot}} = n_x + L + L n_b + \textstyle\sum_r K_r n_{d,r}
= 12 + 8 + 72 + 10 + 100 + 90 = 292~\text{system qubits},
\end{equation}
to which the reversible-arithmetic workspace of Section~\ref{ssec:soft_coulomb} adds $\sim84$ ancilla qubits.
Each reservoir term $K_r n_{d,r}$ counts the as-constructed finite-temperature chain: the fermionic lead's thermofield-doubled filled/empty branches and the bosonic baths' negative-frequency (T-TEDOPA) extension are realized within these registers rather than in addition to them, so the doubling contributes per-step gates but no further qubits and $n_{\mathrm{tot}} = 292$ is the doubled-construction total.
Each width scales as Equation~\ref{eq:nyquist_HO} in $E_{\mathrm{char}}/\omega$, so the budget re-evaluates for any material without re-deriving gate counts.

\subsection{Chain-Mapped Reservoirs}\label{ssec:chain_construction}

The algorithmic move rendering the open system Trotterizable is mapping each continuum bath (Eq.~\ref{eq:H_CL_main}) to a one-dimensional chain of $K$ nearest-neighbor-coupled bosonic modes by the Stieltjes recurrence~\cite{chin2010exact} (orthogonal polynomials with respect to $J(\omega)/\pi$). 
The system--bath coupling reduces to a single bond from the system to chain site $0$ with strength $c_0 = \sqrt{\mu_0}$, $\mu_0 = \int J(\omega)/\pi\, d\omega$:
\begin{equation}\label{eq:H_chain}
H_{\mathrm{chain}} = c_0\, q\, q_0^{\mathrm{chain}}
 + \sum_{k=0}^{K-1} e_k a_k^\dagger a_k
 + \sum_{k=0}^{K-2} t_k\bigl(a_k^\dagger a_{k+1} + \mathrm{h.c.}\bigr)
 + H_{\mathrm{CT}}.
\end{equation}
The chain recovers bath moments
\begin{equation}\label{eq:moments_match}
c_0^2\,[T^k]_{00} = \mu_k \equiv
  \int_0^{\Omega_c} \omega^k J(\omega)/\pi\, d\omega,
\quad k = 0, 1, \dots, 2K-1,
\end{equation}
with $T$ the tridiagonal Jacobi matrix (verified to $1.42\times10^{-14}$ at $K = 8$ for the Ohmic cathode bath). 
The Lanczos truncation property (lower coefficients are $K$-independent) makes $K$-convergence testable by chain extension. 
The spectral-moment match is complemented in the time domain, wherein the finite chain reproduces the target Ohmic bath correlation $C(t)$ to $<10^{-12}$ (a.u.) within the usable window. 
The production density is validated specific to various densities independent of the Drude-Lorentz kernels used for HEOM cross-checks (Section~\ref{ssec:lfp_validation}), chosen for their compact Pad\'e decomposition. 
The production density is Ohmic with a benign hard cutoff, since chain coefficients approach the analytically-known band asymptotics $e_k \to \Omega_c/2$, $t_k \to \Omega_c/4$. 
Its only artifact is the sharp recurrence reflection at $\Trec = \pi K/\Omega_c \approx 172$~fs ($K = 10$). 
Fast dynamics retain $>17\times$ margin (Section~\ref{ssec:lfp_validation}).
The discrete Holstein/Jahn--Teller mode carries the optical Franck--Condon coupling, where the Ohmic chain represents the residual low-frequency continuum. 
Reservoirs are mapped to independent chains, coupling to distinct system operators (charge, Li occupation, lattice displacement).
We carry the Caldeira--Leggett counter-term explicitly, written as $H_{\mathrm{CT}} = \alpha\,(c_0^2/2e_0)\, q^2$.
Its coefficient $\alpha = e_0\,[V_{cc}^{-1}]_{00}$ is fixed by cancellation between the chain's static self-energy $c_0^2[V_{cc}^{-1}]_{00}$ on the system mode ($V_{cc}$ the chain potential block), making the polaron binding $E_p = g^2/\omega$ bath-independent (Appendix~\ref{app:counter_term}).

For finite temperature, we incorporate the construction of Tamascelli et al.~\cite{tamascelli2019efficient, lorenzoni2024systematic,de2015thermofield, mascherpa2020optimized}, where the thermal bath with spectral function $J_T(\omega) = J(|\omega|)[\theta(\omega)(1 + N(|\omega|)) + \theta(-\omega)N(|\omega|)]$ is operationally equivalent to a vacuum-state $T = 0$ bath over $\pm\omega$.
The negative-frequency branch is encoded as an auxiliary chain, so state preparation costs zero gates. 
This split has the fluctuation--dissipation ratio built-in, where the bath power spectrum satisfies $S(-\omega)/S(\omega) = N/(1+N) = e^{-\beta\omega}$, so reduced dynamics obey detailed balance.
The three reservoirs each contribute one chain register: the bosonic phonon and Li baths via the Stieltjes construction with T-TEDOPA~\cite{chin2010exact, tamascelli2019efficient}, and the fermionic electron lead by block-Lanczos tridiagonalization of the lead Green's function~\cite{de2015discretize}, validated against the exact Landauer transmission on a semicircular-density-of-states (DOS) problem (spectral moments $m_0\ldots m_7$ to $8.5\times10^{-18}$; geometric self-energy convergence at log-slope $-0.996$). 
Block-Lanczos couplings are real-valued, making the lead Hamiltonian Hermitian and the reconstructed spectral function and fermionic occupations positive.
At finite temperature and chemical potential the lead is constructed by fermionic thermofield doubling~\cite{de2015thermofield, de2015discretize}: hybridization is split by the Fermi weight into a filled chain ($\Gamma f$, below $\mu_e$) and an empty chain ($\Gamma[1-f]$, above $\mu_e$), two independent $T=0$ chains whose vacua encode the grand-canonical state with zero preparation gates. The split conserves retarded hybridization and reproduces the lead occupation and fermionic detailed balance to machine precision, recovering the $T=0$ Landauer chain as the bands separate (the full construction is in Appendix~\ref{app:fermionic_lead}).
Emission of one reservoir layer per Trotter half-step is summarized in Algorithm~\ref{alg:chain_layer}.

\begin{algorithm}
\caption{Chain-reservoir layer (one half-step, one reservoir $r$)}
\label{alg:chain_layer}
\begin{algorithmic}
\REQUIRE chain coefficients $\{e_k, t_k\}_{k<K_r}$ from the Stieltjes recurrence on $J_r(\omega)$; coupling $c_0$; counter-term weight $\alpha$
\STATE $V_{SB}(\tau/2)$: bilinear phase $\exp(-i c_0\, q\, q_0\,\tau/2)$ between the system coupling coordinate and chain site $0$ (CPhase grid)
\FOR{$k = 0$ to $K_r - 2$}
    \STATE Chain hop $t_k$: inner symmetric split $q_k q_{k+1}(\tau/4) \to \mathrm{QFT~pair} \to p_k p_{k+1}(\tau/2) \to \mathrm{iQFT~pair} \to q_k q_{k+1}(\tau/4)$
\ENDFOR
\STATE On-site $e_k$ and counter-term $\alpha c_0^2/(2e_0)\,q^2$: diagonal phases (folded into $V_{\mathrm{diag}}$)
\STATE Chain kinetic: cQFT $\to$ diagonal $R_z(p^2)$ $\to$ cIQFT per chain mode
\end{algorithmic}
\end{algorithm}

Dissipation emerges from the $V_{SB} \to T_{\mathrm{chain}} \to V_{SB}$ sandwich acting on the chain
registers, read out by partial trace instead of methods from Lindblad-dilation routes~\cite{khatri2020principles, cleve2016efficient, childs2016efficient, ding2024simulating}, which instead consider an alternative weak-coupling form (Eq.~\ref{eq:born_markov_bounds}).
The only bath approximation is the finite chain length $K$, controlled by the recurrence-time bound $\Trec = \pi K/\Omega_c$ and convergent to the numerically exact reference at sufficient depth~\cite{tamascelli2018nonperturbative}.

\subsection{The Trotter Step}\label{ssec:trotter_step}

We compose each step with seven gate primitives (Table~\ref{tab:primitives}). 
Multi-controlled rotations whose angle depends on a control register are realized as M\"ott\"onen uniformly controlled rotations (UCR$_z^{(m)}$), applied by a Walsh-Hadamard transform acting on multi-controlled phase operations, decomposing circuits into one- and two-qubit
gates ($H$, $X$, $R_z$, $R_y$, CNOT, CPhase) with Toffoli gates entering only through the arithmetic block of Section~\ref{ssec:soft_coulomb}.

\begin{table}[h]
\centering
\caption{The primitives composing every Trotter step.
The first six are circuit gate primitives; \texttt{partial\_trace\_dark} (final row) is an operation of the classical reference simulator only. 
On hardware, readout is performed by the measurement protocol of Section~\ref{ssec:measurement_protocols}.
No partial trace is ever executed as a gate.}
\label{tab:primitives}
\small
\begin{tabular}{p{4.6cm} p{5.0cm} p{5.0cm}}
\toprule
Primitive & Unitary realized & Gate-count scaling \\
\midrule
\texttt{apply\_diag\_phase} &
$U = \exp(i\boldsymbol{\theta})$ elementwise on the computational basis
($V_{\mathrm{diag}}$ layer) &
$\mathcal{O}(Nn)$ via UCR$_z$ Gray-code over $n$ qubits \\
\midrule
\texttt{apply\_axis\_kinetic} &
$U = \icQFT\cdot\mathrm{diag}(e^{-ik^2\tau/2m_{\mathrm{eff}}})\cdot\cQFT$
on one axis &
$\cQFT$: $n$ CNOT $+\,n(n{-}1)/2$ CPhase; diagonal block
$\mathcal{O}(n^2)$ binary-expansion gates \\
\midrule
\texttt{apply\_two\_qubit\_XY} &
$U = \exp[-i\alpha(X_aX_b + Y_aY_b)/2]$ (JW hop) &
3 CNOT $+$ 2 $R_z$ per bond~\cite{jiang2018quantum} \\
\midrule
\texttt{apply\_bilinear\_cross} &
$U = \exp(i\boldsymbol{\theta}_{ab})$, precomputed 2D phase grid on axes
$a, b$ (system--bath coupling) &
$\mathcal{O}(N_aN_b)$ controlled-$P$ gates \\
\midrule
\texttt{apply\_chain\_hop\_bosonic} &
$U = \exp[-i(t_k/2)(q_kq_{k+1} + p_kp_{k+1})\tau]$ &
2 bilinear-cross $+$ 4 QFT/iQFT pairs \\
\midrule
\texttt{apply\_cphase} &
$U = e^{-iU_{ij}\tau}\,|11\rangle\langle 11|_{i,j}$ (per-pair Coulomb) &
1 CPhase per pair; $\binom{L}{2}$ pairs \\
\midrule
\texttt{partial\_trace\_dark} &
$\rho_S = \Tr_D |\psi\rangle\langle\psi|$ (simulator-only readout) &
$\mathcal{O}(d_S^2 d_D)$ classical operations; zero circuit cost \\
\bottomrule
\end{tabular}
\end{table}

For concision, we show the second-order step~\cite{suzuki1991general}
\begin{equation}
\label{eq:trotter_step}
\boxed{
\begin{aligned}
\mathcal{U}_{\mathrm{Strang}}(\Delta t)
    &= T_X(\tfrac{\Delta t}{2})\, V_{\mathrm{diag}}(\tfrac{\Delta t}{2})\,
       V_{SB}(\tfrac{\Delta t}{2})\,
       T_{\mathrm{chainhops}}(\tfrac{\Delta t}{2})\,
       T_{\mathrm{ph}}(\tfrac{\Delta t}{2})\,
       T_{\mathrm{chain}}(\tfrac{\Delta t}{2}) \\
    &\quad\cdot T_{\mathrm{JWhop}}(\Delta t) \\
    &\quad\cdot T_{\mathrm{chain}}(\tfrac{\Delta t}{2})\,
       T_{\mathrm{ph}}(\tfrac{\Delta t}{2})\,
       T_{\mathrm{chainhops}}(\tfrac{\Delta t}{2})\,
       V_{SB}(\tfrac{\Delta t}{2})\,
       V_{\mathrm{diag}}(\tfrac{\Delta t}{2})\, T_X(\tfrac{\Delta t}{2}),
\end{aligned}}
\end{equation}
where $T_X$ is the carrier-coordinate kinetic, $V_{\mathrm{diag}}$ bundles every diagonal term (washboard, on-site $\varepsilon_i n_i$, phonon and chain on-site potentials with the counter-term, Holstein
$g n_i q_i$, Coulomb $U_{ij}n_in_j$, and the Li--polaron attraction), $V_{SB}$ and $T_{\mathrm{chainhops}}$/$T_{\mathrm{chain}}$ are reservoir layers (Algorithm~\ref{alg:chain_layer}), and $T_{\mathrm{JWhop}}$ is the color-scheduled carrier hopping (Section~\ref{ssec:jw_coloring}). 
Local error is $\mathcal{O}(\Delta t^3)$ per second-order step,
\begin{equation}\label{eq:trotter_error}
\mathcal{U}_{\mathrm{Strang}}(\Delta t) - e^{-iH\Delta t}
 = -\tfrac{(\Delta t)^3}{24}\bigl(2[H_o,[H_o,H_e]] +
   [H_e,[H_o,H_e]]\bigr) + \mathcal{O}(\Delta t^5),
\end{equation}
with dominant partners $[T_{\mathrm{ph}}, V_{\mathrm{diag}}]$ and
$[T_{\mathrm{chain}}, V_{SB}]$. 
The latter is bounded by $32\lambda^2\omega_{\mathrm{ph}}\Delta t^3$, giving $\sim10^{-15}$ per
step at the cathode operating point.
Halving $\Delta t$ empirically reduces trajectory error by a factor in $[2, 6]$, consistent with second order.

\paragraph{Integrator order and the qubitization alternative.}
We identify a lower total cost with fourth-order product formulas, used in our resource estimations. 
The second-order formula here is given for visual clarity.
Explicit forms are given in Appendix~\ref{sm:qsvt}.
At working precision, fourth-order decomposition costs $1.7\times10^8$ T-gates per trajectory versus $2.3\times10^8$ for second-order, the gap growing at tighter targets (Table~\ref{tab:integrator_comparison}).
The second-order point is retained throughout as a comparison.
Every trajectory-derived number below is quoted as production (fourth-order) with the validated second-order figure recoverable by tabulated ratios. 

A qubitization/quantum singular value transformation (QSVT) alternative~\cite{low2017optimal, low2019hamiltonian, gilyen2019quantum, martyn2021grand, courtney2026oracle} is also natural here.
The chain-mapped Hamiltonian is $d$-sparse with an $\mathcal{O}(L^2 + K)$-term linear-combination-of-unitaries (LCU) block encoding.
Its query cost $\sim\alpha T + \mathcal{O}(\log 1/\epsilon)$ is additive in the precision but multiplicative in the LCU one-norm. 
At the production operating point the bosonic amplitude factors inflate the one-norm to $\alpha = 2.74$~Ha ({Appendix~\ref{sm:qsvt}}), flooring the query count at $\alpha T \approx 1.9\times10^4$. 
Each query requires one \textsc{select}$+$\textsc{prepare} pass over the LCU, i.e.\ the coefficient quantum read-only memory (QROM) over the $\mathcal{O}(L^2 + K)$ terms plus one run of the same diagonal and bilinear emitters. 
One Trotter step with $3.0\times10^5$ T-gates floors the QSVT trajectory at $\sim5.8\times10^9$ T-gates: the crossover at which qubitization overtakes the fourth-order product formula sits at $\epsilon^* \approx 4\times10^{-9}$, below electrochemical requirement but nevertheless useful for method contextualization.
This consideration is necessary of other systems which have weaker coupling, more structured state couplings, or that may require greater precision. 

\begin{table}[h]
\centering
\caption{Per-trajectory T-gates versus target propagator error at the LFP production point ($T_{\mathrm{evol}} = 172$~fs; commutator constants calibrated to the per-step bounds of this section. 
Derivations and LCU one-norm assembly are in Appendix~\ref{sm:qsvt}.
The QSVT column is queries $\times$ per-query cost: $(\alpha T + 1.4\ln 1/\epsilon)$
\textsc{select}$+$\textsc{prepare} passes, each charged at one Trotter step ($3.0\times10^5$ T-gates). 
The QSVT column is constant to two significant figures because the $\alpha T\approx1.9\times10^4$-query floor dominates the additive $\ln(1/\varepsilon)$ term ($\le0.15$) across this range.
Best per row in bold.}
\label{tab:integrator_comparison}
\small\renewcommand{\arraystretch}{1.15}
\begin{tabular}{r r r r}
\toprule
$\epsilon$ & Strang-2 & Suzuki-4 & QSVT ($\alpha = 2.74$~Ha) \\
\midrule
$5\times10^{-3}$ & $2.3\times10^{8}$ & $\mathbf{1.7\times10^{8}}$ & $5.8\times10^{9}$ \\
$10^{-4}$ & $1.6\times10^{9}$ & $\mathbf{4.6\times10^{8}}$ & $5.8\times10^{9}$ \\
$10^{-6}$ & $1.6\times10^{10}$ & $\mathbf{1.4\times10^{9}}$ & $5.8\times10^{9}$ \\
$10^{-8}$ & $1.6\times10^{11}$ & $\mathbf{4.6\times10^{9}}$ & $5.8\times10^{9}$ \\
$10^{-9}$ & $5.1\times10^{11}$ & $8.1\times10^{9}$ & $\mathbf{5.8\times10^{9}}$ \\
\bottomrule
\end{tabular}
\end{table}

\begin{figure}[h]
\centering

\begin{tikzpicture}[
    >=Stealth, font=\small,
    sys/.style={rectangle, rounded corners=2pt, draw=blue!55!black, fill=blue!8, align=center, minimum height=8mm, minimum width=2.0cm},
    bath/.style={rectangle, rounded corners=2pt, draw=black!55, fill=black!5, align=center, minimum height=8mm, minimum width=2.0cm},
  ]
  \node[sys] (xgrid) at (0,1.1) {Li grid\\($n_x$)};
  \node[sys] (jw)    at (0,0)   {JW polaron\\($L$)};
  \node[sys] (ph)    at (0,-1.1){intra-cell\\phonons ($n_b$)};
  \node[draw=blue!40!black, dashed, rounded corners=4pt, fit=(xgrid)(jw)(ph),
        inner sep=5pt, label=above:{\textbf{System register}}] (sysbox) {};
  \node[bath] (be)  at (6,1.1)  {$e^-$ chain\\($K,n_{d}$)};
  \node[bath] (bli) at (6,0)    {Li$^+$ chain\\($K,n_{d}$)};
  \node[bath] (bph) at (6,-1.1) {phonon chain\\($K,n_{d}$)};
  \node[draw=black!50, dashed, rounded corners=4pt, fit=(be)(bli)(bph),
        inner sep=5pt, label=above:{\textbf{Bath registers (+ counter-terms)}}] (bathbox) {};
  \draw[<->, green!45!black, thick] (jw.east)  -- (be.west);
  \node[green!45!black, font=\scriptsize, fill=white, inner sep=1pt] at (4.4,0.93) {$H_{SB}^{e}$};
  \draw[<->, orange!80!black, thick] (xgrid.east) -- (bli.west);
  \node[orange!70!black, font=\scriptsize, fill=white, inner sep=1pt] at (1.6,0.93) {$H_{SB}^{\mathrm{Li}}$};
  \draw[<->, red!60!black, thick] (ph.east) -- (bph.west) node[midway,below,font=\scriptsize]{$H_{SB}^{\mathrm{ph}}$};
\end{tikzpicture}

\vspace{1.2em}

\begin{tikzpicture}[font=\scriptsize, >=Stealth]
  \definecolor{Tcol}{RGB}{60,110,180}
  \definecolor{Vcol}{RGB}{200,140,40}
  \definecolor{SBcol}{RGB}{170,60,60}
  \def\lw{0.95}
  \def\hw{0.55}
  \foreach \i/\lbl/\col/\w in {
      0/$\tfrac12 T$/Tcol/\hw,
      1/$\tfrac12 V$/Vcol/\hw,
      2/$T$/Tcol/\lw,
      3/$V$/Vcol/\lw,
      4/$\tfrac12 T$/Tcol/\hw,
      5/$V_{SB}$/SBcol/\lw,
      6/$\tfrac12 T$/Tcol/\hw,
      7/$V$/Vcol/\lw,
      8/$T$/Tcol/\lw,
      9/$\tfrac12 V$/Vcol/\hw,
      10/$\tfrac12 T$/Tcol/\hw } {
    \pgfmathsetmacro{\xpos}{\i*1.02}
    \node[rectangle, draw=black!40, fill=\col!30, minimum height=8mm,
          text width={\w cm}, align=center, inner sep=1pt] (L\i) at (\xpos,0) {\lbl};
  }
  \draw[SBcol, very thick, dashed] (L5.north) ++(0,0) -- ++(0,-0.9);
  \node[SBcol, anchor=south, font=\scriptsize] at (L5.north) {symmetric mirror};
  \draw[->, black!60] (L0.south west)+(0,-0.35) -- node[below]{Trotter time $\Delta t$} ($(L10.south east)+(0,-0.35)$);
  \node[anchor=west, font=\scriptsize] at (0,-1.55)
    {\textcolor{Tcol}{$\blacksquare$} kinetic $T$\quad
     \textcolor{Vcol}{$\blacksquare$} potential $V$\quad
     \textcolor{SBcol}{$\blacksquare$} system--bath $V_{SB}$};
\end{tikzpicture}
\caption{Architecture and Trotter-step timeline. \emph{Upper:} the system register (Li coordinate, Jordan--Wigner polaron register, intra-cell phonons) couples through three bilinear interactions $H_{SB}$ to the chain-mapped bath registers, each carrying its Caldeira--Leggett counter-term. \emph{Lower:} simplified second-order Strang step (Eq.~\ref{eq:trotter_step}) as a banded timeline of kinetic ($T$), potential ($V$), and system--bath ($V_{SB}$) layers, symmetric about the central mirror; half-steps drawn narrower.}
\label{fig:arch_strang}
\end{figure}
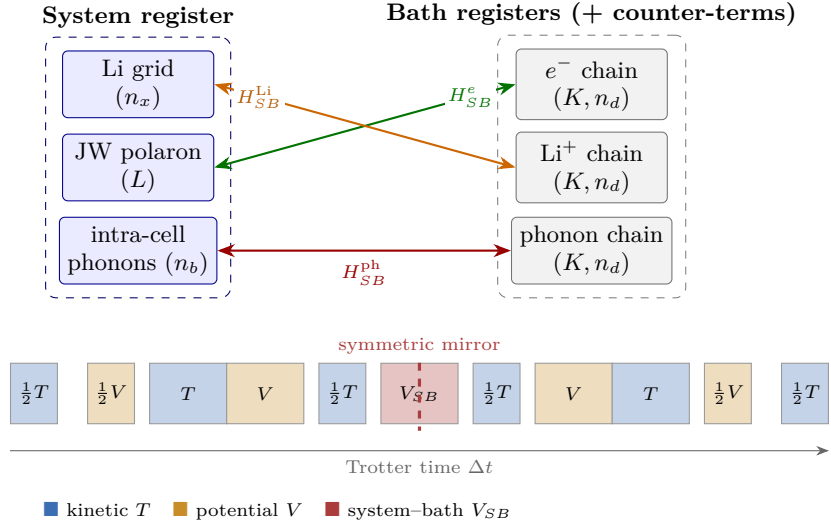

\subsection{On-Diagonal Potential and Locality Collapse}
\label{ssec:locality_collapse}

A generic multi-channel diabatic synthesis multiplexes one potential surface per electronic configuration: a uniformly controlled rotation over all $m = L$ channel qubits, inserting a factor $M = 2^L$ in the on-diagonal gate count.
That is an informed cost in the vibronic setting this machinery was built for~\cite{courtney2026oracle}, where the $M$ diabatic surfaces are distinct functions of the nuclear coordinates. 
For the cathode setting it becomes an artifact, now having a structural configuration dependence. 
Specifically, Equation~\ref{eq:master_H} is set to be at most $2$-local in the occupations $\{n_i\}$: the diagonal potential is a sum of (i) an occupation-independent reference surface, (ii) occupation-linear terms (Holstein $g\sqrt2\,q_in_i$ and the Li--polaron coupling), and (iii) occupation-bilinear terms ($U_{ij}n_in_j$, fixed classical angles $U_{\mathrm{nn}}/|i{-}j|^\nu$ evaluated at compile time). 
A diagonal operator that is a sum of $\le 2$-local terms is synthesized as a fixed set of singly- and doubly-controlled phase blocks. 
This is one block per coupling term, $\mathcal{O}(L + L^2)$ blocks in all (Algorithm~\ref{alg:vdiag}).

\begin{algorithm}
\caption{$V_{\mathrm{diag}}$ emission under the locality collapse (one half-step)}
\label{alg:vdiag}
\begin{algorithmic}
\REQUIRE $\le2$-local diagonal $V = V_{\mathrm{ref}} + \sum_{i} f_i(\mathbf{q})\,n_i + \sum_{i<j} U_{ij}\, n_i n_j$; register widths $n$
\STATE Reference surface $V_{\mathrm{ref}}$ (harmonic, chain on-site, counter-term): unconditional diagonal phases, $\mathcal{O}(d\,n^2)$ rotations via binary expansion
\STATE Washboard $V_{\mathrm{wb}}(R)$: exact periodic-support Gray-code diagonal on the low $\log_2(a_{\mathrm{lat}}/\Delta x)$ bits (Section~\ref{ssec:soft_coulomb})
\FOR{$i = 1$ to $L$}
    \STATE Occupation-linear $f_i(\mathbf{q})\,n_i$ (Holstein; JT when present): singly-controlled phase, linear-in-$q$ block of $1{+}n$ rotations, control $\jw_i$
\ENDFOR
\FOR{all pairs $i < j$ with $|U_{ij}| > 0$}
    \STATE Coulomb $U_{ij}\,n_in_j$: one doubly-controlled phase
\ENDFOR
\STATE Soft-Coulomb $V_{\mathrm{LiP}}(R - x_i)\,n_i$: Algorithm~\ref{alg:soft_coulomb} (the one non-polynomial diagonal)
\end{algorithmic}
\end{algorithm}

\noindent That a local Hamiltonian admits a polynomial product-formula
step is a standard efficiency property of Trotterization~\cite{lloyd1996universal, childs2021theory}. 
We record it in this form not as a novelty claim but to fix the corrected counting that replaces the inherited $2^L$ multiplexing, parameter-resolved for cathode Hamiltonians.

\begin{theorem}\label{thm:poly_L}
[Locality collapse of the diagonal multiplexing]
For the chain-mapped Caldeira--Leggett construction with a system Hamiltonian at most $2$-local in the Jordan--Wigner occupations $\{n_i\}$, under Nyquist-binding register sizing $n \ge \lceil\log_2(4c^2E_{\mathrm{char}}/\pi\omega)\rceil$, the per-Trotter-step T-count obeys
\begin{equation}\label{eq:poly_L_scaling}
T_{\mathrm{step}} = \mathcal{O}\Bigl(
  T_{\mathrm{RS}}(\varepsilon_{\mathrm{rot}})\,(d + L^2)\,
  \log_2^2\tfrac{E_{\mathrm{char}}}{\omega}\Bigr),
\end{equation}
polynomial in the lattice size $L$ and bosonic mode count $d$ and logarithmically mild in $E_{\mathrm{char}}/\omega$, with no $2^L$ factor. 
The screened Coulomb pair count $\mathcal{O}(L^2)$ tightens to $\mathcal{O}(L\,z(r_c))$ under truncation at a screening radius $r_c$.
\end{theorem}

\noindent The proof reclassifies coupling terms by JW locality ($|\mathcal{G}_1| = \mathcal{O}(L)$ occupation-linear, $|\mathcal{G}_2^{<}| = \mathcal{O}(L^2)$ occupation-bilinear, plus occupation-independent coordinate cross terms) and synthesizes each block once.
The parameter-resolved statement and per-component counts are given with the material specifications (Sections~\ref{sec:lfp_spec},\ref{sec:lmo_spec}).
Two qualifications keep the result in proportion. 
The dense-emitter budget at $L = 8$ would carry $M = 256$ multiplexed surfaces, giving a $\sim400\times$ inflation of the on-diagonal block, but is dominated by rotation synthesis of the reference and bilinear surfaces ($\sim2.0\times10^5$ of the $3.0\times10^5$
T-gates per substep) and by the reversible soft-Coulomb arithmetic of Section~\ref{ssec:soft_coulomb} ($7.4\times10^4$). 
The collapsed Coulomb block costs $28$ rotations and $56$ Toffoli.
Integrator order (Section~\ref{ssec:trotter_step}) moves the trajectory-level estimate.
The full per-observable cost multiplies by anchor and estimation factors, $T^{\mathrm{obs}} = N_{\mathrm{anchor}}\,\max(n_{\mathrm{sh}},\, K_{\mathrm{QAE}})\, N_{\mathrm{step}} T_{\mathrm{step}}$: incoherent sampling ($n_{\mathrm{sh}}$ repetitions) and amplitude estimation ($K_{\mathrm{QAE}}$ coherent repetitions) are alternative estimation strategies per anchor (Section~\ref{sec:resources}).

\subsection{Non-Polynomial Diagonals: Soft Coulomb and Washboard}\label{ssec:soft_coulomb}

Diagonal term polynomials in their register coordinates are emitted at $\mathcal{O}(n)$--$\mathcal{O}(n^2)$ rotations by binary expansion, since we are excluding anharmonics in the present construction.
The screened Li--polaron attraction $V_{\mathrm{LiP}}(R) = -[\varepsilon_\infty\sqrt{(R - x_i)^2 + \rho_c^2}\,]^{-1}$ is the exception, being transcendental in the continuous register $R$, and a generic diagonal over $n_x$ qubits has up to $2^{n_x}$ multilinear coefficients. 
We therefore carry it by reversible arithmetic~\cite{haner2018optimizing, munoz2018t} and build up in stages until the phase is linear in an ancilla register, then rotate, then uncompute (Algorithm~\ref{alg:soft_coulomb}).

\begin{algorithm}
\caption{Soft-Coulomb evaluation via reversible $1/\sqrt{\cdot}$ arithmetic}
\label{alg:soft_coulomb}
\begin{algorithmic}
\REQUIRE coordinate register $R$ ($n_x$ qubits); site constant $x_i$; working width $w$; Newton depths $N_R, N_S$
\STATE $u \gets R - x_i$ (constant adder); $s \gets u^2 + \rho_c^2$ (squarer, $\sim w^2$ Toffoli)
\STATE $t \gets 1/s$: seeded table lookup $+$ $N_R$ Newton iterations ($\sim 2w^2$ Toffoli each)
\STATE $v \gets \sqrt{t} = 1/\sqrt{s}$: seeded lookup $+$ $N_S$ Newton iterations ($\sim 2w^2$ Toffoli each)
\STATE Phase: $w$ occupation-controlled $R_z$ on $v$, angle $\tau/\varepsilon_\infty$ per bit weight \COMMENT{argument now linear}
\STATE Uncompute $v, t, s, u$ (reverse arithmetic)
\end{algorithmic}
\end{algorithm}

One evaluation (\ref{fig:softcoulomb}) costs $\approx 2(1 + 2N_R + 2N_S)\,w^2 \approx 5.2\times10^3$ Toffoli at $w = 14$, $\sim w$ rotations, on a $\sim6w \approx 84$-qubit workspace reused serially.
In a single-carrier sector $n_i = 1$, so the occupied-site index $\mathrm{idx} = \sum_i i\,n_i$ is formed once, $x_{\mathrm{idx}}$ read by a small QROM, and a single evaluation suffices per half-step ($+7\times10^4$ T-gates per step); the naive $L$-term bound is $+6\times10^5$. 
In the discrete-site carrier encoding the term reduces to occupation-bilinear phases with classically precomputed angles and the block disappears (Section~\ref{sec:lmo_spec}).

\begin{figure}[t]
\centering
\resizebox{\textwidth}{!}{%
\begin{tikzpicture}[
  font=\scriptsize, >=stealth,
  op/.style={draw, rounded corners, align=center, text width=2.2cm,
             minimum height=1.0cm, inner sep=2pt},
  unc/.style={draw, dashed, rounded corners, align=center, text width=2.2cm,
              minimum height=1.0cm, inner sep=2pt}]
\node[op] (u) at (0,0) {$u\gets R-x_i$\\[1pt] adder};
\node[op] (s) at (3.0,0) {$s\gets u^2+\rho_c^2$\\[1pt] squarer};
\node[op] (t) at (6.0,0) {$t\gets1/s$\\[1pt] lookup $+$ $N_R$ Newton};
\node[op] (v) at (9.0,0) {$v\gets\sqrt{t}$\\[1pt] lookup $+$ $N_S$ Newton};
\node[op, fill=black!8] (p) at (12.0,0) {diagonal phase\\[1pt]
   $R_z(\tau/\varepsilon_\infty)$};
\node[unc] (un) at (15.2,0) {uncompute\\[1pt] $v,t,s,u$ (mirror)};
\foreach \a/\b in {u/s,s/t,t/v,v/p,p/un} \draw[->,very thick] (\a)--(\b);
\node[font=\tiny, text=black!55] at (3.0,-0.75) {$\sim w^2$};
\node[font=\tiny, text=black!55] at (6.0,-0.75) {$\sim 2N_R w^2$};
\node[font=\tiny, text=black!55] at (9.0,-0.75) {$\sim 2N_S w^2$};
\draw[rounded corners, fill=black!6, draw=black!30] (-1.2,-2.05) rectangle (16.4,-1.5);
\node[font=\scriptsize] at (7.6,-1.78)
   {\textbf{Dominant per-step cost:} $\approx 5.2\times10^3$ Toffoli at
    $w{=}14$ $\cdot$ ${\sim}84$-qubit workspace (reused serially) $\cdot$
    $+7\times10^4$ T-gates/step};
\end{tikzpicture}}
\caption{The reversible $1/\sqrt{\cdot}$ evaluation of the soft-Coulomb potential. The argument $u=R-x_i$ is squared, the reciprocal and inverse-square-root computed by seeded table lookup plus Newton iteration, the occupation-controlled diagonal phase applied, and the workspace uncomputed (dashed mirror) to clear it.}
\label{fig:softcoulomb}
\end{figure}
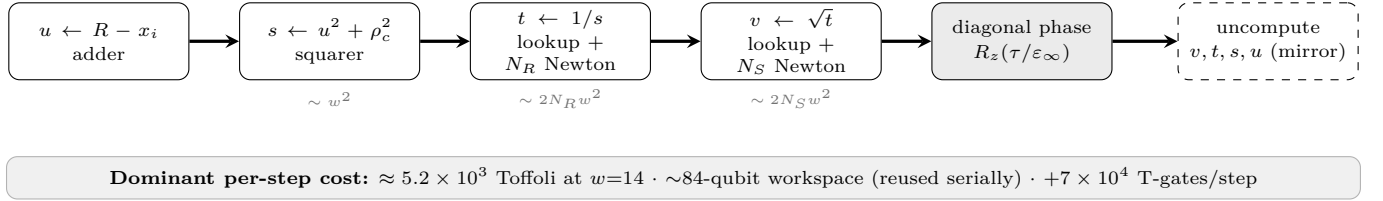

\paragraph{Washboard potential synthesis by lattice periodicity.}
The washboard $V_{\mathrm{wb}}(R) = V_b\sin^2(\pi R/a_{\mathrm{lat}}) = V_b/2 - (V_b/2)\cos(2\pi R/a_{\mathrm{lat}})$ is the second non-polynomial diagonal, being periodic with the lattice constant. 
For $L_c = N_{\mathrm{sites}}\,a_{\mathrm{lat}}$ with $N_{\mathrm{sites}} = 8$, the $2^{12}$-point grid holds $2^{12}/8 = 2^9$ points per lattice period. 
The phase $e^{-i\tau V_{\mathrm{wb}}(R)}$ then depends only on $R \bmod a_{\mathrm{lat}}$, i.e.\ on the lowest $p = n_x - \log_2 N_{\mathrm{sites}} = 9$ bits of the register, factorizing as identity on the top qubits times an exact $p$-qubit diagonal, synthesized by the Gray-code walk at $2^p = 512$ rotations $+\ 2^p - 1$ CNOTs per half-step with coefficients precomputed once by the fast Walsh-Hadamard transform. 
Per step this adds $1024$ rotations ($+2.6\times10^4$ T-gates at $T_{\mathrm{RS}} = 25$), itemized in Equation~\ref{eq:lfp_tstep}. 
When the compiler is depth-limited, the same term can instead be emitted through the arithmetic workspace (a CORDIC sine sharing the $1/\sqrt{\cdot}$ registers), inducing higher T-count but $\sim30\times$ lower serial depth, being a per-metric emission choice. 
In the discrete-site encoding the washboard, like the soft Coulomb, vanishes, where barrier physics is carried by the effective hopping amplitude (Section~\ref{ssec:encoding_tradeoffs}).

\subsection{Off-Diagonal Coupling: Jordan--Wigner Hopping}\label{ssec:jw_coloring}

Carrier hopping $-t_{\mathrm{hop}}\sum_{\langle i,j\rangle}(c_i^\dagger c_j + \mathrm{h.c.})$ maps under Jordan--Wigner to two-qubit XY bond operators $B_i = (X_iX_{i+1} + Y_iY_{i+1})/2$. 
Adjacent bonds fail to commute, $[B_i, B_{i+1}] = \tfrac{i}{2}(X_iZ_{i+1}Y_{i+2} - Y_iZ_{i+1}X_{i+2})$, while non-adjacent bonds commute. 
Bond sets partition into two color classes within which the layer factorizes~\cite{whitfield2011simulation, wecker2014gate,
jiang2018quantum}. 
An inner product-formula split $T_{\mathrm{JWhop}}(\Delta t) = e^{-i\Delta t H_E/2} e^{-i\Delta t H_O} e^{-i\Delta t H_E/2}$ carries local error
\begin{equation}\label{eq:jw_strang_bound}
\bigl\|T_{\mathrm{JWhop}}(\Delta t) - e^{-i\Delta t H_{\mathrm{hop}}}\bigr\|
\le \tfrac{3(L-2)\,t_{\mathrm{hop}}^3 (\Delta t)^3}{4},
\end{equation}
which at the production timestep evaluates to $\sim2.0\times10^{-6}$ per step ($L = 8$, $t_{\mathrm{hop}} = 25$~meV), several orders above empirically observed error.
Each bond costs 3 CNOT $+$ 2 $R_z$. 
Multi-channel off-diagonal machinery (XOR-class fragmentation with Clifford focusing) is inherited from the closed-system protocol~\cite{courtney2026oracle, motlagh2025quantum}.
In the cathode Hamiltonians the only off-diagonal terms are these hopping bonds, so the two-color schedule is the complete off-diagonal layer.

\subsection{Gate-Level Verification: Quality Control}\label{ssec:phase_qc}

Since the cathode problem requires large registers and is resistant to dimensional reduction, we use a quality control scheme to certify gate-level equivalence. 
We find
\begin{equation}
\label{eq:phase_qc_def}
\boxed{\;
\bigl\|\mathcal{U}_{\mathrm{circuit}}|\psi_0\rangle
 - \mathcal{U}_{\mathrm{classical}}|\psi_0\rangle\bigr\|_\infty
 < \varepsilon, \qquad \varepsilon = 10^{-12},
\;}
\end{equation}
where $\mathcal{U}_{\mathrm{classical}}$ applies a second-order Strang composition (Eq.~\ref{eq:trotter_step}) by direct array operations (FFTs, tensor contractions, elementwise phases) and $\mathcal{U}_{\mathrm{circuit}}$ is a statevector simulation of the emitted gate list, with no dense unitaries or generic diagonal-phase blocks permitted in the emission path.
Sub-convergent benches reach machine precision ($\|\Delta\psi\|_\infty \approx 10^{-17}$--$10^{-15}$), while the convergence gate (Trotter ratio in $[2,6]$ under timestep halving), and charge/norm conservation gates together constitute a standing verification harness applied to both material specifications (Sections~\ref{sec:lfp_spec},\ref{sec:lmo_spec}).

\subsection{Validation Scope and Extrapolation}\label{ssec:scope}

Because the production geometry ($292$-qubit register, $\sim2^{292}$ amplitudes) is far beyond statevector simulation, we delineate what is run on small, classically tractable instances and what is calculated for the production size explicitly before any quantitative claim.
The largest statevector instances actually simulated are of order $10$--$13$ qubit-equivalents (Hilbert dimension $\sim10^{3}$--$10^{4}$): up to $L = 4$ Holstein sites with phonon Fock registers in the Trotter-constant staircase, the two-site excitonic dimer with two independent baths in the HEOM cross-check (Section~\ref{ssec:lfp_validation}), and the single spin-boson chain. 
This reduction is for gate-level validation of operator primitives, being explicitly sub-convergent for the chemical systems being considered.
Production is calculated for $L = 8$ ($\LFP$) / multi-polaron $\LMO$ at $292$ system qubits and $\sim84$ arithmetic qubits.

\begin{table}[h]
\centering
\caption{List of what is run, costed, or extrapolated. The empirical observables are
computed on classically tractable tiers (RUN). Costing rows are marked
RESULT, with exact gate counts of the emitted production circuit, distinct from the EXTRAP.\ fault-tolerant stack and the NEITHER production
dynamical trajectory; per-curve T-totals additionally fold in the extrapolated
step count $N_{\mathrm{step}}$.}
\label{tab:scope}
\small\renewcommand{\arraystretch}{1.15}
\begin{tabular}{l l l}
\toprule
Quantity & Status & How obtained / size \\
\midrule
Quality control gate $\equiv$ classical diagonal ($6\times10^{-16}$) & RUN & statevector of emitted gates; sub-conv.\ bench \\
Strang ratio $3.99$/$4.00$; ED match ($6\times10^{-11}$) & RUN & timestep halving / dense ED; sub-conv. \\
Weak Lindblad ($\alpha\le0.1$); HEOM ($2.07\%$) & RUN & $1$ site $+$ bath, spin-boson \\
Multi-site HEOM dimer ($3.45\%$) & RUN & $2$ sites $+$ $2$ baths \\
$\LFP$ FWHM $12.6$~mV, peak $\pm5$~mV & RUN & equilibrium two-channel solver ($K_r\to\infty$ tier) \\
$\LMO$ two-plateau split ($125$~mV) & RUN & Ewald lattice-gas $+$ grand-canonical trace \\
Trotter constant $W_2(L)$, register $\langle n\rangle$ vs bound & RUN & growing-$L$ staircases (this section) \\
$\LFP$: $376$ logical qubits; $3.0\times10^5$ T/step $\to 8\times10^{11}$/curve & RESULT & emitted production circuit \\
$\LMO$: $368$ logical qubits; $3.0\times10^5$ T/step $\to 1.7\times10^{12}$/curve & RESULT & emitted $\LMO$ circuit (cond.\ Table~\ref{tab:lmo_params}) \\
FT stack ($d{=}25$, $1.4\times10^7$ qubits, $7$~h) & EXTRAP. & derived FT estimator (Section~\ref{ssec:lfp_resources}) \\
Full $292$-qubit dynamical trajectory & NEITHER & FTQC future \\
\bottomrule
\end{tabular}
\end{table}

The extrapolations of correctness checks extend across size constraints. 
We use a three-stage quality control protocol to certify the emitted gate list equals the classical diagonal step to $6\times10^{-16}$ (Section~\ref{ssec:phase_qc}), validating the emitter and transferring to any system size.
Gate counts are exact for the emitted production circuit, run at $292$ qubits and the gates counted without amplitude evolution. 
The convergence ratio fixes the integrator order; the production step count $N_{\mathrm{step}}$ is set by the second-order error constant $W_2(L)$, which we measure on a growing-$L$ Holstein staircase and find to
be at most extensive in $L$ (the per-bond constant is non-increasing), so $N_{\mathrm{step}}\sim\sqrt{W_2 T^3/\varepsilon}$ grows at most as $\sqrt{L}$. 
An extrapolation to $L=8$ reproduces the production $N_{\mathrm{step}}$ within a small factor, and an a-priori double-commutator bound brackets it from above.
Per-mode bounds (Fock-tail, Macridin, dual-space Nyquist) plus the Lanczos truncation property for $K$ are analytic.
A growing-$L$ staircase confirms the intensive observables converge and no site register's realized occupancy inflates past its per-mode bound $\langle n\rangle\approx(g/\omega)^2$ under cross-mode entanglement.
Finally, headline empirical numbers ($12.6$~mV, $\pm5$~mV) are classically-computed equilibrium results; the quantum dynamics reproduces them by $E_p$ bath-independence (Appendix~\ref{sm:chainbath}), which we clarify to avoid crediting the quantum algorithm for a result produced by the classical solver.

\section{Error Analysis}\label{sec:error}

Total simulation error decomposes into Trotter error of the split propagator, discretization error of the finite grid registers, synthesis error of compiling rotations to Clifford$+$T, and the chain-truncation error of the finite bath registers.
By the triangle inequality,
\begin{equation}\label{eq:error_total}
\epsilon_{\mathrm{total}} \;\le\;
\epsilon_{\mathrm{Trotter}} + \epsilon_{\mathrm{disc}}
+ \epsilon_{\mathrm{synth}} + \epsilon_{\mathrm{chain}}(K),
\end{equation}
each term an operator-norm distance between ideal and implemented evolution. 
We emphasize that finite-temperature representation is not approximate, as the Tamascelli doubling is operationally equivalent~\cite{tamascelli2019efficient}, verified numerically to $0.21\%$ against a thermally initialized chain.
There is also no Born--Markov error because no master equation is ever invoked.
The target precision is set by electrochemical observables: $\pm 5$~mV on the $\VOCV$ plateau ($1.8\times10^{-4}$~Ha) with a $1.5$~mV per-anchor estimation target driving the measurement budget (Section~\ref{sec:resources}).

\subsection{Trotter Error: Two Nested Splittings}\label{ssec:trotter_error}

We compose two splitting schemes within the Trotter step (\ref{fig:arch_strang}): a standard symmetric Strang splitting over the layer groups (kinetic, diagonal potential, system--bath, chain), and a nested Jordan--Wigner even/odd color splitting of the fermionic hopping layer (plus a second inner symmetric split of the bosonic chain hops). 
Because every layer of the composition is palindromic, the nested product remains a symmetric (time-reversible) integrator and the global order-2 property is preserved. 
Per-step errors of the three channels add at
$\mathcal{O}(\Delta t^3)$:
\begin{equation}\label{eq:trotter_decomp}
\epsilon_{\mathrm{Trotter}} \;=\;
\epsilon_{\mathrm{Strang}}^{\mathrm{outer}}
+ \epsilon_{\jw}^{\mathrm{inner}}
+ \epsilon_{\mathrm{chainhop}}^{\mathrm{inner}}
+ \mathcal{O}(\Delta t^5).
\end{equation}

\paragraph{Outer channel.}
The leading error is the double-commutator expression (Eq.~\ref{eq:trotter_error}), its magnitude set by which operator pairs fail to commute. 
Polynomial-degree casework of the closed-system analysis~\cite{courtney2026oracle, childs2021theory} carries over: for a kinetic operator $T = p^2/2\mu$ on an $N = 2^n$ grid with Nyquist-bounded
momentum, occupation-independent potential terms contribute commutator norms that grow with their polynomial degree in the coordinates, with bilinear system--bath coupling $V_{SB} = c_0\,q\, q_0^{\mathrm{chain}}$ playing the role of the dominant cross term. 
Binding partners at the cathode operating point are $[T_{\mathrm{ph}},V_{\mathrm{diag}}]$ and $[T_{\mathrm{chain}}, V_{SB}]$, the latter bounded by $32\lambda^2\omega_{\mathrm{ph}}\Delta t^3 \sim 10^{-15}$ per step at the production timestep. 
The occupation-diagonal terms are benign in this channel: $\sum_{i<j}U_{ij}n_in_j$ and $\sum_i\varepsilon_in_i$ commute with bosonic kinetic and potential layers, leading the Coulomb sector to contribute no outer-Strang error at any coupling strength.

\paragraph{Jordan--Wigner channel.}
The occupation sector pays Trotter error against the hopping layer, confined to the inner color split.
The inter-class commutator of the hopping bonds is bounded by Equation~\ref{eq:jw_strang_bound} as $3(L-2)\,t_{\mathrm{hop}}^3(\Delta t)^3/4 \approx 2.0\times10^{-6}$ per step at $L = 8$ production parameter. 
The cross commutator $[T_{\jw\mathrm{hop}}, V_{\mathrm{diag}}]$ is the second contributor, where the occupation-dependent part scales as $t_{\mathrm{hop}}\,U_{\mathrm{nn}}\,\Delta t^3$ per adjacent pair.
The (multi-polaron-active) Coulomb term enters the Trotter budget through this channel, with both contributors are $\mathcal{O}(\Delta t^3)$ and are absorbed in the symmetric composition. 
Accumulated over $N_{\mathrm{step}} = 10^3$ steps the conservative bound is $\lesssim 2\times10^{-3}$ in state norm, exceeding the empirically observed error by up to several orders of magnitude, a behavior noted among Trotter error derivations and analyses~\cite{childs2021theory}.

A state-norm bound considers the action of a gate-constructed operator against the ideal action of that operator. 
The translation to observables is the propagation $|\delta\langle O\rangle| \le 2\|O\|\,\|\delta\psi\|$ for a bounded observable $O$. 
The protocol's primary observable is the normalized occupancy $O = N/N_{\max}$ with $\|O\| = 1$, so even the conservative $2\times10^{-3}$ state-norm accumulation translates to $\delta\SOC \le 4\times10^{-3}$, residing inside the binding $\delta\SOC = 0.02$ shape requirement of Section~\ref{ssec:measurement_bottleneck} with a factor-of-five margin.
The voltage axis is a control (Section~\ref{ssec:measurement_protocols}), carrying no Trotter error by construction.
Applying the same inequality to $N$ and $N^2$ gives a worst-case relative error of $\sim4\%$ on $\mathrm{Var}(N) \approx 11.6$, essentially at the $4\%$ FWHM grade on paper.
We resolve this empirically rather than by tightening the analytic bound: the $2\times10^{-3}$ figure is a triangle-inequality worst case that the measured quality control residuals undercut by several orders of magnitude, and the convergence gate of Section~\ref{ssec:phase_qc} is applied to the observables themselves rather than to the state norm. 
The production configuration is certified at the level at which the budget is stated. The analytic margin is independently restorable, and tightening the propagator target to $10^{-4}$ costs only $2.7\times$ in trajectory T-count ($k_4 = 303$ versus $114$; Table~\ref{tab:integrator_comparison} and Appendix~\ref{sm:integrators}).

\paragraph{Chain-hop channel.}
The bosonic chain hop $t_k(a_k^\dagger a_{k+1} + \mathrm{h.c.}) = (t_k/2)(q_kq_{k+1} + p_kp_{k+1})$ mixes non-commuting quadratures ($[q\cdot q,\, p\cdot p] \neq 0$), making its emission a symmetric inner split, $q\cdot q(\tfrac{\Delta t}{4}) \to p\cdot p(\tfrac{\Delta t}{2}) \to q\cdot q(\tfrac{\Delta t}{4})$ (Algorithm~\ref{alg:chain_layer}).
Residual error scales as $t_k^3\Delta t^3$ with the chain coefficients $t_k < \Omega_c$, numerically below $10^{-12}$ per step at production parameters.

\paragraph{Composed order verification.}
The decomposition in Equation~\ref{eq:trotter_decomp} is empirically confirmed by the convergence gate: halving $\Delta t$ must reduce end-to-end trajectory error by a factor in $[2, 6]$ (theoretical value $4$ for a clean second-order integrator). 
The production composition (the symmetric fourth-order scheme) is validated on the simulable instances, where the ratio is measured at $\approx 3.99$, certifying integrator order apart from any time evolution of the $292$-qubit register.
The size-dependent step count $N_{\mathrm{step}}$ is grounded separately by the Trotter-constant staircase of Section~\ref{ssec:scope}. 
The gate is run per material specification (Sections~\ref{sec:lfp_spec},\ref{sec:lmo_spec}) so that a regression in any one splitting channel is caught by the composed observable. 
The decomposition of Equation~\ref{eq:trotter_decomp} carries over unchanged to the production fourth-order composition of Section~\ref{ssec:trotter_step}, being a palindromic product of $S_2$ substeps. 
Each channel's bound enters per substep with the composition coefficients. 
The empirical gate above was exercised at order 2, while the order-4 gate (halving ratio window centered on $2^4 = 16$) is part of the standing harness for production runs. 
The order-2 configuration is retained as the gate-validated comparison baseline.

\subsection{Discretization Error}\label{ssec:disc_error}

Each bosonic register represents its wavefunction on $2^n$ DVR points.
For smooth states the truncation error falls exponentially in resolution, $\epsilon_{\mathrm{disc}} \le C\exp(-\alpha\,2^n/N_{\mathrm{eff}})$, with $N_{\mathrm{eff}}$ the effective support~\cite{macridin2018digital}. 
The binding constraint at the cathode operating point is the same as with vibronics and other linearly-scaled DVR representations, being the dual-space Nyquist bound (Eq.~\ref{eq:nyquist_HO}), guaranteeing coherent propagation of the state \emph{and its momentum-space dual} on the periodic grid.
These are grid-simulation limits (information-theoretic), whereby dynamical simulation must contain the bulk of the wavepacket in both domains for the full evolution time, being a stronger demand than static eigenstate representation. 
Production registers ($n_x = 12$, $n_b = 9$, $n_{d,\mathrm{Li}} = 10$, $n_{d,\mathrm{ph}} = 9$ at $c = 4$ containment) hold $\epsilon_{\mathrm{disc}} \le 10^{-6}$ per the three-bound analysis (Fock tail, Macridin, Nyquist) summarized in Section~\ref{sec:lfp_spec}.
The reversible-arithmetic block adds a fixed-point truncation at working width $w$.
With $w = 14$ and Newton depths $N_R = N_S = 3$ the arithmetic representation error sits below the rotation-synthesis floor and is folded into $\epsilon_{\mathrm{synth}}$~\cite{haner2018optimizing}.

\subsection{Rotation-Synthesis Error}\label{ssec:synth_error}

Arbitrary $R_z(\theta)$ rotations compile to Clifford$+$T using Ross-Selinger synthesis, with $N_T(\epsilon_\theta) = 3\log_2(1/\epsilon_\theta) + \mathcal{O}(\log\log 1/\epsilon_\theta)$~\cite{ross2016optimal}.
Synthesis errors accumulate as a random phase walk, $\epsilon_{\mathrm{synth}} \approx \sqrt{N_{\mathrm{rot}}}\,\epsilon_\theta$ over the $N_{\mathrm{rot}} = N_{\mathrm{step}}R_{\mathrm{step}}$ rotations of a trajectory, and the budget requires this to stay below the $5\times10^{-5}$~Ha ($1.5$~mV) per-anchor voltage target including margin for the derivative observable $\dQdV$. 
An audit of the cumulative budget finds $\epsilon_\theta = 10^{-10}$ binding at $\sim0.67$~mV.
We tighten to $\epsilon_\theta = 10^{-12}$, at which the naive per-rotation cost $T_{\mathrm{RS}}^{\mathrm{naive}} \approx 120$ is reduced by the amortized synthesis pipeline (mixed-fallback Ross--Selinger with T-count optimization passes~\cite{courtney2026oracle, heyfron2019efficient}) to $T_{\mathrm{RS}} \approx 25$ per rotation, the value used throughout Section~\ref{sec:resources}.
Classical emulations use double precision ($\epsilon_\theta \approx 10^{-16}$), conservative against this budget.

In the worst case, a union bound with every synthesis error coherently aligned has a linear error accumulation, $\epsilon_{\mathrm{synth}} \le N_{\mathrm{rot}}\,\epsilon_\theta$. 
At the fourth-order production trajectory the aggregate count is $N_{\mathrm{rot}} = 5k_4 R_{\mathrm{step}} \approx 570 \times 8.1\times10^3 \approx 4.6\times10^6$ rotations, so the union bound at $\epsilon_\theta = 10^{-12}$ gives $\epsilon_{\mathrm{synth}} \le 4.6\times10^{-6}$, i.e.\ $\approx 0.13$~mV against the $1.5$~mV per-anchor target ($\approx 0.19$~mV for the LMO trajectory at $R_{\mathrm{step}} = 1.2\times10^4$).
The random-walk model therefore only buys back margin and the per-rotation synthesis precision $\epsilon_\theta$, hence $T_{\mathrm{RS}}$, is set by the aggregate trajectory rotation count either way.

\subsection{Chain-Truncation Error and the Recurrence Bound}
\label{ssec:chain_error}

The only bath approximation is the finite chain length.
Truncation error is exponentially suppressed beyond the causal light cone, $\epsilon_{\mathrm{chain}} \le C\exp[-(K - K_{\mathrm{cone}})/\xi]$, and the usable evolution window is set by the recurrence time $\Trec = \pi K/\Omega_c$.
Observables are extracted at $T_{\mathrm{evol}} \lesssim \Trec$ before reflections from the chain end return to the system. 
The practical certificate is convergence against a numerically exact reference: the chain-mapped dynamics approach hierarchical equations of motion monotonically in $K$, reaching $2.07\%$ root-mean-square (RMS) on the system
population at $K = 7$ in the strong-coupling regime where the weak-coupling Lindblad reference departs by $4$--$16\%$ (Section~\ref{sec:lfp_spec}). 
The production geometry uses $K = 10$ with the Lanczos truncation property making $K$-extension tests cheap. The full per-source budget at the production operating point is collected in Table~\ref{tab:error_budget}.

\begin{table}[h]
\centering
\caption{Consolidated error budget at the production operating point.
Trotter entries are per step; cumulative figures assume
$N_{\mathrm{step}} = 10^3$. Finite-$T$ doubling and the absence of a
master equation contribute no error by construction.}
\label{tab:error_budget}
\small
\begin{tabular}{l l l}
\toprule
Source & Scaling & Value (production) \\
\midrule
Outer Strang & $32\lambda^2\omega_{\mathrm{ph}}\Delta t^3$ &
  $\sim10^{-15}$ / step \\
JW color split & $3(L{-}2)t_{\mathrm{hop}}^3\Delta t^3/4$ &
  $2.0\times10^{-6}$ / step; $\lesssim 2\times10^{-3}$ cum. \\
JW $\times$ Coulomb cross & $\propto t_{\mathrm{hop}}U_{\mathrm{nn}}\Delta t^3$ &
  within JW budget \\
Chain-hop inner split & $\propto t_k^3\Delta t^3$ &
  $< 10^{-12}$ / step \\
Register discretization & Nyquist/Macridin/Fock (3 bounds) &
  $\le 10^{-6}$ \\
Rotation synthesis & $\sqrt{N_{\mathrm{rot}}}\,\epsilon_\theta$ &
  $< 1.5$~mV at $\epsilon_\theta = 10^{-12}$ (union bound: $0.13$~mV) \\
Arithmetic fixed point & $w = 14$, $N_R = N_S = 3$ &
  below synthesis floor \\
Chain truncation & $C\,e^{-(K - K_{\mathrm{cone}})/\xi}$ &
  $2.07\%$ @ $K{=}7$ vs HEOM; $K{=}10$ prod. \\
Finite-$T$ doubling & exact & $0$ ($0.21\%$ numerical check) \\
Born--Markov & not invoked & $0$ (avoided $\approx 0.29$) \\
\bottomrule
\end{tabular}
\end{table}

\section{Specification I: \texorpdfstring{$\LFP$}{LFP} (Single-Polaron)}\label{sec:lfp_spec}

We first specify Equation~\ref{eq:master_H} to olivine $\LFP$: a single carrier on the Fe sublattice, the continuous Li coordinate retained, and all three reservoirs attached.
This is the validation-grade instance and every gate of the verification harness runs
here.
The single-carrier sector makes the inter-site Coulomb $U_{ij}n_in_j$ emitted but dynamically inert, a fundamental contrast to Section~\ref{sec:lmo_spec}.

\subsection{Instantiation and Parameters}\label{ssec:lfp_instantiation}

$\LFP$ in the olivine phase admits a small-polaron description at each Fe site~\cite{maxisch2006ab, hoang2011tailoring}: the carrier localizes on a single Fe, distorting the surrounding $\mathrm{O_6}$ octahedron, while Li migrates along the $[010]$ channel over a washboard barrier, coupled to the polaron by the screened attraction $V_{\mathrm{LiP}}$ and to the lattice by the Holstein term~\cite{yamada2006room}.
Polaron hopping is nearest-neighbor~\cite{asari2011formation, johannes2012hole}. 
All parameters of Equation~\ref{eq:master_H} are atomic-unit-valued and tabulated in Table~\ref{tab:lfp_params}.

\begin{table}[h]
\centering
\caption{$\LFP$ reference parameters. Values are atomic units in the simulation; physical conversions and sources in the third column.}
\label{tab:lfp_params}
\begin{tabular}{lcc}
\toprule
Parameter & Symbol & Physical value (source) \\
\midrule
Polaron hopping & $t_{ij}$ & $25~\mathrm{meV}$ (DFT~\cite{maxisch2006ab}) \\
Holstein coupling & $g$ & $1.64\,\omega_{\mathrm{ph}}$ (DFT-anchored, Section~\ref{ssec:lfp_instantiation}) \\
Polaron binding & $E_p = g^2/\omega_{\mathrm{ph}}$ & $175~\mathrm{meV}$ (mid-Maxisch~\cite{maxisch2006ab}) \\
Phonon frequency & $\omega_{\mathrm{ph}}$ & $65~\mathrm{meV}$ (DFT~\cite{maxisch2006ab}) \\
Bath cutoff & $\Omega_c$ & $120~\mathrm{meV}$ (full phonon bandwidth) \\
NN Coulomb & $U_{\mathrm{nn}}$ & $0.50~\mathrm{eV}$ (screened) \\
Coulomb decay & $\nu$ & $1.0$ (Coulomb $1/r$) \\
Dielectric const. & $\varepsilon_\infty$ & $4.5$ (DFT~\cite{maxisch2006ab}) \\
$V_{\mathrm{LiP}}$ depth & $V_{\mathrm{LiP}}^0$ & $0.5$ a.u.\ ($200~\mathrm{meV}$ NN binding) \\
Washboard barrier & $V_{b,\mathrm{wb}}$ & $0.15~\mathrm{eV}$ \\
Lattice constant & $a_{\mathrm{lat}}$ & $6.01$~\AA \\
Soft-Coulomb core & $\rho_c$ & $0.1$~\AA \\
$\VOCV$ plateau & --- & $3.45~\mathrm{V}$ (Yamada~\cite{yamada2006room}) \\
$\dQdV$ peak FWHM & --- & $\sim 20~\mathrm{mV}$ (Yamada~\cite{yamada2006room}) \\
\bottomrule
\end{tabular}
\end{table}

\noindent\emph{The coupling convention and the production point.} The Holstein term is realized as $\exp[-i\,g\sqrt{2}\,q\,n\,\tau]$, i.e.\ $H_{\mathrm{int}} = g(a+a^\dagger)\,n$ in the standard convention, for which polaron binding is $E_p = g^2/\omega_{\mathrm{ph}}$ and mean phonon number $\langle n\rangle = (g/\omega_{\mathrm{ph}})^2$. 
$S_{HR} = (g/\omega_{\mathrm{ph}})^2/2$ quoted below is a halved labelling convention instead of the standard Huang--Rhys factor. 
At the shipped default $g/\omega_{\mathrm{ph}} = 1.5$ the realized binding is already $E_p = 146~\mathrm{meV}$, at the lower edge of Maxisch's DFT range $150$--$200~\mathrm{meV}$~\cite{maxisch2006ab} (the $E_p \leftrightarrow g/\omega_{\mathrm{ph}}$ mapping is tabulated in Appendix~\ref{sm:convention}). We adopt the production point $g/\omega_{\mathrm{ph}} = 1.64$, placing $E_p = 175~\mathrm{meV}$ at the center of the DFT band.
$E_p$ is linear in occupancy and absorbed by the voltage anchor, leaving the $\VOCV/\dQdV$ observables and the register sizing unchanged (Section~\ref{ssec:lfp_registers}).
This deepens self-trapping ($t_{\mathrm{eff}} = t\,e^{-S}$ falls further) and strengthens the $U_{\mathrm{nn}} \gg E_p \gg k_BT \gg t_{\mathrm{eff}}$ hierarchy of Section~\ref{ssec:lfp_robustness}.
The dielectric $\varepsilon_\infty = 4.5$ is the optical constant. 
Since the lattice and carrier response are simulated explicitly, folding them into $\varepsilon$ would double-count, leaving only the electronic polarizability in a static screening constant.

\emph{Why $\LFP$ first.} The (halved-convention) Huang--Rhys factor $S = g^2/2\omega_{\mathrm{ph}}^2 \approx 1.3$ at the production point places it in the self-trapping regime where perturbative polaron
methods lose control, while the chain-bath construction carries no assumption on $S$.
Li migrates strictly along $[010]$ (cross-channel hopping suppressed $\sim10^4\times$, thermodynamic cross-channel contribution $\lesssim 1$~mV), exposing polaron physics in one clean dimension. 
The Yamada plateau~\cite{yamada2006room} ($\VOCV(\SOC{=}0.5) = 3.45$~V, $\dQdV$ FWHM $\sim20$~mV, $\sim80\%$ extent) anchors the voltage, with the shape acting as a falsifiable prediction. 
The simulation broadens the peak by ensemble heterogeneity. 
The plateau is a thermodynamic $\LFP$/$\FP$ coexistence~\cite{zhou2006configurational, malik2011kinetics} that the single cell represents at the mean-field level. 
Polaron-mediated effective attraction $U_{\mathrm{eff}} = -E_p/Z$ produces a flat $V(N)$ region by a
Bragg--Williams crossing rather than true phase separation (Appendix~\ref{sm:provenance}).
Explicit active-Coulomb many-body and coexistence tests are deferred to multi-polaron $\LMO$
(Section~\ref{sec:lmo_spec}).

\subsection{Register Specification}\label{ssec:lfp_registers}

The production geometry instantiates Equation~\ref{eq:tensor_product} at
\begin{equation}
\label{eq:convergence_sizes}
n_x = 12,\quad n_b = 9,\quad
n_{d,e} = 1,\quad n_{d,\mathrm{Li}} = 10,\quad n_{d,\mathrm{ph}} = 9,\quad
K_e = K_{\mathrm{Li}} = K_{\mathrm{ph}} = 10,\quad L = 8,
\end{equation}
the per-register widths set by the binding bound among Fock tail, Macridin, and dual-space Nyquist (Table~\ref{tab:register_convergence}), totaling the $292$ system qubits of Equation~\ref{eq:total_qubits} plus the $\sim84$-qubit arithmetic workspace. 
The binding criterion is the dual-space Nyquist bound (Eq.~\ref{eq:nyquist_HO}), guaranteeing coherent propagation of the wavefunction and its dual, being the relevant criterion for real-time evolution and is far more demanding than ground-state-energy convergence (a $4$-qubit phonon grid converges the static energy but supplies only $4$ of the $n^{\mathrm{Nyq}} \ge 9$ qubits the coherent propagation requires). 
The DFT production point (Section~\ref{ssec:lfp_instantiation}) raises the system-phonon dynamical mean $\langle n\rangle = 1.21 \to 1.43$ and $N_{\mathrm{ph}}^{\mathrm{dyn}} = 22 \to 24$, staying below the $n^{\mathrm{Nyq}} = 9 \to 10$ boundary at $N_{\mathrm{ph}} = 25$.

\begin{table}[h]
\centering
\caption{Per-register convergence bounds at the $\LFP$ operating point.
$N_{\mathrm{ph}}^{\mathrm{dyn}}$ is the Fock-tail cutoff at
$\varepsilon_{\mathrm{tail}} = 10^{-8}$ with a $2\times$
dynamical-excitation safety factor; $n^{\mathrm{Fock}}, n^{\mathrm{Mac}},
n^{\mathrm{Nyq}}$ are the three bounds, the binding one bold.}
\label{tab:register_convergence}
\small
\begin{tabular}{p{4cm} c c c c c c}
\toprule
Register & $\langle n\rangle$ & $N_{\mathrm{ph}}^{\mathrm{dyn}}$ &
 $n^{\mathrm{Fock}}$ & $n^{\mathrm{Mac}}$ & $n^{\mathrm{Nyq}}_{c=4}$ & $n$ \\
\midrule
$X$ (Li coord.) & --- & --- & --- & --- & --- & $\mathbf{12}$ \\
$\jw_i$ (polaron) & --- & --- & --- & --- & --- & $\mathbf{1}$ (fermionic exact) \\
$B_i$ (sys.\ phonon) & $1.21$ & $22$ & $5$ & $7$ & $\mathbf{9}$ & $\mathbf{9}$ \\
$D_{e,k}$ (e bath) & --- & --- & --- & --- & --- & $\mathbf{1}$ (fermionic exact) \\
$D_{\mathrm{Li},k}$ (Li bath, DVR) & $\le 2.0$ & $28$ & $5$ & $7$ & $\mathbf{10}$ & $\mathbf{10}$ \\
$D_{\mathrm{ph},k}$ (ph bath, DVR) & $0.56$ & $16$ & $4$ & $6$ & $\mathbf{9}$ & $\mathbf{9}$ \\
$K_r$ (chain length) & --- & --- & --- & --- & --- & $\mathbf{10}$ \\
\bottomrule
\end{tabular}
\end{table}

\subsection{Observable Extraction}\label{ssec:lfp_observables}

The production geometry is a $292$-qubit chain-bath register, with real-time dynamics costed by the resource budget of Section~\ref{ssec:lfp_resources}.
Validation is therefore on reduced instances (Section~\ref{ssec:scope}). 
Dynamics are benchmarked on the small-$L$ Trotter/register staircases and the HEOM cross-checks. (ii)~The reduced two-channel solver computes the thermodynamic $\VOCV(\SOC)$ and $\dQdV(V)$ from ground-state chemical-potential sweeps, as the $K_r \to \infty$ mean-field limit of the full chain-bath, with Holstein dressing entering as a static Lang--Firsov renormalization.
For a static carrier coupled linearly to harmonic modes the equilibrium binding is exactly a static number, $E_p = g^2[V^{-1}]_{ss}$, which the Caldeira--Leggett counter-term renders independent of bath length (Section~\ref{ssec:scope}).
The open-system chain-bath machinery is load-bearing for the dynamical observables such as $Z(\omega)$, the temperature-dependent lineshape $S_{\mathrm{eff}} = S\coth(\omega/2k_BT)$, and the multi-polaron $\LMO$ active-Coulomb sector. 
The novelty and the single-polaron validation probe different observables at different cost, with convergence governed by the $\mu$-grid and anchor count. 
Registers are verified converged independently, where we see the system--chain Schmidt-entropy increment shrinking $\sim340\times$ from $n_b = 2{\to}3$ vs $3{\to}4$; the closed propagator matches dense exact diagonalization to $6\times10^{-11}$).

$\VOCV$ is computed by sweeping the electron and Li$^+$ reservoir chemical potentials under charge neutrality, $N_e + N_{\mathrm{Li}} = N_{\mathrm{cell}}(1 - \SOC)$, and reading cell voltage at the self-consistent fixed point,
\begin{equation}
\label{eq:vocv_combined}
\VOCV(\SOC) = \mu_e^{\mathrm{coll}} - \mu_{\mathrm{Li}}^{\mathrm{elyte}}
 + V_{\mathrm{ref}},
\end{equation}
the two channels tied by a Legendre-transform identity with residuals $\le 0.10\%$. $\dQdV(V)$ follows by centered finite difference on a uniform $\SOC$ grid, cross-checked by a Savitzky--Golay extraction with sum-rule agreement $\le 0.36\%$ and quadratic-fit peak refinement.
The energy scale is Maxisch's DFT $t_{ij} = 25$~meV ($\pm 20\%$ functional uncertainty), and the absolute voltage is shifted once so that $\VOCV(\SOC{=}0.5) \equiv 3.45$~V at the Yamada plateau, following the DFT intercalation-voltage construction~\cite{aydinol1997ab}. 
The anchor is a one-parameter physical-units rescaling, leaving peak position, FWHM, shoulders, plateau extent, and shape to be structural predictions conditional on Table~\ref{tab:lfp_params}.

\paragraph{Equilibration versus the recurrence window.}
$\VOCV$ is an equilibrium quantity while the finite chain is faithful only for $T_{\mathrm{evol}} \lesssim \Trec = \pi K/\Omega_c \approx 172$~fs at $K = 10$. 
We clarify here why the extracted observable is the steady state and not a transient proxy.
Equilibrium is prepared by the Tamascelli construction~\cite{tamascelli2019efficient}, preparing the grand-canonical thermal state $\rho(\mu, \beta)$ directly (vacuum chains, zero preparation gates). 
At this coupling the reduced equilibrium is the mean-force Gibbs state $\rho_S = \mathrm{Tr}_B\,[e^{-\beta H}]/Z$, the bath-dressed stationary state.
Tracing over the prepared global state returns this quantity. 
The thermodynamic observables ($\VOCV$, $\dQdV$) are read from chemical-potential sweeps of the equilibrium two-channel solver. 
The slow reservoir-exchange scales set the fluctuation statistics of prepared states, sampled by population measurements detailed in Algorithm~\ref{alg:anchor_protocol}. 
These time scales can never fit the window ($\gamma_e^{-1} \approx 20$~ps and $\gamma_{\mathrm{Li}}^{-1} \approx 2$~ps against $172$~fs), but do not need to. 
Fast dynamics are probed by the trajectories, where polaron dressing occurs at $1/\omega_{\mathrm{ph}} \approx 10$~fs and the bath memory at $1/\Omega_c \approx 5.5$~fs. 
These give margins of $17\times$ and $31\times$ against the recurrence.
Dynamical observables like ($Z(\omega)$ poles and Kubo--Mori correlators at imaginary times $\tau \le \beta$) are extracted within that window where the chain is moment-exact.

\subsection{Validation}\label{ssec:lfp_validation}

We validate the specification beyond the quality control protocol given above in Eq.~\ref{eq:phase_qc_def}. 
We find the closed-system propagator matches dense exact diagonalization to $6\times10^{-11}$ on sub-convergent geometry. 
At Ohmic coupling $\alpha \le 0.1$, where the Born--Markov expansion is valid, the chain-bath dynamics agree with the Lindblad reference. 
Where Lindblad fails, we compare the chain-mapped dynamics to numerically exact, representation independent hierarchical equations of motion (HEOM) on a shared spin-boson kernel (Drude--Lorentz $J(\omega)$, $T = 0.4\,\omega_0 \equiv 300$~K), which was recently realized experimentally with structured baths~\cite{sun2025quantum}. 
The framework converges to HEOM monotonically in chain length, reaching $2.07\%$ RMS on the system population at $K = 7$, while the Lindblad reference departs by $4$--$16\%$ as coupling grows (\ref{fig:heom_xv}). 
The HEOM reference is converged itself, changing by $< 0.01\%$ under increased hierarchy depth or Pad\'e count, two orders of magnitude below the $2.07\%$ quoted, ensuring agreement measures chain-mapping fidelity rather than reference truncation (Appendix~\ref{sm:heom_dimer}).
Because a single spin-boson is the regime where chain mapping is already exact, we additionally validate the multi-site structure against the canonical excitonic dimer~\cite{ishizaki2009unified}. 
This is a two-site system with coherent inter-site coupling, each carrying its own independent Drude--Lorentz bath. 
The two-chain construction reproduces the HEOM site population to $3.45\%$ RMS at the production reorganization, with inter-site transport as the multi-site analog of spin-boson agreement (Appendix~\ref{sm:heom_dimer}).
\begin{figure}[h]
\centering
\includegraphics[width=0.92\textwidth]{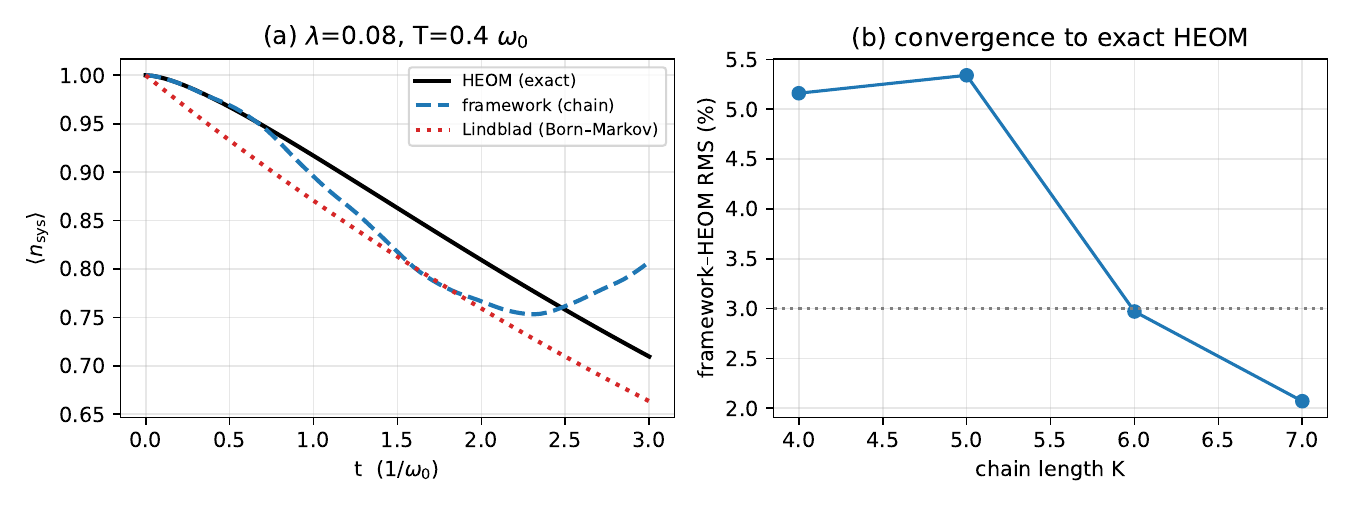}
\caption{Strong-coupling validation against numerically exact HEOM on the shared spin-boson kernel (Drude--Lorentz $J(\omega)$, $T = 0.4\,\omega_0 \equiv 300$~K). \emph{(a)} Coupling sweep: the chain-mapped dynamics track HEOM to within a few percent across the regime where the weak-coupling Lindblad reference differs by $4$--$16\%$. 
\emph{(b)} Chain-length convergence: the framework approaches HEOM monotonically in $K$, reaching $2.07\%$ RMS on the system population at $K = 7$.}
\label{fig:heom_xv}
\end{figure}

We check norm and charge drift $< 10^{-13}$ over production trajectories, and the composed Trotter order verified at ratio $\approx 3.99$ (Section~\ref{ssec:trotter_error}).
All verifications can be found in the codebase associated to this work~\cite{courtney2026multipolaron}.

\subsection{Results}\label{ssec:lfp_results}

\begin{table}[h]
\centering
\caption{Headline $\LFP$ observables after the single voltage anchor.}
\label{tab:gates}
\begin{tabular}{lcc}
\toprule
Observable & Target & Observed \\
\midrule
$\VOCV(\SOC = 0.5)$ & $3.45 \pm 0.05~\mathrm{V}$ & $3.45000~\mathrm{V}$ (anchor) \\
$\dQdV$ peak position & $3.45 \pm 0.04~\mathrm{V}$ & $3.4528~\mathrm{V}$ \\
$\dQdV$ peak FWHM & $\sim 20~\mathrm{mV}$ (factor-of-2 band) & $\approx 12.6~\mathrm{mV}$ \\
\bottomrule
\end{tabular}
\end{table}

\begin{figure}[h]
\centering
\includegraphics[width=0.95\textwidth]{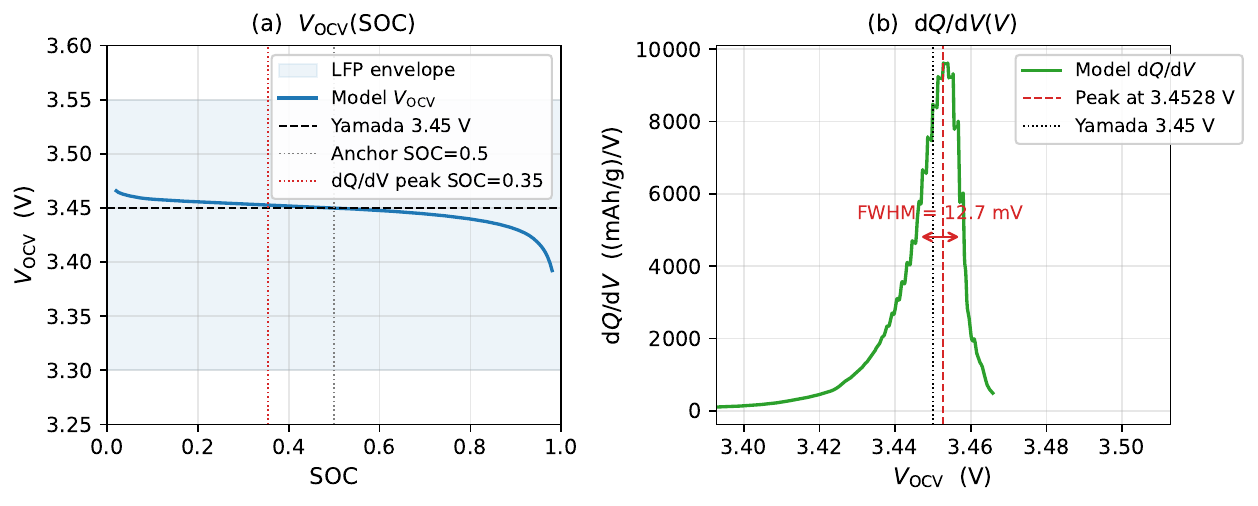}
\caption{(a) Predicted $\VOCV(\SOC)$ for $\LFP$ in physical volts after
the single voltage-zero anchor; the dashed horizontal line is the Yamada
plateau~\cite{yamada2006room}, vertical dotted lines the calibration
anchor and the $\dQdV$-peak SOC. (b) Extracted $\dQdV(V)$ at converged
spectral resolution ($K_e = K_{\mathrm{Li}} = 10$, $N_{\mathrm{anc}} =
400$): peak at $3.4528$~V vs the experimental $3.45$~V ($2.8$~mV away,
within the $\pm 5$~mV plateau reproducibility), FWHM $12.65$~mV vs
$\sim20$~mV experimental. An under-resolved $K = 4$ grid instead gives
$3.4515$~V and an artificially broad $18.16$~mV; both observables
converge only with $K$, and the $\dQdV$-peak SOC drifts with $K$, which
is why the structural claim is framed at the peak \emph{voltage} within
experimental reproducibility rather than at fixed SOC.}
\label{fig:F2_F3_vocv_dqdv}
\end{figure}

The $\dQdV$ peak lands $2.8$~mV from experiment, within the
$\sim\pm5$~mV reproducibility of the plateau itself
(Table~\ref{tab:gates}). Beyond the anchored voltage, the predicted shape
comprises: (i)~plateau extent (SOC fraction within $\pm 25$~mV of
center) $\sim80\%$, matching Yamada within $5\%$; (ii)~solid-solution
shoulders at $\SOC < 0.1$ and $\SOC > 0.9$ deviating $\sim50$~mV in
both directions, matching the two-phase onset; (iii)~monotone
non-increasing curve with unimodal $\dQdV$ peak.

\subsection{Robustness}\label{ssec:lfp_robustness}

Across the Maxisch $\pm 20\%$ DFT-uncertainty band on $t_{ij}$, the $\dQdV$ peak position moves $\pm 0.6$~mV (well inside the $\pm 40$~mV gate band), with the gates also holding across $\pm 40\%$ sweeps of $V_{\mathrm{LiP}}$, $\gamma_{\mathrm{Li}}$, and $\tau_{\mathrm{Li}}$.
Folded through the Lang--Firsov transformation, the zero-phonon kinetic-only floor at deep coupling is
\begin{equation}
\mathrm{FWHM}^{(0)}_{\mathrm{FC}}(g/\omega_{\mathrm{ph}}{=}2.5)
 = e^{-S_{HR}/2}\,\mathrm{FWHM}_{\mathrm{bare}}
 = 0.210 \times 12.6~\mathrm{mV} = 2.65~\mathrm{mV},
\label{eq:FC_lower_bound}
\end{equation}
being below the gate band because the plateau width is interaction-set per the energy hierarchy
\begin{equation}
\underbrace{U_{\mathrm{nn}} \approx 500~\mathrm{meV}}_{\text{inter-site}}
\gg
\underbrace{E_p \approx 175~\mathrm{meV}}_{\text{self-trapping (prod.)}}
\gg
\underbrace{k_BT \approx 25.9~\mathrm{meV}}_{300~\mathrm{K}}
\gg
\underbrace{t_{\mathrm{eff}}(T) \approx 1.7~\mathrm{meV}}_{\text{LF-renormalized hopping}}.
\label{eq:lfp_hierarchy}
\end{equation}

Testing at $S_{HR} = 3.125$ verifies the chain invariants ($< 10^{-13}$ drift) and the
Holstein $g^2$ scaling (excitation ratios $1, 2.38, 4.28, 6.72$).

The three width values are differing physical quantities: the bare Franck--Condon polaron line (Eq.~\ref{eq:FC_lower_bound}) ($2.6$~mV, the un-broadened floor), the equilibrium configurational width ($12.6$~mV, single particle), and the experimental ensemble width ($\sim20$~mV).
Upon a Gibbs--Thomson surface-energy shift over a
literature size distribution, the $12.6$~mV peak broadens to $\sim17$--$20$~mV with nothing fitted to the width (Appendix~\ref{sm:width}).
The $30$--$37\%$ residual between the single-particle model and experiment can be accounted by ensemble heterogeneity with particle-size distribution, interfacial strain, and non-equilibrium front broadening~\cite{bai2011suppression, cogswell2012coherency,dreyer2010thermodynamic}.

Satellites sit at $n\,\hbar\omega_{\mathrm{ph}} = 65, 130, \dots$~meV, roughly $3\times$ the entire $20$~mV width, and the progression is asymmetric. 
Therefore, it cannot symmetrically broaden a thermodynamic peak.
For single-carrier $\LFP$ the width does not act as a many-body correlation test, since the $U_{\mathrm{nn}}$ term is dynamically inert at $N_{\mathrm{pol}} = 1$ ($\langle U_{ij}n_in_j\rangle \equiv 0$), entering as a classical scale in the equilibrium configurational width (Table~\ref{tab:holstein_sweep}).
We look to the multi-polaron $\LMO$ specification for the active-Coulomb many-body test~(Section~\ref{sec:lmo_spec}).

\begin{table}[h]
\centering
\caption{Holstein-coupling sweep. 
Perturbative column: interaction-set FWHM via $U_{\mathrm{tot}} = U_{\mathrm{nn}} - E_p/Z$, with $Z = 2$ the quasi-1D Li-channel coordination and $U_{\mathrm{nn}} \approx 1.08$~eV fixed by the $g/\omega_{\mathrm{ph}} = 1.0$ anchor ($U_{\mathrm{tot}}/U_{\mathrm{nn}} = 0.97$), anchored at $g/\omega_{\mathrm{ph}} = 1.5$. 
Non-perturbative column: Lang--Firsov zero-phonon kinetic-only floor (Eq.~\ref{eq:FC_lower_bound}).
The interaction-set prediction passes the factor-of-2 gate across the sweep; the kinetic floor lies uniformly below it, confirming the width is interaction-controlled per \ref{eq:lfp_hierarchy}.
The perturbative column is calibrated at a default $g/\omega_{\mathrm{ph}} = 1.5$ and extrapolated across the sweep; at the production coupling $g/\omega_{\mathrm{ph}} = 1.64$ the converged two-channel solver gives FWHM $= 12.65$~mV (\ref{fig:F2_F3_vocv_dqdv}), the $0.05$~mV offset from the row's $12.7$~mV being perturbative-extrapolation error.}
\label{tab:holstein_sweep}
\small
\begin{tabular}{cccccccc}
\toprule
$g/\omega_{\mathrm{ph}}$ & $S$ & $E_p$ (meV) & $U_{\mathrm{tot}}/U_{\mathrm{nn}}$
 & FWHM$^{\mathrm{pert}}$ (mV) & FWHM$^{(0)}_{\mathrm{FC}}$ (mV) &
 G$^{\mathrm{pert}}$ & G$^{(0)}_{\mathrm{FC}}$ \\
\midrule
1.0 & 0.500 & 65.0 & 0.970 & 12.3 & 9.8 & PASS & FAIL \\
1.5 (default) & 1.125 & 146.0 & 0.932 & 12.6 (anchor) & 7.2 & PASS & FAIL \\
1.64 (production, Maxisch) & 1.345 & 175.0 & 0.919 & 12.7 & 6.4 & PASS & FAIL \\
2.0 & 2.000 & 260.0 & 0.880 & 13.0 & 4.6 & PASS & FAIL \\
2.5 (deep coupling) & 3.125 & 406.0 & 0.812 & 13.5 & 2.6 & PASS & FAIL \\
\bottomrule
\end{tabular}
\end{table}

\subsection{Resource Budget}\label{ssec:lfp_resources}

The locality-collapsed per-step rotation count at the production geometry decomposes as $R_{\mathrm{step}} = R_{\mathrm{kin}} + R_{\mathrm{ondiag}} + R_{\jw\mathrm{hop}} \approx 2.9\times10^3 + 5.2\times10^3 + 16 \approx 8.1\times10^3$, with on-diagonal count comprising the reference and quadratic surfaces ($339 + 1164$ rotations), occupation-linear blocks ($\approx 160$), coordinate cross terms ($\approx 900$), and the $\binom{L}{2} = 28$ Coulomb phases ($56$ Toffoli per step from the doubly-controlled pairs).
At the amortized synthesis cost $T_{\mathrm{RS}} \approx 25$ and with the single-carrier index-optimized soft-Coulomb arithmetic ($+7\times10^4$ T-gates per step, Algorithm~\ref{alg:soft_coulomb}), the per-step budget is \footnote{Toffoli gates are converted to T at the conservative no-ancilla ratio $C_T = 7$ used internally throughout (the ``$7\times$'' factors below); the community-standard $\sim\!4{:}1$ conversion is used only for the external literature comparison of Section~\ref{ssec:positioning}, where it would lower the Toffoli contributions.}
\begin{equation}\label{eq:lfp_tstep}
T_{\mathrm{step}}^{\LFP} \approx
\underbrace{25 \times 8.1\times10^3 + 7\times56}_{\text{rotations + Coulomb Toffoli}}
+ \underbrace{2.6\times10^4}_{\text{washboard UCR}}
+ \underbrace{7.4\times10^4}_{\text{soft-Coulomb arithmetic}}
\;\approx\; 3.0\times10^5~\text{T-gates},
\end{equation}
(the non-indexed arithmetic bound is $+6\times10^5$). A trajectory to a working fraction of $\Trec$ costs
$\approx 1.7\times10^8$ T-gates under the production fourth-order composition ($114$ composition steps $= 570$ certified $S_2$ substeps; this extends to $\approx 3.0\times10^8$ at the order-2 comparison point of
$N_{\mathrm{step}} = 10^3$, Table~\ref{tab:integrator_comparison}) at $\sim8.5\times10^3$
T-depth per step (rising to $\sim10^4$ when the washboard is emitted arithmetically in the depth-limited mode).
The register is $292 + 84 = 376$ logical qubits. 
Logical error is bounded per trajectory as an independent estimator sample read by partial trace.
A logical fault adds estimator variance, which we average over the $\max(n_{\mathrm{sh}}, K)$ repetitions. 
This is exact for the mean $\langle N\rangle$ ($\VOCV$). 
The $\dQdV$ estimator is a second moment $\mathrm{Var}(N)$, where a fault biases slightly rather than only
broadening, but at the $10^{-2}$ fault rate over $K \approx 160$ repetitions this is a $\sim1\%$ effect, sub-dominant to the $4\%$ FWHM grade. 
We also assume faults are uncorrelated across repetitions. Any correlated calibration drift on a reused layout would reintroduce a small bias, bounded by the same per-trajectory budget.
At a $10^{-2}$ per-trajectory target the spacetime volume ($n_{\mathrm{q}} \times$ T-depth $\times d$) requires a surface code distance $d_{\mathrm{sc}} \ge 23$ at $p = 10^{-3}$ ($p_{\mathrm{th}} \approx 10^{-2}$). 
We use $d_{\mathrm{sc}} = 25$, conservative by two. 
The per-curve budget treats the full curve as a coherent computation, which would need $d \approx 31$.
The $376$-patch data region with routing is $\sim4.8\times10^5$ physical qubits, with footprint dominated by magic-state factories. We use the convention of $1~\mu $s logical cycle time, being an optimistic future fault-tolerant goal.
Sustaining the $\sim35$~T/$\mu$s reaction-limited consumption rate needs $\sim9\times10^2$ 15-to-1 factories~\cite{fowler2012surface, litinski2019game, bravyi2005universal}, $\sim97\%$ of a $\sim1.4\times10^7$-physical-qubit total. 
Magic-state cultivation~\cite{gidney2024magic} lowers this to $\sim1.8\times10^6$. 
The high factory fraction is characteristic of a dynamics workload (large total-$T$ at modest depth), instead of static-energetics estimates of Section~\ref{sec:resources}. 
Dominant end-to-end cost is nonetheless the measurement: $K\approx 1.6\times10^2$ coherent repetitions per anchor under the direct protocol (finite-difference envelope $1.8\times10^4$), deferred to a cross-specification comparison of Section~\ref{sec:resources}.

\section{Specification II: \texorpdfstring{$\LMO$}{LMO} (Multi-Polaron)}\label{sec:lmo_spec}

Here, we fix Equation~\ref{eq:master_H} to spinel LMO, extending directly to its manganese-rich relatives (LMNO, LMRO, the Mn fraction of NMC) as a direct commercial alternative to LFP. 
At the $\SOC \approx 0.5$ working point LMO carries a $1{:}1$ Mn$^{3+}$/Mn$^{4+}$ ratio, i.e.\ $L_{\mathrm{JT}} = L/2 \approx 4$ simultaneous polarons on the eight-site active cell. 
The $\binom{L}{2}$ Coulomb phase blocks act on multiply-occupied states, contributing a finite screened repulsion ($U_{\mathrm{nn}} \approx 0.5$~eV bare; Ewald-screened at the electronic-floor dielectric $\varepsilon_\infty$, the bare $1/r$ tail being unphysical in a polar oxide, Section~\ref{ssec:lmo_validation}) to the dynamics and the extracted $\VOCV(\SOC)$.
The multi-carrier active sector also carries the Li--Li and Li--polaron screened-Coulomb blocks that complete the many-body interaction.
All are $\le 2$-local, so the active-Coulomb sector adds $\binom{L}{2} + \binom{L_{\mathrm{Li}}}{2} + L L_{\mathrm{Li}} = 120$ controlled-phase blocks per step (Table~\ref{tab:lmo_block_counts}), producing a two-plateau structure.
The Mn--Mn block has the same per-step gate count as in LFP, and the added Li--Li and Li--polaron blocks are likewise $\le 2$-local. 
The whole active-Coulomb sector stays polynomial, with resource figures below calculated through closed-form
estimates from the LFP-validated gate-count machinery. 
Their validation status is stated in Section~\ref{ssec:lmo_validation}.

\paragraph{Why LMO acts as a many-body test.}
A single cell cannot host two-phase coexistence, so the LFP plateau is a mean-field (Bragg--Williams) crossover, whose effective attraction is supercritical ($u = 5.49 = 1.37\,u_c$ against the mean-field threshold $u_c = 4$, giving an $\sim80\%$ supercritical plateau extent). 
A single finite cell cannot realize a true binodal.
Two-phase coexistence requires an inter-site interaction above the two-dimensional Ising binodal, $k_BT_c = 0.567\,W \Rightarrow W > 46$~meV at $300$~K; LFP's effective attraction does not clear it, whereas
the LMO inter-site interaction ($W_{\mathrm{eff}} \approx 66$~meV, $T_c \approx 434$~K) does (Appendix~\ref{sm:provenance}).
Tests for LMO are therefore in the multi-polaron specification.

\subsection{Instantiation and Parameters}\label{ssec:lmo_instantiation}

Mn$^{3+}$ ($d^4$, high-spin) is JT-active, coupling each Mn$^{3+}$ to a twofold $e_g$ phonon doublet $(Q_\theta, Q_\epsilon)$, adding a register pair per JT-active site with coupling $g_{\mathrm{JT}}\, n_i (q_{\theta,i}\cos\phi_i + q_{\epsilon,i}\sin\phi_i)$ realized per axis as $g_{\mathrm{JT}}(a+a^\dagger)n_i$. 
In this operator convention~(Section~\ref{ssec:lfp_instantiation}), the relaxation energy is $E_{\mathrm{JT}} = g_{\mathrm{JT}}^2/\omega_{\mathrm{JT}}$, fixing $g_{\mathrm{JT}} = \sqrt{E_{\mathrm{JT}}\omega_{\mathrm{JT}}} = 136.9~\mathrm{meV}$ through the first-principles Mn$^{3+}$ Jahn--Teller stabilization scale~\cite{marianetti2001first} $E_{\mathrm{JT}} \approx 250$~meV.
Li occupies tetrahedral 8a sites linked through octahedral 16c sites~\cite{thackeray1983w, thackeray1984electrochemical}, creating site-to-site hopping transport behavior instead of channel migration, a reduction quantitatively justified by deep-well parameters $\eta = E_b/k_BT \approx 12$--$15$ and $\zeta \approx 0.04$ (\ref{ssec:encoding_tradeoffs}).
We therefore replace the continuous $X$ register by an $L_{\mathrm{Li}}$-qubit JW occupancy register with its own hopping layer (effective coordination $z_{\mathrm{eff}} \approx 6$) with a discrete-site (A1) encoding. 
Li--polaron attraction becomes a set of fixed-angle occupation-bilinear phases with classically precomputed angles, removing the soft-Coulomb arithmetic block.
Finally, LMO shows flat plateaus near $4.0$ and $4.15$~V from Mn$^{3+}$/Mn$^{4+}$ ordering~\cite{ohzuku1990electrochemistry}, extracted natively through the two-channel protocol of Section~\ref{ssec:lfp_observables} through the same chemical-potential sweeps, making the second plateau a calibration question rather than a new register.

\begin{table}[h]
\centering
\caption{LMO physical parameters of the instantiation (code defaults of
the Stage-1 reference simulator; Mn--Mn hopping is
superexchange-parameterized).}
\label{tab:lmo_params}
\begin{tabular}{lcc}
\toprule
Parameter & Symbol & Value (source) \\
\midrule
Phonon frequency & $\omega_{\mathrm{ph}}$ & $70~\mathrm{meV}$ (Mn--O optical) \\
Bath cutoff & $\Omega_c$ & $140~\mathrm{meV}$ \\
Holstein coupling & $g$ & $1.8\,\omega_{\mathrm{ph}}$ (JT-amplified vs LFP's $1.64$) \\
JT mode frequency & $\omega_{\mathrm{JT}}$ & $75~\mathrm{meV}$ ($e_g$ stretch, JT-active) \\
JT stabilization & $E_{\mathrm{JT}}$ & $250~\mathrm{meV}$ per Mn$^{3+}$ \\
JT coupling & $g_{\mathrm{JT}} = \sqrt{E_{\mathrm{JT}}\omega_{\mathrm{JT}}}$ & $136.9~\mathrm{meV}$ ($E_{\mathrm{JT}} = g_{\mathrm{JT}}^2/\omega_{\mathrm{JT}}$) \\
JT Huang--Rhys & $S_{\mathrm{JT}} = E_{\mathrm{JT}}/\omega_{\mathrm{JT}}$ & $3.3$ \\
Coulomb sector & screened (Ewald, $\varepsilon_\infty{=}4.5$) & Mn--Mn, Li--Li, Li--polaron; \emph{active}, $\le 2$-local \\
Two-plateau split & --- & $125$~mV derived (Section~\ref{ssec:lmo_validation}); $150$~mV Ohzuku~\cite{ohzuku1990electrochemistry} \\
Plateau targets & --- & $4.0$ / $4.15$~V (Ohzuku~\cite{ohzuku1990electrochemistry}) \\
\bottomrule
\end{tabular}
\end{table}

\subsection{Registers and Encoding}\label{ssec:lmo_registers}

Scaling parameters relative to the LFP geometry of Equation~\ref{eq:convergence_sizes}: $L = 8$ (8 Mn per spinel primitive cell), $L_{\mathrm{Li}} = 8$ (8a occupancy), $L_{\mathrm{JT}} = 4$ (Mn$^{3+}$ count at $\SOC \approx 0.5$), $n_b = 9$ (same dual-space Nyquist bound; Mn--O $70$~meV vs Fe--O $65$~meV), $n_{\mathrm{JT}} = 10$ (dual-space Nyquist at the JT Huang--Rhys $S_{\mathrm{JT}} = E_{\mathrm{JT}}/
\omega_{\mathrm{JT}} \approx 3.3$, $\langle n\rangle_{\mathrm{JT}} \approx 3.4$, $q_{\max}^{\mathrm{JT}} = 6$), reservoirs unchanged ($K_r = 10$; $n_{d} = 1, 10, 9$). 
The occupation-bilinear count grows from $|\mathcal{G}_2^{<}| = 54$ (LFP) to $70$: $+2L_{\mathrm{JT}} = 8$ JT bilinears ($n_i Q_\theta$, $n_iQ_\epsilon$) and $+\sim L_{\mathrm{Li}} = 8$ correlated 8a$\to$16c$\to$8a Li-hop pairs.
The qubit ledger under the A1 encoding is
\begin{equation}\label{eq:lmo_qubits}
n_{\mathrm{tot}}^{\LMO}
 = \underbrace{292}_{\text{LFP system}}
 - \underbrace{12}_{X\text{ removed}}
 + \underbrace{8}_{L_{\mathrm{Li}}\text{ JW}}
 + \underbrace{2L_{\mathrm{JT}}n_{\mathrm{JT}} = 80}_{\text{JT doublet}}
 = 368~\text{logical qubits},
\end{equation}
with no arithmetic workspace. The larger JT doublet registers ($n_{\mathrm{JT}} = 10$, sized for the $S_{\mathrm{JT}} \approx 3.3$ self-trapping) nearly offset by discretizing away the soft-Coulomb ancilla.

\subsection{The LMO Trotter Step}\label{ssec:lmo_strang}

We build a second-order Trotter step below, adapting Equation~\ref{eq:trotter_step} with two Jahn--Teller layers inserted and the Li kinetic replaced by JW hopping:
\begin{equation}
\label{eq:trotter_step_lmo}
\begin{aligned}
\mathcal{U}^{\LMO}_{\mathrm{Strang}}(\Delta t) =\;
&V_{\mathrm{diag}}(\tfrac{\Delta t}{2})\,
 V_{SB}(\tfrac{\Delta t}{2})\,
 T_{\mathrm{chainhop}}(\tfrac{\Delta t}{2})\,
 T_{\mathrm{ph}}(\tfrac{\Delta t}{2})\,
 \mathbf{T_{\mathrm{JT}}(\tfrac{\Delta t}{2})}\,
 T_{\mathrm{chain}}(\tfrac{\Delta t}{2}) \\
&\times\; T_{\mathrm{Li\,hop}}(\Delta t)\,
 T_{\mathrm{pol\,hop}}(\Delta t) \\
&\times\; T_{\mathrm{chain}}(\tfrac{\Delta t}{2})\,
 \mathbf{T_{\mathrm{JT}}(\tfrac{\Delta t}{2})}\,
 T_{\mathrm{ph}}(\tfrac{\Delta t}{2})\,
 T_{\mathrm{chainhop}}(\tfrac{\Delta t}{2})\,
 V_{SB}(\tfrac{\Delta t}{2})\,
 V_{\mathrm{diag}}(\tfrac{\Delta t}{2}),
\end{aligned}
\end{equation}
where (boldface) $T_{\mathrm{JT}}$ is the JT-doublet kinetic layer absent in LFP, and $V_{\mathrm{diag}}$ carries the Mn$^{3+}$/Mn$^{4+}$ on-site splitting, the active Coulomb, and the JT coupling (Mn$^{3+}$ only). 
Both inner splits (JW color classes for the two hopping layers; the chain-hop quadrature split) are
symmetric, so error analysis of Section~\ref{ssec:trotter_error} applies with the JW channel now carrying two hopping layers and the $t_{\mathrm{hop}}U_{\mathrm{nn}}$ cross term physically active. 
Per-function operation counts are given in Table~\ref{tab:lmo_block_counts}.

\begin{table}[h]
\centering
\caption{Per-step function operations for multi-polaron LMO
($L = L_{\mathrm{Li}} = 8$, $L_{\mathrm{JT}} = 4$, $n_b = 9$,
$n_{\mathrm{JT}} = 10$, $K_r = 10$). Rotation counts are
locality-collapsed (Algorithm~\ref{alg:vdiag}); no $2^L$ factor enters,
and no arithmetic Toffolis appear (discrete-site encoding).}
\label{tab:lmo_block_counts}
\small\renewcommand{\arraystretch}{1.15}
\begin{tabular}{l l r l}
\toprule
Strang block & Operation & Rotations & Notes \\
\midrule
$T_{\mathrm{pol\,hop}}$ & Mn--Mn JW-XY hop & $2(L - z_{\mathrm{pol}}) \approx 14$ & even/odd colors \\
$T_{\mathrm{Li\,hop}}$ & 8a--8a JW-XY hop & $2(L_{\mathrm{Li}} - z_{\mathrm{eff}}) = 4$ & spinel graph, $z_{\mathrm{eff}}{=}6$ \\
$V_{\mathrm{diag}}$: on-site & $\varepsilon_{\mathrm{Mn}}$ splitting & $L = 8$ & single-qubit $Z$ \\
$V_{\mathrm{diag}}$: \textbf{Mn--Mn Coulomb} & $\sum_{i<j}W^{\mathrm{MnMn}}_{ij}n_in_j$ & $\binom{L}{2} = 28$ & \textbf{active}, screened CPhase \\
$V_{\mathrm{diag}}$: \textbf{Li--Li Coulomb} & $\sum_{i<j}W^{\mathrm{LiLi}}_{ij}n_in_j$ & $\binom{L_{\mathrm{Li}}}{2} = 28$ & screened CPhase (NEW) \\
$V_{\mathrm{diag}}$: \textbf{Li--polaron} & $\sum_{ij}W^{\mathrm{LiP}}_{ij}n^{\mathrm{Li}}_in^{\mathrm{pol}}_j$ & $L\,L_{\mathrm{Li}} = 64$ & screened CPhase (composite) \\
$V_{\mathrm{diag}}$: Holstein & $g\,n_iq_i$ & $n_bL = 72$ & controlled-linear \\
$V_{\mathrm{diag}}$: JT coupling & $g_{\mathrm{JT}}n_iq_{\theta,\epsilon}$ & $2n_{\mathrm{JT}}L_{\mathrm{JT}} = 80$ & Mn$^{3+}$ only \\
$T_{\mathrm{JT}}$ (kinetic) & JT doublet $p^2$ & $4L_{\mathrm{JT}}[n_{\mathrm{JT}} + \binom{n_{\mathrm{JT}}}{2}] = 880$ & \\
Tier-3 bilinears & $\mathcal{G}_2^{<}$: $54 \to 70$ & $2\cdot16\cdot n_{\mathrm{avg}}^2 \approx 3.2\times10^3$ & JT $+$ Li-hop pairs \\
Reservoir layers & as LFP & $\approx 8\times10^3$ & Algorithm~\ref{alg:chain_layer} \\
\midrule
\textbf{Total} & & $R_{\mathrm{step}}^{\LMO} \approx 1.2\times10^4$ & $1.5\times$ LFP rotations \\
\bottomrule
\end{tabular}
\end{table}

\subsection{Resource Estimate}\label{ssec:lmo_resources}

\emph{Headline.} The emitted $\LMO$ production circuit costs $368$ logical qubits, $\sim3.0\times10^5$ T-gates per Trotter step, $\sim1.8\times10^8$ T-gates per trajectory (fourth order), and $\sim1.7\times10^{12}$ T-gates for the full two-plateau $\VOCV$ curve ($\sim15$~h reaction-limited at $t_{\mathrm{log}} = 1~\mu$s). These are exact operation counts of the emitted circuit (Equations~\ref{eq:lmo_qubits}, \ref{eq:lmo_tstep}, \ref{eq:lmo_curve}), conditional only on the gate-count machinery validated at LFP (Section~\ref{ssec:lmo_validation}). The derivation follows below.

At the amortized synthesis cost $T_{\mathrm{RS}} \approx 25$ with locality-collapsed Toffoli count ($72$ per step from doubly-controlled pairs),
\begin{equation}\label{eq:lmo_tstep}
T_{\mathrm{step}}^{\LMO} \approx 25 \times 1.2\times10^4 + 7\times72
\approx 3.0\times10^5~\text{T-gates per step},
\end{equation}
$1.0\times$ the arithmetic-inclusive LFP budget (Eq.~\ref{eq:lfp_tstep}).
Comparatively, JT registers and extra bilinears roughly offset the absent soft-Coulomb arithmetic.
The per-step T-depth estimate is $\sim8\times10^3$ (rotation-scheduling-driven, scaled from the LFP color-class analysis). 
A trajectory costs $T_{\mathrm{traj}}^{\LMO} \approx 1.8\times10^8$ T-gates under the production fourth-order composition ($3.1\times10^8$ at order 2). 
Integrator analysis of Section~\ref{ssec:trotter_step} carries over unchanged, the LMO step being the
same symmetric composition. 
An anchor under the direct protocol costs $\approx 2.9\times10^{10}$ ($3.2\times10^{12}$ at the
finite-difference envelope).
A two-plateau $\VOCV$ curve ($60$ anchors, $30$ per plateau, spanning the $4.0$/$4.15$~V
plateaus and the transition region) costs
\begin{equation}\label{eq:lmo_curve}
T_{\VOCV}^{\LMO} \approx 60 \times 2.9\times10^{10}
\approx 1.7\times10^{12}~\text{T-gates}\quad(\text{30 anchors per plateau, direct protocol}),
\end{equation}
about $2.1\times$ the LFP curve at matched per-anchor precision.
The curve is an ensemble of $\sim\!10^4$ independent trajectories (60 anchors $\times$ $K\approx1.6\times10^2$ estimator repetitions; $\sim\!5\times10^3$ for the single-plateau $\LFP$ case), capable of being parallel across anchors and repetitions. 
Reported per-curve T-count is therefore a total operation count, not a sequential-depth requirement: given sufficient logical processors (QPUs) the wall-clock floor is a single trajectory's reaction-limited depth. 
Wall-clock is the reaction-limited logical depth at $t_{\mathrm{log}} =1~\mu$s, $\sim15$~h ($\sim7$~h for LFP). 
We do not quote a throughput-limited figure, since a sequential-depth floor cannot be beaten by adding factories (the factory bank instead enters the register footprint (Section~\ref{ssec:lfp_resources}). 
Asymptotic scaling adds only linear terms, $T^{\LMO}_{\mathrm{step}} \sim \mathcal{O}((K + L + L_{\mathrm{JT}}) N_{\mathrm{Trotter}})$ per Theorem~\ref{thm:poly_L}. 
A $\mathcal{O}(L^2)$ Coulomb pair count is conservative, tightening to $\mathcal{O}(L\,z(r_c))$ under screening-radius truncation (quasi-linear in system size). 
We do not use this optimization in our budgeting.

\subsection{Validation Status of LMO Costing and Forward Work}\label{ssec:lmo_validation}

Although the full multi-carrier production trajectory has not been run (requires fault-tolerant quantum hardware), resource costing of Section~\ref{ssec:lmo_resources} is produced with exact counts for the emitted circuit and are size-independent, and every block class appearing in the LMO step is either HEOM-validated (the open-system reservoir machinery, Section~\ref{ssec:lfp_validation}) or ED-validated (the active inter-site Coulomb, below). The remaining milestone is therefore the dynamical trajectory itself.

The LMO reference simulator (the Stage-1 instantiation of Eq.~\ref{eq:trotter_step_lmo}) passes baselines of norm and polaron/Li charge conservation to $<10^{-13}$, composed Strang ratio $3.99$, JT-doublet entanglement discrimination (24 orders of magnitude between $g_{\mathrm{JT}} > 0$ and $g_{\mathrm{JT}} = 0$), and Hamiltonian assembly reproducing the design targets ($E_{\mathrm{JT}} = 250$~meV,
$g/\omega_{\mathrm{ph}} = 1.8$, $\Omega_c = 140$~meV) exactly.
We exercise the multi-carrier sector, where a two-polaron configuration on three Mn sites conserves $N_{\mathrm{pol}} = 2$ to $10^{-13}$ and carries the finite (screened) adjacent-pair Coulomb energy that the single-carrier sector provably cannot ($\langle W_{ij}n_in_j\rangle \equiv 0$ at $N_{\mathrm{pol}} = 1$).

Beyond conservation, we focus to dynamically validate on a two-carrier, three-site instance, where the symmetric-Strang propagator reproduces the exact diagonalization of the active-Coulomb dynamics for a $U$-controlled observable (the adjacent-pair correlation $\sum_i\langle n_in_{i+1}\rangle(t)$) to $4.3\times10^{-4}$, converging at second order (ratio $3.99$), while that observable differs by ${\sim}40\times$ from the $U = 0$ trajectory.
The open-system machinery is therefore HEOM-validated (Section~\ref{ssec:lfp_validation}), the active-Coulomb blocks are the same $\le2$-local structural class as the HEOM-validated terms and the multi-carrier sector passes conservation and Strang-ratio conditions. 
We find the algorithm to operate correctly in the multi-polaron regime with this battery of tests~\cite{courtney2026multipolaron}, with full-scale dynamics and the two-plateau experiment as the declared milestone for future industry-scale fault-tolerant quantum computation.

We produce a $\VOCV$-versus-Ohzuku structural test through the following methods: (i)~An equilibrium model ladder (exact
grand-canonical enumeration plus finite-$T$ sector diagonalization), showing the active-Coulomb sector yielding a two-peak $\dQdV$ when the interaction is screened and physically coordinated. 
A bare flat-tail parameterization gives only a charging staircase.
(ii)~The split is derived parameter-free from spinel electrostatics: an exact Ewald two-species enumeration (Li$^+$ on 8a, polaron$^-$ on 16d) yields a convex hull with vertices at $\{0,4,8\}$ occupancy, exhibiting a first-order half-filling kink from crystallography and unit charges alone, reproducing the Madelung constant to seven digits. 
At $\varepsilon_\infty = 4.5$ a $300$~K split of $125$~mV is obtained (\ref{fig:lmo_twoplateau}; $-17\%$ vs Ohzuku's $150$~mV), within the single-cell/rigid-charge/scalar-$\varepsilon$ approximation. 
A better fit match selects $\varepsilon_{\mathrm{eff}} = 3.6$), and the derived interaction sits above the two-dimensional coexistence threshold ($W > 46$~meV), indicating true coexistence instead of a mean-field crossover.
This is a first-principles electrostatic prediction at the equilibrium tier (Section~\ref{ssec:scope}), not a quantum-dynamical simulation output. 
Finally, (iii)~at the gate-emitter level, the upgraded Trotter step carrying the screened-Coulomb blocks reproduces the observable through the measurement protocol (gate $\equiv$ classical diagonal to
$6\times10^{-16}$; Strang ratio $4.00$) as a circuit-construction check.

\begin{figure}[t]
\centering
\includegraphics[width=0.97\textwidth]{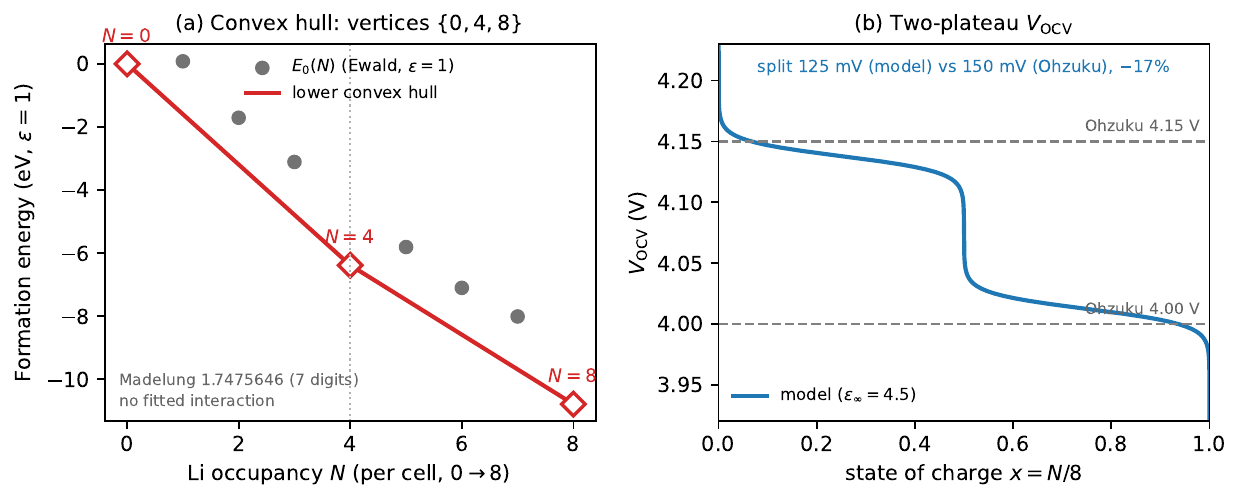}
\caption{First-principles LMO two-plateau structure. 
An equilibrium-tier electrostatic prediction (Section~\ref{ssec:scope}). \textbf{(a)} The Ewald lattice-gas convex hull of the formation energy $E_0(N)$ has vertices exactly at $\{0,4,8\}$ occupancy: a first-order half-filling kink fixed by spinel crystallography and unit
charges, with no fitted interaction (Madelung constant $1.7475646$ to seven digits). \textbf{(b)} The resulting $300$~K two-plateau $\VOCV$ at the electronic-floor dielectric $\varepsilon_\infty = 4.5$ against the Ohzuku plateaus ($4.00$/$4.15$~V): a $125$~mV split versus the experimental $150$~mV ($-17\%$; an exact match selects $\varepsilon_{\mathrm{eff}} = 3.6$).}
\label{fig:lmo_twoplateau}
\end{figure}

As it would require a robust, scaled fault-tolerant quantum computer, we do not yet validate the full multi-carrier dynamical production trajectory (the analog of Section~\ref{ssec:lfp_results} at LMO scale) beyond the single-cell statevector budget exercised to date; the $n_{\mathrm{JT}} = 10$ width's dedicated convergence sweep likewise remains outstanding.
Derived coupling places the ordering temperature just above $300$~K, consistent with LMO's known low-temperature transition~\cite{rodriguez1998electronic} (distinct from the Verwey charge ordering at full lithiation).
The resource estimates above are conditional on the same gate-count machinery validated at LFP and on the parameters of Table~\ref{tab:lmo_params}.

\section{Resource Comparison}\label{sec:resources}

\subsection{LFP versus LMO}\label{ssec:lfp_vs_lmo}

\begin{table}[h]
\centering
\caption{Resource comparison of the two specifications at matched SOC resolution. 
LFP figures include the index-optimized soft-Coulomb arithmetic (Eq.~\ref{eq:lfp_tstep}), while LMO figures are exact emitted-circuit counts (RESULT, Table~\ref{tab:scope}) from the same gate-count machinery (Section~\ref{sec:lmo_spec}); the $\VOCV$ \emph{prediction} itself is not yet production-validated.
The two qualitative rows at bottom: the inter-site Coulomb is dynamically active only in multi-carrier LMO, while the soft-Coulomb arithmetic ancilla exists only in continuous-coordinate LFP.
Trajectory-derived rows are quoted at the production fourth-order integrator point; the gate-validated
second-order baseline is larger by $\sim1.8\times$ (Table~\ref{tab:integrator_comparison}).
Direct-protocol curves use $30$ ($\LFP$) / $60$ ($\LMO$, two plateaus) jointly optimized anchors; the finite-difference envelope instead inherits the dense pre-optimization grid ($\sim\!200$ / $400$ trajectory evaluations, i.e.\ $\sim\!100$ / $200$ anchors over the two channels) that the direct protocol's joint anchor--shot optimization removes. 
The direct curve carries no separate two-channel factor because the Legendre-transform identity (Eq,~\ref{eq:vocv_combined}, residual $\le0.10\%$) ties the electron and Li$^+$ channels, so one $\mu$-sweep fixes both; the finite-difference route evaluates the two channels independently and therefore pays the $\times2$.}
\label{tab:lfp_lmo_resource}
\small\renewcommand{\arraystretch}{1.15}
\begin{tabular}{l r r r}
\toprule
 & LFP (single-pol.) & LMO (multi-pol.) & Ratio \\
\midrule
Logical qubits $n_{\mathrm{tot}}$ & $376$ & $368$ & $0.98\times$ \\
$|\mathcal{G}_2^{<}|$ bilinear pairs & $54$ & $70$ & $1.30\times$ \\
$R_{\mathrm{step}}$ (rotations) & $8.1\times10^3$ & $1.2\times10^4$ & $1.5\times$ \\
$T_{\mathrm{step}}$ (T-gates) & $\sim3.0\times10^5$ & $3.0\times10^5$ & $1.0\times$ \\
T-depth / step & $\sim8.5\times10^3$ & $\sim8\times10^3$ & $\sim1\times$ \\
$T_{\VOCV}$ (curve, direct protocol) & $\sim8\times10^{11}$ & $\sim1.7\times10^{12}$ & $2.1\times$ \\
$T_{\VOCV}$ (FD envelope, $K{=}1.8{\times}10^4$) & $\sim6.0\times10^{14}$ & $1.3\times10^{15}$ & $2.2\times$ \\
Wall-clock (depth-lim., $1~\mu$s) & $\sim7$~h & $\sim15$~h & $2.1\times$ \\
Physical qubits (15-to-1; \% factories) & $\sim1.4\times10^7$ (97\%) & $\sim1.5\times10^7$ (97\%) & $1.1\times$ \\
$\quad$(cultivation, optimistic) & $\sim1.8\times10^6$ & $\sim1.9\times10^6$ & $1.1\times$ \\
Inter-site Coulomb $U_{ij}$ & emitted, inert & \textbf{active} & --- \\
Soft-Coulomb $V_{\mathrm{LiP}}$ & arithmetic ($\sim$84q) & fixed CPhase & --- \\
Washboard $V_{\mathrm{wb}}$ & periodic UCR ($2^9$ low bits) & absent (in $t_{\mathrm{Li}}$) & --- \\
\bottomrule
\end{tabular}
\end{table}

The multi-polaron cell is not inherently more expensive, with its Jahn--Teller registers and extra bilinear pairs ($1.5\times$ on rotations) offset by discretizing away the soft-Coulomb ancilla, leaving per-step T-cost within $10\%$ of LFP and the qubit count $9\%$ smaller.
The end-to-end $2.2\times$ on the $\VOCV$ curve is dominated not by multi-polaron physics but by the two-plateau anchor doubling, contributing to measurement-protocol overhead.

\subsection{Carrier-Encoding Trade-offs: Continuous versus Discrete Li}
\label{ssec:encoding_tradeoffs}

A discrete-site (tight-binding) reduction is justified when the inverse barrier ratio $\eta^{-1} = k_BT/E_b$ and the localization ratio $\zeta =\sigma_{\mathrm{well}}/a_{\mathrm{hop}}$ are small, with $\sigma_{\mathrm{well}}= (2m_{\mathrm{Li}}\omega_{\mathrm{well}})^{-1/2}$ the zero-point width of Li in its site well. 
For LMO the 8a$\to$16c barrier is $E_b \approx 0.3$--$0.4$~eV~\cite{xu2010factors} (spanning the DFT/electrochemical-impedance-spectroscopy (EIS) range), giving $\eta \approx 12$--$15$ at $300$~K and $\zeta \approx 0.04$ ($\sigma \approx 0.08$~\AA{} against the $\approx 1.8$~\AA{} 8a--16c leg). 
Site occupancy is an excellent quantum number, and the thermal saddle weight $\sim e^{-\eta} \sim 10^{-6}$ bounds the solid-solution leakage a continuous register would resolve.
For LFP the washboard barrier $V_b = 0.15$~eV gives $\eta \approx 6$ and $e^{-\eta} \approx
2.5\times10^{-3}$ and the migration observables ($D_{\mathrm{Li}}$, $E_a$) are explicit targets, warranting a continuous coordinate. 

Both encodings are available for both materials, with the four-case $2\times2$ analysis tabulated in Appendix~\ref{tab:sm_encoding}.
We find the continuous-coordinate encoding (production for LFP) costs $376$ logical qubits and $3.0\times10^5$ T/step and resolves transport ($D_{\mathrm{Li}}$, $E_a$, the saddle-point Li--polaron binding), whereas the discrete-site encoding (production for LMO) drops to $\sim288$--$368$ qubits and $\sim2$--$3\times10^5$ T/step but makes the migration observables Marcus-rate inputs rather than
predictions. 

Headline occupancy thermodynamics ($\VOCV$, $\dQdV$) are encoding-insensitive.
Discrete reduction loses only intra-well vibrational free energy, which is largely absorbed by the voltage anchor.
What is lost in discrete encoding specifically is transport: barrier-crossing
dynamics ($D_{\mathrm{Li}}$, $E_a$ become Marcus-rate inputs rather
than predictions), the 16c intermediate with its correlated two-step
hops, and Li--polaron binding at the saddle-point. 
This introduces a polaron--vacancy-complex effect~\cite{asari2011formation} that shifts
migration barriers, being non-representable in a per-site $V_{\mathrm{LiP}}$ construction.
Third, the multi-carrier sector breaks the index-optimization of the arithmetic block, so a continuous-coordinate LMO pays the $L$-term arithmetic ($\sim9\times10^5$ T-gates/step at $\sim480$ qubits) (Table~\ref{tab:sm_encoding}). 
If first-principles LMO transport coefficients become the target, the configuration cost in the
$2\times2$ is an added expense for the encoding choice.

\subsection{Positioning Against the Battery Fault-Tolerant Literature}
\label{ssec:positioning}

Trotterization~\cite{suzuki1991general, berry2007efficient} is the alternative to randomized compilation~\cite{campbell2018random} and to qubitization/QSVT-based first-quantized phase estimation~\cite{low2017optimal, low2019hamiltonian, childs2012hamiltonian, gilyen2019quantum, martyn2021grand} typically used in quantum computational perspectives on battery chemistry. 
One series targets static ground-state energetics in equilibrium voltage and stability from qubitized quantum phase estimation (QPE)~\cite{delgado2022simulating, zini2023quantum}, as well as spectroscopic response via time-domain Green's-function methods~\cite{fomichev2024simulating,fomichev2025fast, kunitsa2025quantum, loaiza2026quantum}. 
These are complementary regimes in the same global target, standing in contrast to \emph{dynamic} open-system observables, computed here.
We have yet to see fault-tolerant estimates for dynamical reservoir-equilibration observables of multiple Holstein polarons, motivating this paper. 
The static-energetics/spectroscopic series results places this work in the same logical-qubit ballpark.
We note explicitly that we do not validate the algorithms presented in prior work, however, the match in register size is direct: our $376$ ($\LFP$) and $368$ ($\LMO$) logical qubits sit alongside $100$ for $18$-orbital Li$_2$MnO$_3$ electron-energy-loss-spectroscopy (EELS)~\cite{kunitsa2025quantum}, $414$ for $20$-orbital Li-excess resonant-inelastic-X-ray-scattering (RIXS)~\cite{loaiza2026quantum}, and $100$ per X-ray-absorption-spectroscopy (XAS) circuit~\cite{fomichev2025fast}.
At matched granularity (per full observable) and in a consistent currency (T-gates, converting Toffoli at $\sim4{:}1$): our per-curve $\sim8\times10^{11}$ ($\LFP$)/$1.7\times10^{12}$ ($\LMO$) T exceeds
$3.25\times10^8$~T~\cite{kunitsa2025quantum}, $\sim8\times10^{10}$~T ($2.0\times10^{10}$ Toffoli)~\cite{loaiza2026quantum}, and $\lesssim1.6\times10^9$~T ($<4\times10^8$ Toffoli)~\cite{fomichev2025fast}. 
This is expected for a full dynamical curve rather than a single static or spectroscopic circuit.
NISQ variational quantum eigensolver (VQE) cathode studies~\cite{farag2022towards} occupy the complementary pre-fault-tolerant regime. 
State-preparation costs folded into those estimates~\cite{fomichev2024initial} correspond here to the UCC system preparation plus thermofield bath chains, the latter carrying zero cost (their doubled register and per-step evolution are already counted in the step budget). 

\subsection{Measurement Protocols: \texorpdfstring{$\VOCV$}{VOCV}, \texorpdfstring{$\dQdV$}{dQdV}, and \texorpdfstring{$Z(\omega)$}{Zomega}}
\label{ssec:measurement_protocols}

For the electrochemist reader we state the operational recipes explicitly, leaving full derivations and the remaining observable ports ($i_0$,$D_{\mathrm{Li}}$, $E_a$, $T_1/T_2$) in Appendix~\ref{sm:tables}.

\paragraph{\texorpdfstring{$\VOCV$}{VOCV} and SOC.}
At each anchor the control pair $(\mu, \beta)$ is fixed, with a grand-canonical state prepared as a product of two pieces. 
The first is the bath thermal state is the Tamascelli doubled chains in vacuum at the set $(\mu, \beta)$, being exact and zero-gate, the fermionic finite-$T$ construction of Section~\ref{ssec:chain_construction}. 
The second is the system electronic/vibrational reference (polaron occupancy at the target filling, Li wavepacket as a grid Gaussian, phonons displaced to the Lang--Firsov minimum) is the $T_{\mathrm{init}} \approx 7\times10^4$-T UCC ansatz. 
The reduced equilibrium that the partial trace returns is the mean-force Gibbs state of  Section~\ref{ssec:lfp_observables}, so ``thermal state'' refers to that bath-dressed reduced state.
By Hellmann--Feynman the electrode potential is the reservoir chemical potential set, $V = \mu/e$: there is no estimator error on the voltage axis.
One computational-basis measurement of the polaron register per shot returns an integer $N^{(s)} = \sum_i n_i^{(s)}$ (the number operator is diagonal), and the SOC abscissa is its mean over $n_{\mathrm{sh}}$ shots. 
The $\VOCV(\SOC)$ curve is generated parametrically by sweeping $\mu$ without finite difference 
(Algorithm~\ref{alg:anchor_protocol}).

\paragraph{\texorpdfstring{$\dQdV$}{dQdV}.}
The same shot record yields the connected charge fluctuation by Welford variance (no large-number cancellation), and the static fluctuation--dissipation estimate $\dQdV|_{\mathrm{static}} = e^2\beta\,\widehat{\mathrm{Var}}(N)$ is a zero-marginal-cost byproduct of the population measurement. 
We therefore require no extra circuit, ancilla, or shots, accurate to the $\approx 2\%$ Kubo--Mori systematic error. 
FWHM-grade exactness adds the short imaginary-time correction, $\dQdV = e^2\beta\langle\delta N;\delta N\rangle = e^2\int_0^\beta d\tau\, C(\tau)$ with $C(0) = \mathrm{Var}(N)$. 
Since $C(\tau)$ is smooth and Kubo--Martin--Schwinger (KMS)-symmetric, $M_t \sim 2$--$3$ Gauss--Legendre nodes reach $< 0.1\%$, each node one Hadamard-test correlator circuit (one control ancilla $+$ controlled evolution, $\approx 2\times$ a trajectory) or an ancilla-free imaginary-time evaluation. 
Estimator statistics: $n_{\mathrm{sh}} \sim 10^2$--$10^3$ per anchor recovers the FWHM to $\sim10\%$ in direct-sampling mode, while the $1.5$~mV QAE-grade target instead pays $K_{\mathrm{QAE}}$ (Section~\ref{ssec:measurement_bottleneck}).
$N$ fluctuates (grand-canonical exchange with the lead), which the chain-mapped
reservoirs provide by construction. 
In a closed fixed-$N$ sector
$\mathrm{Var}(N) \equiv 0$ and the route is inapplicable.

\begin{algorithm}
\caption{Per-anchor measurement protocol ($\VOCV$, SOC, $\dQdV$)}
\label{alg:anchor_protocol}
\begin{algorithmic}
\REQUIRE anchors $\{\mu_a\}$ ($M_{\mathrm{anchor}} \sim 30$, uniform in SOC); shots $n_{\mathrm{sh}} \sim 10^2$--$10^3$; optional quadrature depth $M_t \sim 3$
\FOR{each anchor $\mu$}
    \STATE Prepare $\rho(\mu, \beta)$: UCC system state; Tamascelli chains in vacuum
    \STATE Measure polaron register in computational basis $n_{\mathrm{sh}}$ times; record bitstrings $\Rightarrow N^{(s)}$ exact integers
    \STATE Emit $V = \mu/e$ (exact); $\SOC = \overline{N}/N_{\max}$; $\dQdV|_{\mathrm{static}} = e^2\beta\,\mathrm{Var}_{\mathrm{Welford}}(N)$
    \IF{FWHM-grade required}
        \STATE Evaluate $C(\tau_k)$ at $M_t$ Gauss--Legendre nodes (Hadamard test or imaginary-time); $\dQdV = e^2\sum_k w_k C(\tau_k)$
    \ENDIF
\ENDFOR
\STATE Assemble $\VOCV(\SOC)$, $\dQdV(V)$; peak voltage to $\pm5$~mV; FWHM by half-max bracket
\end{algorithmic}
\end{algorithm}

\paragraph{\texorpdfstring{$Z(\omega)$}{Zomega}}
A textbook protocol to extract impedance involves a sinusoidal modulation of the electron-reservoir
potential, $\mu_e(t) = \mu_e^* + \Delta\mu\cos\omega t$ with $\Delta\mu
\lesssim k_BT/e$, steady-state current readout $I(t) =
-e\,d\langle n_b\rangle/dt$, and $Z(\omega) = (\Delta\mu/I_0)
e^{-i\phi}$. 
This is physically faithful but operationally infeasible on the
quantum register across the EIS band. 
At $\omega = 1$~MHz a single period is $5\times10^{12}$ Trotter steps ($\sim10^{18}$ T-gates
for each frequency point).
Instead, a $200$~fs trajectory with Liouvillian-resonance QPE on the $K_{\mathrm{eff}} \sim 30$ chain modes extracts the pole content, and $Z(\omega)$ at any $\omega \in [10^{-3}, 10^{13}]$~Hz follows as
a closed-form Lehmann sum in classical post-processing, giving zero quantum cost per frequency point. 
Each chain mode contributes a Lorentzian pseudomode, summing to the canonical Randles picture.
The charge-transfer semicircle is set by the rate-limiting reservoir because modulation drives the electron reservoir and the readout is electron-transfer current. 
This is the lead ($\gamma_e^{-1} \approx 20$~ps) rather than the faster Li interfacial exchange ($\gamma_{\mathrm{Li}}^{-1} \approx 2$~ps): the electron-limited resistance dominates by $R_{\mathrm{ct}}^{e}/R_{\mathrm{ct}}^{\mathrm{Li}} = \gamma_{\mathrm{Li}}/\gamma_e \approx 10$, peaking near $\gamma_e$. 
This $R_{\mathrm{ct}}$ is thermally activated $R_{\mathrm{ct}} \propto
k_BT/[\gamma_e f(\varepsilon_d)(1-f(\varepsilon_d))]$, with $f$ the lead Fermi function at the redox level $\varepsilon_d$ offset by the activation energy $E_a$, diverging as $T\to 0$. 
As such, a finite-$T$ fermionic chain is required, since a zero-temperature Landauer lead cannot reproduce it.
The low-frequency Warburg $\omega^{-1/2}$ tail is the Li chain's quasi-continuum (solid-state Li diffusion), as before. 

Per-observable measurement costs are tabulated in Appendix~\ref{tab:sm_per_observable}.
We show that measurement still dominates the total simulation cost, with a full $\VOCV$/$\dQdV$ curve costing $\sim8\times10^{11}$ T-gates at $K \approx 1.6\times10^2$ coherent repetitions per anchor ($\sim6\times10^{14}$ in the worst-case finite-difference envelope). 
The voltage axis is free by Hellmann--Feynman, the screening-grade $\dQdV$ is a byproduct of the $\VOCV$ record, the FWHM-grade $\dQdV$ adds $\sim3$ correlator circuits per anchor, and $Z(\omega)$ adds about one further trajectory.
The end-to-end cost therefore reduces to how precisely each anchor must be estimated, the subject of Section~\ref{ssec:measurement_bottleneck}.

\subsection{The Measurement Bottleneck}\label{ssec:measurement_bottleneck}

Under the direct protocol of Section~\ref{ssec:measurement_protocols} the binding per-anchor requirement (the $\SOC$-axis shape target, $\delta\SOC = 0.02$) is $K \approx 1.6\times10^2$ Heisenberg-limited coherent repetitions~\cite{brassard2000quantum}, or $n_{\mathrm{sh}} \approx 2.5\times10^3$ incoherent shots. 
Under alternative estimation strategies, an anchor pays $\max(n_{\mathrm{sh}}, K)$ trajectory repetitions. The production per-anchor cost is therefore $\approx 2.7\times10^{10}$ ($\LFP$) to $2.9\times10^{10}$ ($\LMO$) T-gates against per-trajectory costs of $\sim2\times10^8$.
If we used a finite-difference-on-voltage instead, we would need $K_{\mathrm{QAE}} = \lceil 1/\varepsilon\rceil \approx 1.8\times10^4$ at the $1.5$~mV target. 
We retain this figure as the conservative envelope ($\approx 3.1\times10^{12}$ per anchor) wherever worst-case expectations are the point. 
Amplitude estimation is the universal price of precision, which we quote separately according to the resource-estimation genre~\cite{reiher2017elucidating, babbush2018encoding, lee2021even}.

The highest-leverage optimization is to change what is measured, which we leverage in Section~\ref{ssec:measurement_protocols}.
A baseline that finite-differences a sampled curve amplifies estimator noise by a condition number $\kappa_{\mathrm{FD}} = 1/(h\,s)$ is removed by the direct Hellmann--Feynman/Kubo--Mori readout~\cite{falk1969susceptibility, babbush2018encoding, ruppeiner1995riemannian}.
Classically emulated benchmarks realize $\kappa_{\mathrm{FD}} \approx 2.5\times10^2$--$3\times10^3$ for $\dQdV$, and joint anchor$\times$shot optimization reduces the $\sim400$-anchor sweep to
the $\sim30$ anchors of Algorithm~\ref{alg:anchor_protocol}.  
This converts the per-curve cost from measurement-overhead bound to trajectory-count bound, and can be applied without alteration to LMO's two plateaus.

\paragraph{The plateau condition number.}
Does the control-swap simply relocate a hidden amplification? 
With $V = \mu/e$ exact and $\SOC$ estimated, the plateau flatness $dV/d\SOC \to 0$ could be feared to demand $\langle N\rangle$ at a precision amplified by the inverse slope. 
Propagating errors through the production curve (logistic crossing calibrated to the converged outputs: $\dQdV$ FWHM $12.6$~mV, plateau span $\Delta\SOC = 0.8$) shows the amplification attaches to the inverse inference only.
Inferring $\SOC$ from a measured voltage would pay $\kappa_{\mathrm{inv}} =V_{\mathrm{scale}}/|dV/d\SOC|_{\mathrm{peak}} \approx 56\times$, eliminated by the present construction.
With voltage as the control, the propagation runs the other way, with $\delta V^* = \delta\SOC \cdot |dV/d\SOC|_{\mathrm{peak}}$ with $|dV/d\SOC|_{\mathrm{peak}} \approx 18$~mV per unit SOC. 
Plateau flatness thereby attenuates $\SOC$ noise on the voltage axis and the $1.5$~mV peak-position grade needs only $\delta\SOC \approx 0.08$, i.e.\ $n_{\mathrm{sh}} \approx 26$ shots at the grand-canonical fluctuation $\mathrm{Var}(N) \approx 11.6$ the plateau itself produces. 
The binding per-anchor requirement is instead the $\dQdV$ width: a $4\%$ FWHM grade needs $n_{\mathrm{sh}} \approx 2.5\times10^3$ incoherent shots or $K \approx 40$--$160$ coherent repetitions. 
The $K_{\mathrm{QAE}} \approx 1.8\times10^4$ envelope retained alongside the budgets is therefore the worst case inherited from the finite-difference-on-voltage route. 
Production figures we calculate use the direct-protocol $K \approx 1.6\times10^2$, tighter by two orders of magnitude.
$\SOC$ noise admittedly amplifies on the voltage axis is the steep solid-solution shoulders ($|dV/d\SOC| \approx 500$~mV per unit SOC), where $\delta\SOC = 0.02$ maps to $\sim10$~mV apparent scatter. 
This is acceptable against the $\sim50$~mV shoulder features, though we caveat our numbers accordingly.

\section{Complexity Comparison}\label{sec:quantum_advantage}

Table~\ref{tab:complexity_comparison} positions the present approach against the classical methods that strain in the conjunctive cathode regime.

\begin{table}[h]
\centering
\caption{Complexity scaling at the cathode-class problem (strong-coupling Holstein, three reservoirs, finite $T = 300$~K). $L$ is the number of metal sites; $K$ chain length; $D$ DMRG bond dimension;
$L_h$ HEOM hierarchy depth; $f$ mode count.}
\label{tab:complexity_comparison}
\small
\begin{tabular}{lll}
\toprule
Method & Asymptotic cost & Cathode-regime verdict \\
\midrule
DFT (KS) & $\mathcal{O}(N_e^3)$ per nuclear step &
  No reservoirs: cannot reach $\VOCV$ \\
DMRG (1D) & $\mathcal{O}(LD^3)$/sweep; $D$ grows with $S$ &
  Finite-$T$ via purification~\cite{feiguin2005finite}; $D \sim 10^4$
  at cathode $S$, \\ & & compounding with three reservoirs \\
HEOM & $\mathcal{O}(L_h^{N_{\mathrm{bath}}})$ tier explosion &
  Strong coupling, finite $T$, \emph{one} structured
  bath~\cite{tanimura2020numerically};\\ & &  tier count explodes with three \\
ML-MCTDH & $\mathcal{O}(d^{f-1})$ basis growth &
  $f \sim 100$ modes at short times~\cite{wang2003multilayer};
  \\ & & three-reservoir steady state at the practical edge \\
KMC & fitted rate equations &
  No coherent polaron dynamics \\
\textbf{This work ($\LFP$)} & $\mathcal{O}((K + L)N_{\mathrm{step}})$,
  poly($L$)/step (Thm.~\ref{thm:poly_L}) &
  Tractable at $K = 10$, $L = 8$ on FT hardware; \\ & & $N_{\mathrm{step}}$
  set by the integrator order (Table~\ref{tab:integrator_comparison}) \\
\textbf{This work ($\LMO$)} & $\mathcal{O}((K + L +
  L_{\mathrm{JT}})N_{\mathrm{step}})$; $+2L_{\mathrm{JT}}$ JT registers &
  Multi-polaron; Coulomb active; no arithmetic block \\
\bottomrule
\end{tabular}
\end{table}

Each classical method excels in a different sub-regime, each pushed substantially. 
We do not claim that these methods fail, but that none simultaneously spans the features the cathode problem combines, being (i) strong polaron coupling ($S \gtrsim 1$), (ii) finite temperature, (iii) three independent reservoirs, and (iv) reservoir-equilibrated electrochemical observables. 
HEOM is the exact reference on one structured bath and is used as such, but its tier count explodes before the three-reservoir production problem is reached, which is precisely where the chain-mapped circuit continues at $\mathcal{O}((K{+}L)N_{\mathrm{step}})$ per trajectory.
The quantum framework is, to our knowledge, the only approach extending along all four axes simultaneously at controllable error. 
FT wall-clock estimates of $\sim7$--$15$~h (reaction-limited depth, production fourth-order integrator and direct measurement protocol, longer under the conservative finite-difference envelope) per full $\VOCV$ curve (Table~\ref{tab:lfp_lmo_resource}) give a conditional advantage claim on the measurement-bottleneck levers of Section~\ref{ssec:measurement_bottleneck} rather than on circuit costs.

\section{Discussion}\label{sec:Discussion}

We verify convergence on the named observable, being the $\dQdV$ peak against the Yamada plateau, including its chain-length dependence.
Measurement cost is quoted explicitly and separately ($K \approx 1.6\times10^2$ direct-protocol; $K_{\mathrm{QAE}} \approx 1.8\times10^4$ finite-difference envelope), and infeasible direct routes
(AC-modulation impedance, mean-squared-displacement diffusion) are re-architected into the same trajectory infrastructure. 
The electrochemical result is guarded by two independent algebraic identities (two-channel vs Legendre $\VOCV$ at $0.36\%$ sum-rule residual; independent finite-difference peak
extraction), complementing quality control and the independent exact diagonalization check.

The single anchor fixes $\VOCV(\SOC{=}0.5) \equiv 3.45$~V by construction. 
The predictions give peak voltage within $\pm 5$~mV, FWHM within the factor-of-2 band, plateau extent and shoulders, holding across the $\pm20\%/\pm40\%$ DFT-uncertainty and Holstein sweeps (Section~\ref{ssec:lfp_robustness}) and making the result a structural prediction conditional on Table~\ref{tab:lfp_params}.
Structural assumptions (1D channel, Holstein form, three-reservoir decomposition) are LFP-specific, and the single-cell Hamiltonian represents the two-phase coexistence coarse-grained.
Structural fidelity is tested with LMO, whose end-to-end costing we deliver here (Section~\ref{sec:lmo_spec}); the production dynamical run and the quantum-dynamical extraction of the two-plateau $\VOCV$-versus-Ohzuku comparison are the remaining milestones.

The same chain-mapped substrate recovers the Leggett macroscopic-quantum-coherence double-well~\cite{leggett1984quantum, leggett1987dynamics} and the two-band BCS $T_c$
enhancement~\cite{suhl1959bardeen, mcmillan1968transition, choi2002origin} unmodified, with the MSW neutrino resonance~\cite{wolfenstein2018neutrino, mikheyev1988neutrino} as an algebraic two-channel check. 
These are portability demonstrations, indicating the reach of the multi-polaron open-system primitive beyond cathodes (Section~\ref{sec:Introduction}, Figure~\ref{fig:universality}).

\begin{figure}[t]
\centering
\includegraphics[width=0.97\textwidth]{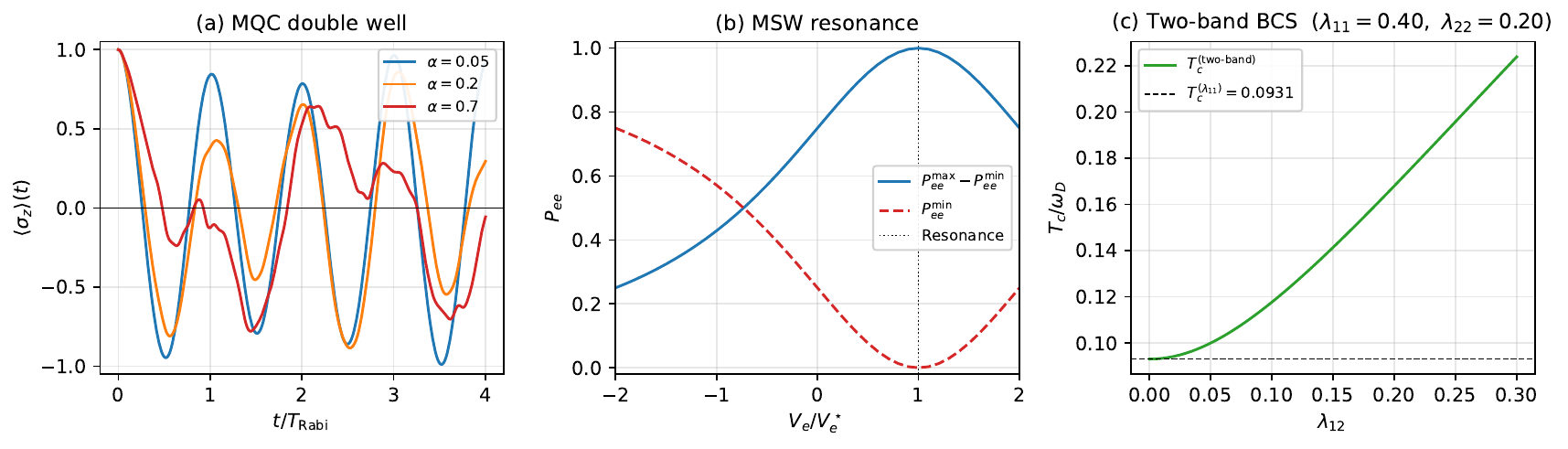}
\caption{The same chain-mapped open-system primitive, when evolved unitarily and read out by partial
trace, reproduces the canonical signatures of three further problems without modification: \textbf{(a)} the Leggett macroscopic-quantum-coherence crossover from coherent to localized $\langle\sigma_z\rangle(t)$; \textbf{(b)} the Mikheyev--Smirnov--Wolfenstein resonance dip; and \textbf{(c)} the two-band BCS interband $T_c$ enhancement.}
\label{fig:universality}
\end{figure}

We leave solid--electrolyte-interphase (SEI) and interfacial chemistry, long-time capacity fade, sweep-rate-dependent voltammetry, the 3D fermion-sign problem, and non-Ohmic bath shapes (the Tamascelli construction admits arbitrary bosonic $J(\omega)$; benches to date are Ohmic with hard
cutoff) to future work. 
Each of these poses a deliberate model boundary, inclusion brought by extending the system Hamiltonian or a kinetic-Monte-Carlo hand-off.
We leave per-trajectory $D_{\mathrm{Li}}(\SOC)$ and $Z(\omega;\SOC)$ extraction on the existing infrastructure to forward work, along with an Arrhenius sweep for $E_a$, an optional second numerically exact reference (matrix-product-state / time-evolving-block-decimation, MPS/TEBD) at the strongest couplings, the LMO validation program of Section~\ref{ssec:lmo_validation}, and hardware-scale demonstration of the reversible-arithmetic block.

\section{Conclusion}\label{sec:Conclusion}

We have presented a fault-tolerant-compatible quantum algorithm for nonadiabatic multi-polaron transport as an open system, using a unitary chain-mapped Caldeira--Leggett embedding with exact finite-temperature doubling, a fully specified Clifford$+$T architecture verified gate-level at $\|\Delta\psi\|_\infty < 10^{-12}$, a per-step budget dominated by the reversible $1/\sqrt{\cdot}$ arithmetic of the soft Coulomb (the one non-polynomial diagonal), with the diagonal multiplexing kept polynomial in system size (Theorem~\ref{thm:poly_L}). 
For single-polaron $\LFP$, a single voltage anchor yields the Yamada-plateau $\dQdV$ peak within $\pm5$~mV and the FWHM within the factor-of-2 band, with plateau extent and shoulders as structural predictions, cross-validated against numerically exact HEOM precisely where the master-equation description fails. 
Production cost is $376$ logical qubits and $\sim3.0\times10^5$ T-gates per step.
For multi-polaron $\LMO$, where the $1{:}1$ Mn$^{3+}$/Mn$^{4+}$ filling activates the inter-site Coulomb repulsion, we deliver the explicit Trotter step and exact emitted-circuit operation counts with $368$ logical qubits, $\sim3.0\times10^5$ T-gates per step, and the end-to-end two-plateau $\VOCV$ curve at $\sim1.7\times10^{12}$ T-gates ($\sim15$~h reaction-limited), giving the full-scale multi-polaron dynamical trajectory to be the declared milestone.
Across both, all trajectory budgets are quoted at the production fourth-order integrator point, with the gate-validated second-order baseline and the qubitization crossover quantified in Section~\ref{ssec:trotter_step}. 
The bottleneck cost is in per-anchor estimation, $K \approx 1.6\times10^2$ coherent repetitions under the direct readout of thermodynamic derivatives (the finite-difference envelope $K_{\mathrm{QAE}} \approx
1.8\times10^4$ bounds worst-case expectations). 
This converts the per-curve cost from measurement-bound to trajectory-bound, putting a full $\VOCV$ curve at $\sim8\times10^{11}$ T-gates. 
Every benchmark is a statevector-exact classical simulation of drastically reduced instances ($\le13$ qubit-equivalents with up to $L=4$ Holstein sites, the two-site dimer, the single spin-boson), while the production figures are exact counts and closed-form estimates of the emitted $292$-qubit circuit.

This work is a required step on the path to applicable materials chemistry being simulated on quantum computers. 
Real-time evolution of multi-polaron transport is now adapted to a fault-tolerant quantum circuit with reservoir-equilibration observables made directly extractable, rather than left implicit in a closed Hamiltonian.
The same chain-mapped kernel recovering cathode electrochemistry carries unmodified to other strongly correlated open-system Hamiltonians of the same algebraic form, from the Leggett macroscopic-quantum-coherence double-well~\cite{leggett1984quantum,leggett1987dynamics}, the two-band BCS superconductor~\cite{suhl1959bardeen,mcmillan1968transition,choi2002origin}, the Mikheyev--Smirnov--Wolfenstein neutrino resonance~\cite{wolfenstein2018neutrino,mikheyev1988neutrino}, and more broadly the polaronic and vibronic materials class. 
Progress on the battery instance can be shared progress for each.
Our impetus for this research is eventual materials discovery, from better cathodes and battery design to the wider family of strongly correlated open-system materials these Hamiltonians describe. 
Although we do not prove quantum speedup beyond complexity calculations, we give a gate-level algorithm with closed-form production costs on which the full-scale dynamical runs can be mounted once the hardware exists.

\section{Data Availability}
The simulation, gate-emission, and resource-estimation code is available on Github~\cite{courtney2026multipolaron}.

\appendix
\section{Higher-Order Product Formulas and the Qubitization Crossover}
\label{sm:integrators}

This section derives fourth-order product-formula compositions and qubitization/quantum singular value transformation (QSVT) query cost quoted in the main text, and locates the precision crossover between them at the LFP production operating
point.

\subsection{Fourth-order symmetric compositions}
\label{sm:fourth_order}

Write the main text's second-order Strang step as $S_2(\Delta t)$. 
Two standard symmetric fourth-order compositions exist, both built entirely from $S_2$ substeps and therefore free in the existing gate primitives. The {Yoshida triple-jump}~\cite{yoshida1990construction}
is the three-substep composition
\begin{equation}
\label{eq:sm_yoshida}
S_4^{\mathrm{Y}}(\Delta t)
 = S_2(w_1\Delta t)\; S_2(w_0\Delta t)\; S_2(w_1\Delta t),
\qquad
w_1 = \frac{1}{2 - 2^{1/3}} \approx 1.351207,
\quad
w_0 = -\frac{2^{1/3}}{2 - 2^{1/3}} \approx -1.702414,
\end{equation}
with $w_0 + 2w_1 = 1$. The {Suzuki fractal}~\cite{suzuki1991general} is the five-substep composition
\begin{equation}
\label{eq:sm_suzuki}
S_4^{\mathrm{S}}(\Delta t)
 = S_2(p\Delta t)^2\; S_2\bigl((1 - 4p)\Delta t\bigr)\; S_2(p\Delta t)^2,
\qquad
p = \frac{1}{4 - 4^{1/3}} \approx 0.4144908,
\quad
1 - 4p \approx -0.6579631 .
\end{equation}
Both contain one negative-coefficient substep. 
For unitary evolution a backward Strang substep is simply the adjoint circuit (reversed layer
order, negated rotation angles), so each substep remains individually
Phase-QC certifiable. 
The main text adopts the Suzuki form. 
Its smaller coefficients give substantially smaller fifth-order error constants than Yoshida's $|w_0| \approx 1.70$ (which enters the local error as $|w_0|^5 \approx 14$).
Yoshida form trades a $3\times$ per-step multiplier for larger constants and is supported by the same emitters.

The error model follows the commutator framework of Childs et al.~\cite{childs2021theory}: a symmetric order-$2k$ composition obeys $\|S_{2k}(\Delta t) - e^{-iH\Delta t}\| \le W_{2k}\,\Delta t^{2k+1}$
per step, giving step counts over an evolution window $T$
\begin{equation}
\label{eq:sm_stepcounts}
k_2(\epsilon) = \Bigl(\frac{W_2\,T^3}{\epsilon}\Bigr)^{1/2},
\qquad
k_4(\epsilon) = \Bigl(\frac{W_4\,T^5}{\epsilon}\Bigr)^{1/4},
\qquad
T_{\mathrm{traj}}^{(2)} = k_2\,T_{\mathrm{step}},
\quad
T_{\mathrm{traj}}^{(4)} = 5\,k_4\,T_{\mathrm{step}} .
\end{equation}
Rather than re-deriving nested-commutator norms, we calibrate $W_2$ to the main text's own dominant per-step bound (the Jordan--Wigner channel: $4.6\times10^{-6}$ at the production timestep $\Delta t = 8.27$~a.u.), giving $W_2 = 8.1\times10^{-9}$~a.u., and estimate the fourth-order constant by the standard one-scale escalation $W_4 \approx W_2\,\Lambda^2$ with $\Lambda = \omega_{\mathrm{ph}}$ the dominant local energy. 
This gives an order-of-magnitude estimate, labelled as such. 
Qualitative conclusions below are insensitive to an order of magnitude in $W_4$ because of the $1/4$-power in Equation~\ref{eq:sm_stepcounts}.

\subsection{Block encoding and QSVT query cost}
\label{sm:qsvt}

The chain-mapped Hamiltonian is a sum of $\mathcal{O}(L^2 + K)$ terms, $d$-sparse with efficiently computable entries: an inexpensive linear-combination-of-unitaries block encoding~\cite{childs2012hamiltonian, low2019hamiltonian}. 
The QSVT/quantum signal processing (QSP) implementation of $e^{-iHT}$~\cite{low2017optimal, gilyen2019quantum,martyn2021grand} requires
\begin{equation}
\label{eq:sm_queries}
Q(\epsilon) \;\approx\; \alpha\,T + c\,\ln(1/\epsilon),
\qquad c \approx 1.4,
\end{equation}
queries to the block encoding (Jacobi--Anger truncation), where $\alpha = \sum_j |h_j|$ is the LCU one-norm. 
Bosonic operators enter the LCU through quadrature/ladder decompositions whose coefficients carry the register amplitude factors ($\sqrt{2N_{\mathrm{ph}}}$ per quadrature, $N_{\mathrm{ph}}$ per number operator), so $\alpha$ inflates with both the coupling strength and the Fock occupancy. 
At the LFP production parameters (Table~\ref{tab:sm_alpha}) the one-norm is $\alpha = 2.74$~Ha, dominated by the chain on-site energies and the bosonic amplitude factors, being the strong-coupling structure that makes the chemistry interesting.

\begin{table}[h]
\centering
\caption{LCU one-norm assembly at the LFP production point
($L = 8$, $K = 10$ per reservoir, $N_{\mathrm{ph}}^{B} = 22$,
$N_{\mathrm{ph}}^{D} = 16$, $\alpha_{\mathrm{bath}} = 0.2$).}
\label{tab:sm_alpha}
\begin{tabular}{l r}
\toprule
LCU group & contribution to $\alpha$ (Ha) \\
\midrule
Jordan--Wigner (JW) hops $2(L{-}1)\,t_{\mathrm{hop}}$ & $0.013$ \\
Coulomb $\sum_{i<j} U_{ij}$ & $0.253$ \\
Holstein $L\,g\sqrt{2N_{\mathrm{ph}}^B}$ & $0.190$ \\
System phonons $L\,\omega_{\mathrm{ph}}(N_{\mathrm{ph}}^B{+}\tfrac12)$ & $0.430$ \\
Chain on-site $3K\,\bar e_k (N_{\mathrm{ph}}^D{+}\tfrac12)$ & $1.092$ \\
Chain hops $3(K{-}1)\,\bar t_k N_{\mathrm{ph}}^D$ & $0.476$ \\
$V_{SB}$ bilinears $3\,c_0\sqrt{4N^B N^D}$ & $0.222$ \\
Washboard $+$ $V_{\mathrm{LiP}}$ diagonals & $0.062$ \\
\midrule
\textbf{Total} $\alpha$ & $\mathbf{2.737}$ \\
\bottomrule
\end{tabular}
\end{table}

Each query (one \textsc{select}$+$\textsc{prepare} pass over the LCU) runs the same diagonal, bilinear, and arithmetic emitters once, so we charge it conservatively at one Strang step's T-count, $T_{\mathrm{query}} \approx T_{\mathrm{step}} = 3.0\times10^5$. 
The QSVT trajectory cost is then floored by the $\epsilon$-independent term:
$\alpha T = 2.74 \times 7.1\times10^3 \approx 1.9\times10^4$ queries $\approx 5.8\times10^9$ T-gates at the $T_{\mathrm{evol}} = 172$~fs recurrence window, being independent of precision until the logarithm competes.

\subsection{The crossover}
\label{sm:crossover}

\begin{table}[h]
\centering
\caption{Per-trajectory T-counts versus target propagator error $\epsilon$ at the LFP production point ($T_{\mathrm{evol}} = 172$~fs, $T_{\mathrm{step}} = 3.0\times10^5$ T-gates).
Best value per row in bold; the crossover $\epsilon^* \approx 4\times10^{-9}$ falls between
the last two rows.}
\label{tab:sm_crossover}
\begin{tabular}{r r r r r r}
\toprule
$\epsilon$ & $k_2$ & Strang-2 & $k_4$ & Suzuki-4 & QSVT \\
\midrule
$5\times10^{-3}$ & $765$ & $2.3\times10^{8}$ & $114$ & $\mathbf{1.7\times10^{8}}$ & $5.8\times10^{9}$ \\
$10^{-4}$ & $5407$ & $1.6\times10^{9}$ & $303$ & $\mathbf{4.6\times10^{8}}$ & $5.8\times10^{9}$ \\
$10^{-6}$ & $54074$ & $1.6\times10^{10}$ & $958$ & $\mathbf{1.4\times10^{9}}$ & $5.8\times10^{9}$ \\
$10^{-8}$ & $5.4\times10^{5}$ & $1.6\times10^{11}$ & $3031$ & $\mathbf{4.6\times10^{9}}$ & $5.8\times10^{9}$ \\
$10^{-9}$ & $1.7\times10^{6}$ & $5.1\times10^{11}$ & $5390$ & $8.1\times10^{9}$ & $\mathbf{5.8\times10^{9}}$ \\
\bottomrule
\end{tabular}
\end{table}

Setting $5\,k_4(\epsilon)\,T_{\mathrm{step}} = Q(\epsilon)\,T_{\mathrm{query}}$ locates the crossover at
$\epsilon^* \approx 4\times10^{-9}$: qubitization overtakes the fourth-order product formula when the target propagator error is driven below $\sim\!10^{-9}$. 
The mechanism is the one-norm: the QSVT query count is additive in $\ln(1/\epsilon)$ but multiplicative in
$\alpha T$, and at strong coupling $\alpha$ carries the $\lambda\sqrt{N_{\mathrm{ph}}}$ amplitude inflation of Table~\ref{tab:sm_alpha}.
Weakly coupled or few-boson instances (smaller $\alpha$) shift $\epsilon^*$ upward toward practical relevance, which is why the main text reports the crossover rather than a categorical preference.
The fourth-order point, by contrast, dominates the second-order baseline at every precision of interest, at the cost of one negative-time Strang substep per composition step.

\section{Coupling Convention, Lang--Firsov, and Register Sizing}
\label{sm:convention}

The Holstein and Jahn--Teller couplings are emitted as $\exp[-i\,g\sqrt{2}\,q\,n\,\tau]$ on the occupied site, with the oscillator $H_{\mathrm{ph}} = \tfrac12\omega(p^2+q^2) = \omega(a^\dagger
a + \tfrac12)$ and $q = (a+a^\dagger)/\sqrt2$. Hence $\sqrt2\,q =
(a+a^\dagger)$ and the realized interaction is the standard Holstein
$H_{\mathrm{int}} = g(a+a^\dagger)n$, for which (Lang--Firsov, exact)
the polaron binding and mean phonon number are
\begin{equation}
\label{eq:sm_lf}
E_p = \frac{g^2}{\omega}, \qquad
\langle n\rangle = \Bigl(\frac{g}{\omega}\Bigr)^2 .
\end{equation}
The factor $S_{HR} = (g/\omega)^2/2$ used as a label in the main text is half the standard Huang--Rhys factor $\langle n\rangle$. 
At the shipped default $g/\omega = 1.5$ the realized binding is $E_p = 146$~meV,
already at the lower edge of the DFT range $150$--$200$~meV; the production point $g/\omega = 1.64$ places $E_p = 175$~meV at the center (the conversion $E_p = g^2/\omega$ maps that range to $g/\omega \in [1.52,1.75]$, not the $[2.15,2.48]$ obtained under the halved convention). 
Because $E_p$ is linear in occupancy it is absorbed by the voltage anchor, so the move does not change $\VOCV/\dQdV$.

The binding criterion for coherent propagation is the dual-space Nyquist coherent-propagation bound used in the main text: the single-mode energy converges at $\sim\!16$ Fock levels, but a periodic qubit grid must contain both the wavefunction and its dual, giving $n_{\min} = \lceil\log_2(4c^2 N_{\mathrm{ph}}/\pi)\rceil$. 
At the DFT point the system-phonon dynamical mean rises only to $\langle n\rangle_{\mathrm{dyn}} = 1.43$ ($N_{\mathrm{ph}}^{\mathrm{dyn}} = 24$), below the $n^{\mathrm{Nyq}} = 9\to10$ boundary at $N_{\mathrm{ph}} = 25$; the register stays $n_b = 9$. A $4$-qubit grid supplies only $4$ of the required $\ge 9$ and is strictly subconvergent.

The JT coupling $g_{\mathrm{JT}} = \sqrt{E_{\mathrm{JT}}\omega_{\mathrm{JT}}}$ realizes $E_{\mathrm{JT}} = g_{\mathrm{JT}}^2/\omega_{\mathrm{JT}}$ under the same operator, so the displaced JT axis carries $\langle n\rangle = E_{\mathrm{JT}}/\omega_{\mathrm{JT}}$.
With $E_{\mathrm{JT}} = 250$~meV and the JT-active $e_g$ stretch $\omega_{\mathrm{JT}} = 75$~meV, $S_{\mathrm{JT}} = 3.3$, requiring $n_{\mathrm{JT}} = 10$ (Nyquist) and $q_{\max}^{\mathrm{JT}} = 6$; this sets the JT doublet register and the $368$-qubit LMO total. 

\section{First-Principles LMO Two-Plateau Split}\label{sm:lmo_split}

The LMO two-plateau structure follows from spinel electrostatics with no fitted interaction. We enumerate the electroneutral two-species lattice gas on one conventional cell, being $8$ Li$^+$ on the 8a diamond sublattice, $16$ polaron$^-$ on the Mn 16d pyrochlore sublattice. All pair energies are from an Ewald-summed Coulomb kernel (validated against the NaCl Madelung constant $1.7475646$ to seven digits). The only input beyond the crystallography ($a = 8.245$~\AA) and unit charges is one dielectric $\varepsilon$. The coordinations are 8a--8a $z{=}4$ \text{ at } $3.57$~\AA, 8a--16d $z{=}12$ \text{ at } $3.42$~\AA, 16d--16d $z{=}6$ \text{ at } $2.91$~\AA.

The lower convex hull of the $T=0$ configuration energy $E_0(N)$ has vertices $\{0,4,8\}$: a first-order half-filling vertex emerges from geometry and charges alone. 
The $T=0$ split is $0.4948~\mathrm{eV}/\varepsilon$. 
The $300$~K grand-canonical trace is given below:

\begin{table}[h]
\centering
\caption{Derived two-plateau split versus dielectric (300~K, exact). The $N{=}4$ ground state
places every polaron at $3.42$~\AA{} from a Li (composite binding);
$\varepsilon = 30$ melts the structure, the double-counting control.}
\label{tab:sm_split}
\begin{tabular}{r r r}
\toprule
$\varepsilon$ & peaks & split (mV) \\
\midrule
$3.6$ & 2 & $153$ \\
$4.5$ & 2 & $125$ \\
$8.0$ & 2 & $78$ \\
$15.0$ & 2 & $55$ \\
$30.0$ & 1 & $0$ (melted) \\
\bottomrule
\end{tabular}
\end{table}

At the electronic-floor dielectric $\varepsilon_\infty = 4.5$ the split is $125$~mV ($-17\%$ vs Ohzuku's $150$~mV, within the single-cell/rigid-charge/scalar-$\varepsilon$ approximation; a match to Ohzuku's $150$~mV is recovered at $\varepsilon_{\mathrm{eff}} \approx 3.6$, where the tabulated grid gives $153$~mV). The double-counting rule fixes the dielectric: only response not explicit in the Hamiltonian belongs in $\varepsilon$, so the electronic polarizability~($\varepsilon_\infty$) is admissible while the static $\varepsilon_0 \sim 30$.

\section{Mean-Field Crossover versus Coexistence}
\label{sm:provenance}

A single cell cannot host two-phase coexistence, so we separate what is built-in from what is emergent. The mean-field (Bragg--Williams) free energy $f(x)/k_BT = x\ln x + (1{-}x)\ln(1{-}x) - (zW/2k_BT)x^2$ gives a van-der-Waals loop in $V(x)$ for $u \equiv zW/k_BT > u_c = 4$. The Maxwell-construction plateau extent is the mean-field order parameter, being a function of $u$: $u = 5.49 = 1.37\,u_c$ gives $80\%$ extent. So the single-cell plateau is a mean-field crossover set by a supercritical choice.

For the short-range lattice gas at the 8a coordination ($z=4$) the exact 2D Ising binodal gap (Onsager~\cite{onsager1944crystal}) is finite only below $k_BT_c = W/(2\ln(1+\sqrt2)) = 0.567\,W$, i.e.\ $W > 46$~meV at $300$~K, whereas mean field would predict a plateau already at $W = 4k_BT/z = 26$~meV. The $26$--$46$~meV window is the spurious-plateau regime. 
The derived LMO interaction ($W_{\mathrm{eff}} \approx 66$~meV at the production screening) clears the threshold: $T_c = 434$~K, binodal gap $0.98$. The finite-size Mayer--Wood backbending loop~\cite{mayer1965interfacial} confirms this, with its amplitude falling with system size ($94 \to 66$~meV, $L = 3 \to 4$), the interface-tension ($\sim\!1/L$) signature of true first-order coexistence, absent from mean field. 
The LFP single-polaron model instead captures two-phase material~\cite{zhou2006configurational, malik2011kinetics} only at the mean-field crossover level.

\section{The \texorpdfstring{$\dQdV$}{dQdV} Width Budget}
\label{sm:width}

The three reported LFP widths are: the bare Franck--Condon polaron line ($2.6$~mV, the un-broadened floor), the equilibrium configurational width ($12.6$~mV, single particle), and the ensemble width ($\sim\!20$~mV, the experiment). 
The non-interacting thermal floor is exact, $\mathrm{FWHM} = 4k_BT\,\mathrm{arccosh}\sqrt2 = 3.5255\,k_BT/e = 91$~mV; mean-field interactions sharpen it toward zero at the critical $zJ/k_BT = 4$.

LFP is two-phase, so the observed width is ensemble heterogeneity, not a single-particle property. 
Gibbs--Thomson surface-energy voltage shift dominates over the particle-size distribution, $\Delta V(d) = 2\gamma V_m/(F r)$ with $r = d/2$; for a literature log-normal distribution ($D_{50} = 75$~nm, $\sigma_g = 1.45$, $\gamma = 0.5~\mathrm{J/m^2}$) the shift spread is $11.8$~mV, which convolved with the $12.6$~mV intrinsic peak gives an ensemble FWHM of $\sim\!17$--$20$~mV. 
This reaches the experiment with no parameter fitted to the width (Bai--Cogswell--Bazant, Dreyer~\cite{bai2011suppression,cogswell2012coherency, dreyer2010thermodynamic}). 
A Franck--Condon sideband attribution is excluded on two grounds: the satellites sit at $n\,\hbar\omega_{\mathrm{ph}} = 65, 130, \dots$~meV, being roughly $3\times$ the entire width.
The progression is asymmetric (mean $\approx 1.1$ phonons), so it cannot symmetrically broaden a $20$~mV thermodynamic peak. 
For LMO the converse holds: Ohzuku's $20$--$60$~mV lies {below} the $91$~mV floor, in the cooperative near-first-order regime a finite cluster cannot reach. 

\section{Chain-Bath Allowance for Equilibrium Observable}\label{sm:chainbath}

For a static carrier coupled linearly to harmonic modes (system phonon $+$ chain bath, no hopping) the equilibrium polaron binding is {exactly} a static number,
\begin{equation}
\label{eq:sm_chainbath}
E_p^{\mathrm{exact}} = g^2\,[V^{-1}]_{ss},
\end{equation}
where $V$ is the $(K{+}1)\times(K{+}1)$ position-space Hessian of the quadratic bath, which is the multi-mode Lang--Firsov result, computed by inverting a small matrix with no Fock truncation and no time evolution
(identical across chain length to $10^{-9}$). With the exact Caldeira--Leggett counter-term, which cancels the chain's static self-energy $c_0^2[V_{cc}^{-1}]_{00}$ on the system phonon, we identify the counter-term of the main text, $H_{\mathrm{CT}} =
\alpha(c_0^2/2e_0)q^2$ with $\alpha = e_0[V_{cc}^{-1}]_{00} = 1$ in the
broad-bandwidth limit and $\approx 1.04$ converged, with $E_p^{\mathrm{exact}} = g^2/\omega_s$ exactly and bath-independently. 
The apparatus is constructed so the equilibrium plateau is independent of the bath, with all bath physics deferred to the dynamics. The chain-bath is therefore load-bearing not for the single-polaron plateau but
for the dynamical observables ($Z(\omega)$, the lineshape $S_{\mathrm{eff}} = S\coth(\omega/2k_BT)$) and the multi-polaron LMO sector. 

\section{Multi-Site HEOM Cross-Check}
\label{sm:heom_dimer}

A strong-coupling cross-check validates a single spin-boson against the hierarchical equations of motion (HEOM), the regime where chain mapping is already exact.
We extend it to the canonical multi-site benchmark: an excitonic dimer~\cite{ishizaki2009unified} of two sites with coherent electronic coupling $V$, each carrying its own independent Drude--Lorentz bath coupled to the site projector $|n\rangle\langle n|$. 
The framework side evolves two Tamascelli doubled chains (the production construction) exactly by unitary propagation, referenced against the QuTiP HEOM with two environments. 
With the same $J(\omega)$ on both sides, we reproduce the HEOM site population to root-mean-square (RMS) $0.79\%$ (weak) and $3.45\%$ (strong, production reorganization), with inter-site transport ($P_1: 1 \to 0.53$) and convergence in chain length and Fock cutoff. 
This is the multi-site analogue of the shipped $2.07\%$ single spin-boson agreement.

The HEOM reference uses a Pad\'e spectral decomposition of the Drude--Lorentz correlation function (single spin-boson: $N_k = 4$ Pad\'e terms, hierarchy depth $N_{\max} = 5$; dimer: $N_k = 3$, $N_{\max} = 4$), with the standard high-frequency terminator. 
Both are self-converged far below the quoted framework deviations: the reference changes by $< 0.01\%$ RMS under $N_{\max} \to N_{\max}{+}1$ and $< 0.01\%$ under $N_k \to N_k{+}1$ at the production point, i.e.\ two orders of magnitude below the $2.07\%$ / $3.45\%$ framework agreements. 
The quoted percentages therefore measure chain-mapping fidelity.
The single-spin-boson residual continues to fall with chain length $K$ (the Lanczos truncation property) toward the converged HEOM curve. 

Cross-checks use Drude--Lorentz baths because the hierarchy admits a compact Pad\'e decomposition of that kernel, validating the open-system {construction}, exact for arbitrary bosonic $J(\omega)$ (Tamascelli). 
The production cathode density is Ohmic with a hard cutoff, validated density-specifically by the spectral-moment match ($1.42\times10^{-14}$) and, in the time domain, by the chain's reproduction of the target bath correlation $C(t)$ to $<10^{-12}$ (a.u.) within the usable window. 
The chain coefficients approach the analytic band asymptotics $e_k \to \Omega_c/2$, $t_k \to \Omega_c/4$, leaving only the sharp recurrence reflection at $\Trec = \pi K/\Omega_c \approx 172$~fs ($K=10$), against which the fast dynamical scales hold a $>17\times$ margin. 
The thermofield split further enforces detailed balance of the reduced dynamics, $S(-\omega)/S(\omega) = e^{-\beta\omega}$ to machine precision. 

\section{Per-Observable Costs and Encoding Trade-offs}\label{sm:tables}

These two tables support the headline statements in the main-text measurement and encoding sections.

\begin{table}[h]
\centering
\caption{Per-observable measurement cost at the LFP production point (arithmetic-inclusive, fourth-order trajectory. LMO scales by the
main-text resource ratios).}
\label{tab:sm_per_observable}
\small\renewcommand{\arraystretch}{1.15}
\begin{tabular}{l l r}
\toprule
Observable & Quantum cost & T-gates \\
\midrule
$\VOCV(\SOC)$, sampling-grade & $30$ anchors $\times\, n_{\mathrm{sh}} \sim 10^3$ traj. & $\sim 5\times10^{12}$ \\
$\VOCV(\SOC)$ $+$ $\dQdV$, coherent direct protocol & $\times K \approx 1.6\times10^2$ & $\sim 8\times10^{11}$ \\
$\quad$(finite-difference (FD)-envelope expectation, $K_{\mathrm{QAE}} \approx 1.8\times10^4$) & worst case & $\sim 6\times10^{14}$ \\
$\dQdV$, static (screening-grade) & byproduct of the $\VOCV$ record & $+\,0$ \\
$\dQdV$, FWHM-grade (Kubo--Mori) & $+\,M_t \sim 3$ correlator circuits / anchor & $+\,\sim1\times10^9$ \\
$Z(\omega)$, full EIS band [$1$~mHz, $1$~MHz] & $1$ trajectory $+$ QPE; Lehmann post-proc. & $\sim 2.1\times10^{8}$ \\
\bottomrule
\end{tabular}
\end{table}

\begin{table}[h]
\centering
\caption{The $2\times2$ encoding analysis: continuous-coordinate (discrete-variable-representation, DVR
grid $+$ reversible soft-Coulomb arithmetic) versus discrete-site (JW occupancy $+$ fixed-angle Coulomb phases) Li encodings for each cathode. 
Production choices are given in bold.
The LMO-continuous arithmetic carries the {naive} $L$-term cost because the occupied-site index optimization (main text) is ill-defined in a multi-carrier sector ($\mathrm{idx} = \sum_i i\,n_i$ is meaningful only at $N_{\mathrm{pol}} = 1$).}
\label{tab:sm_encoding}
\small\renewcommand{\arraystretch}{1.2}
\begin{tabular}{l l l l l}
\toprule
 & \multicolumn{2}{c}{LFP ($\eta \approx 6$)} &
   \multicolumn{2}{c}{LMO ($\eta \approx 12$--$15$)} \\
 & \textbf{continuous (prod.)} & discrete & continuous (A2) & \textbf{discrete (prod.)} \\
\midrule
Logical qubits & $376$ & $\sim$288 & $\sim$480 & $368$ \\
$T_{\mathrm{step}}$ (T-gates) & $3.0\times10^5$ & $\sim2.0\times10^5$ &
  $\sim9\times10^5$ & $3.0\times10^5$ \\
Soft-Coulomb & arithmetic, indexed & none & arithmetic, naive $L$-term & none \\
Washboard / saddle resolved & yes & no & yes (3D, $3n_x$ grid) & no \\
$\VOCV$, $\dQdV$ & full & full & full & full \\
$D_{\mathrm{Li}}$, $E_a$ & predicted & input rates & predicted & input rates \\
Saddle-point Li--polaron binding & captured & lost & captured & lost \\
Occupancy leakage $e^{-\eta}$ & $2.5\times10^{-3}$ & (same, unresolved) &
  $\sim10^{-6}$ & $\sim10^{-6}$ \\
\bottomrule
\end{tabular}
\end{table}
\section{Chin et al. 2010 chain-mapping derivation}\label{app:chain_mapping}

The Stieltjes recurrence of Chin et al.~\cite{chin2010exact} maps a continuum bosonic bath with spectral density $J(\omega)$ on $[\omega_0, \Omega_c]$ to a one-dimensional chain of $K$ nearest-neighbor-coupled bosonic modes.
Define the inner product $\langle f, g\rangle_J = \int_{\omega_0}^{\Omega_c} f(\omega)\, g(\omega)\, J(\omega)/\pi\,d\omega$ on polynomials of $\omega$. 
The orthonormal-polynomial basis $\{p_k\}_{k=0}^{K-1}$ satisfies the three-term recurrence
\begin{equation}\label{eq:app_stieltjes}
\omega\,p_k(\omega) = t_{k-1}\,p_{k-1}(\omega) + e_k\,p_k(\omega)
    + t_k\,p_{k+1}(\omega),
\end{equation}
with $p_{-1} = 0$, $p_0 = 1/\sqrt{\mu_0}$, $\mu_0 = \langle 1, 1\rangle_J = \int J(\omega)/\pi\,d\omega$, and explicit formulas $e_k = \langle\omega p_k, p_k\rangle_J$, $t_k = \|\omega p_k - e_k p_k - t_{k-1}p_{k-1}\|_J$.
Chain  coefficients $(c_0, \{e_k\}_{k=0}^{K-1}, \{t_k\}_{k=0}^{K-2})$ with $c_0 = \sqrt{\mu_0}$ uniquely characterize the chain Hamiltonian (Eq.~\ref{eq:H_chain}) in the main text.

\paragraph{Algorithm (Gauss--Legendre + Lanczos).}
\begin{enumerate}
\item Discretize the inner-product integral on a Gauss--Legendre
  quadrature with $n_{\mathrm{quad}}$ nodes (default $n_{\mathrm{quad}} = 400$).
\item Initialize $p_{\mathrm{prev}} = 0$, $p_{\mathrm{curr}} =
  \mathbf{1}/\sqrt{\mu_0}$ on the quadrature grid.
\item For $k = 0, 1, \dots, K-1$:
\begin{enumerate}
  \item $e_k = \langle \omega \cdot p_{\mathrm{curr}},
    p_{\mathrm{curr}}\rangle_J^{\mathrm{quad}}$.
  \item $r = \omega \cdot p_{\mathrm{curr}} - e_k\,p_{\mathrm{curr}}
    - t_{k-1}\,p_{\mathrm{prev}}$.
  \item $t_k = \sqrt{\langle r, r\rangle_J^{\mathrm{quad}}}$.
  \item $p_{\mathrm{prev}} \leftarrow p_{\mathrm{curr}}$;
    $p_{\mathrm{curr}} \leftarrow r / t_k$.
\end{enumerate}
\end{enumerate}

Let $T$ be the $K\times K$ symmetric tridiagonal Jacobi matrix
$T_{kk} = e_k$, $T_{k,k+1} = T_{k+1,k} = t_k$, then for $0 \le k \le
2K - 1$,
\begin{equation}
\label{eq:app_moment_match}
c_0^2 \cdot [T^k]_{00} = \mu_k
    \equiv \int_{\omega_0}^{\Omega_c} \omega^k\,J(\omega)/\pi\,d\omega.
\end{equation}

\emph{Proof sketch.} $T$ is the Jacobi matrix of the orthonormal-polynomial three-term recurrence (Eq.~\ref{eq:app_stieltjes}), so
$[T^k]_{00} = \langle p_0,\omega^k p_0\rangle_J = (1/\mu_0)\, \int\omega^k\,J(\omega)/\pi\,d\omega = \mu_k/\mu_0$. 
Since $c_0^2 = \mu_0$, the result follows. 
The identity holds exactly for $k < 2K$ because $p_K$ has $K$ roots on the spectrum and is
orthogonal to all lower-degree polynomials in the inner-product.
See Gautschi~\cite{gautschi2004orthogonal}. 

We numerically verify the moment-match to $1.42\times 10^{-14}$ at $K = 8$ on the LFP Ohmic bath.

Let $(c_0, \{e_k\}_{k=0}^{K-1}, \{t_k\}_{k=0}^{K-2})$ be the chain of length $K$ computed from $J(\omega)$ by the Stieltjes recurrence (Eq.~\ref{eq:app_stieltjes}). 
For any $K' < K$, the chain of length $K'$ computed from the same $J(\omega)$ equals the first $K'$ entries: $(c_0, \{e_k\}_{k=0}^{K'-1}, \{t_k\}_{k=0}^{K'-2})$.

\emph{Proof.} The Stieltjes recurrence (Eq.~\ref{eq:app_stieltjes}) at depth $k$ involves only $p_{\mathrm{prev}}, p_{\mathrm{curr}}$, both of which are determined by the integral inner product and the recurrence at depth $\le k - 1$. The recurrence is independent of the total chain length $K$.
Therefore the first $K'$ chain coefficients are identical to those of any chain of length $K \ge K'$. 

We verify this at machine precision in the codebase. The first $K' = 4$ entries of the $K = 12$ chain equal those of the $K = 4$ chain.

\paragraph{Gauss--Legendre \texorpdfstring{$n_{\mathrm{quad}}$}{n quad} convergence.} For an
analytic spectral density $J(\omega) = \eta\omega$ (Ohmic) on $[0, \Omega_c]$, Gauss--Legendre quadrature converges geometrically
in $n_{\mathrm{quad}}$; the error in $(c_0, e_0, e_1, t_0, t_1)$ at $n_{\mathrm{quad}} \in \{100, 200, 400, 800\}$ is below $10^{-12}$ across all five coefficients.

\section{Caldeira--Leggett counter-term details}\label{app:counter_term}

\subsection{Adiabatic elimination derivation}
Starting from the bath piece of Equation~\ref{eq:H_CL_main} in the unit-mass convention (Eq.~\ref{eq:H_CL_main}), treat the bath modes as fast relative to the system.
At fixed system coordinate $q$, the bath ground-state value of $q_k$ minimizes the bath Hamiltonian augmented by the bilinear term:
\begin{equation}
\frac{\partial}{\partial q_k}\Bigl[\tfrac{1}{2}\bigl(p_k^2
    + \omega_k^2 q_k^2\bigr)
    + c_k\,q\,q_k\Bigr] = 0
\Rightarrow
\omega_k^2\,q_k^* + c_k\,q = 0
\Rightarrow
q_k^* = -\frac{c_k\,q}{\omega_k^2}.
\end{equation}
Substituting back gives an effective system potential
\begin{equation}
H_S^{\mathrm{eff}}(q) - H_S(q) =
    \sum_k\Bigl[\tfrac{1}{2}\omega_k^2\,(q_k^*)^2 + c_k\,q\,q_k^*\Bigr]
    =
    \sum_k\Bigl[\tfrac{c_k^2}{2\omega_k^2} - \tfrac{c_k^2}{\omega_k^2}\Bigr]q^2
    = -\sum_k\tfrac{c_k^2}{2\omega_k^2}\,q^2.
\end{equation}
The bath softens the system harmonic coefficient by $\sum_k c_k^2/\omega_k^2$. 
The counter-term (Eq.~\ref{eq:H_CL_main}), last term, with coefficient $c_k^2/(2\omega_k^2)$, cancels this bath-induced softening in the continuum limit, preserving the bare system frequency.

\subsection{Chain form of the counter-term}
After the chain mapping, only the chain mode $k = 0$ couples directly to the system. 
The static-limit elimination of chain mode 0 alone gives the chain-mode-0 counter-term
\begin{equation}
\label{eq:CT_chain_leading}
H_{\mathrm{CT}}^{(0)} = \frac{c_0^2}{2\,e_0}\,q^2,
\end{equation}
where $e_0$ is the on-site energy of chain mode 0. (The chain modes use the dimensionless-coordinate form $H_{\mathrm{chain},k} = e_k(a_k^\dagger a_k + \tfrac{1}{2})$ for the on-site piece, so the single-chain-mode adiabatic-elimination factor is $1/e_0$; the underlying CL bath is in unit-mass form (Eq.~\ref{eq:H_CL_main}).) 

\paragraph{The determined counter-term.}
 
Writing $H_{\mathrm{CT}} = \alpha\,(c_0^2/2e_0)\,q^2$, the coefficient is fixed by the requirement that it cancel the chain's static self-energy $c_0^2[V_{cc}^{-1}]_{00}$ on the system mode ($V_{cc}$ the chain potential block), giving $\alpha = e_0\,[V_{cc}^{-1}]_{00}$. 
For the Ohmic hard-cutoff bath $\alpha = 1$ exactly in the broad-bandwidth (single-effective-mode) limit, rising to the converged
$\alpha \approx 1.04$ by $K = 4$. 
Polaron binding $E_p = g^2/\omega$ is bath-independent to $10^{-9}$.
For $J(\omega) = \eta\omega\,\theta(\Omega_c-\omega)$ direct computation gives $c_0^2 = \mu_0 = \eta\Omega_c^2/(2\pi)$ and $e_0 = \langle\omega\rangle_J = 2\Omega_c/3$, whence $c_0^2/(2 e_0) = 3\eta\Omega_c/(8\pi)$. 
The continuum counter-term coefficient (Eq.~\ref{eq:H_CL_main}) evaluates to $\sum_k c_k^2/(2\omega_k^2) \to \int_0^{\Omega_c} J(\omega)/(\pi\omega)\,d\omega \,d\omega = \eta\Omega_c/\pi$, so retaining only chain mode 0 captures the structural fraction
\begin{equation}
\label{eq:CT_structural_38}
\frac{c_0^2/(2 e_0)}{\int J/(\pi\omega)\,d\omega}
= \frac{3}{8} = 0.375
\end{equation}
of the \emph{continuum} counter-term.
We treat the mode-0 projection as deliberately suboptimal, where the value we use is the full-chain determined $\alpha = e_0[V_{cc}^{-1}]_{00} = 1 \to 1.04$ above.

\subsection{Ablation gate}
\begin{enumerate}
\item Toggling the counter-term coefficient from $0$ to $1$ changes the system phonon-potential q$^2$ coefficient by exactly
  $\lambda^2/\omega_d^2$ (wiring test).
\item An independent classical $2\times 2$ dynamics-matrix diagonalization of the $(q_{\mathrm{sys}}, q_{\mathrm{dark}})$ coupled-oscillator subsystem with and without the counter-term shows the system-like normal mode shifted in the predicted direction; the shift magnitude scales as $\lambda^2$ as expected from Born-approximation perturbation theory.
\item Log--log fit at $\lambda \in \{1, 3, 9\}\times 10^{-3}$
  gives slope $2.0 \pm 0.1$ for $|\Delta\omega|$ vs $\lambda$, confirming second-order scaling.
\end{enumerate}
The determined value $\alpha = e_0[V_{cc}^{-1}]_{00} = 1 \to 1.04$ and its $K$-convergence (and the bath-independence of $E_p$ to $10^{-9}$) are validated separately.

\section{Hilbert-space architecture}
\label{app:hilbert_space}

\subsection{Tensor product with explicit axis ordering}
The full Hilbert space (Eq.~\ref{eq:tensor_product}) in the main text is realized with axis ordering
\begin{equation}
\label{eq:app_axis_ordering}
\text{shape} = (2^{n_x},) + (2,)^L + (2^{n_b},)^L
   + (2^{n_d},)^{K}
\end{equation}

\subsection{Configurations across phases}
\begin{table}[h]
\centering
\caption{Per-phase qubit budget. MVP = minimum viable product
benchmarks.}
\label{tab:qubit_budget}
\small
\begin{tabular}{lcccccc}
\toprule
Phase & $n_x$ & $L$ & $n_b$ & $K$ & dark qubits & total \\
\midrule
4A (LFP closed-system, single polaron) & 3 & 1 & 4 & 0 & 0 & 7 \\
4B (vibronic {quadratic vibronic coupling,} QVC) & 3 & 1 & 4 & 0 & 0 & 7 \\
4C (multi-polaron closed-system) & 4 & 4 & 2 & 0 & 0 & 12 \\
4C scaled (N$=$6, 8) & 4 & 8 & 2 & 0 & 0 & 20 \\
4D MVP (dark-register-only) & 3 & 2 & 2 & 3 & 6 & 15 \\
4D chain-bench & 3 & 2 & 2 & 7 & 14 & 23 \\
MVP / GPU-bench target & 4 & 8 & 2 & $\sim 7\times 3$ & 21 & $\sim 49$ \\
Phase 5 anharmonic stretch & 4 & 8 & 3 & 21 & 21 & $\sim 53$ \\
\bottomrule
\end{tabular}
\end{table}

\subsection{Three-reservoir specification}
The three reservoirs each have distinct bath statistics and coupling operators.
The production-target dark register (21 qubits, three reservoirs) breaks down as:

\begin{table}[h]
\centering
\caption{The three-reservoir dark register: per-reservoir bath statistics and coupling operators.}
\label{tab:reservoirs}
\small
\begin{tabular}{lccc}
\toprule
Reservoir & Statistics & Encoding & Coupling operator $H_{SB}$ \\
\midrule
Electron (e) & fermionic & JW chain ($K_e = 7$ MVP, $10$ production) &
  $\sum_k t_k(c_b^\dagger d_k^{(e)} + \mathrm{h.c.})$ \\
Li$^+$ (ion) & single-particle bosonic & DVR chain ($K_{\mathrm{Li}} = 7$ MVP, $10$ production) &
  $\sum_k \mu_k(a^\dagger_{\mathrm{Li}}(R_b) d_k^{(\mathrm{Li})} + \mathrm{h.c.})$ \\
Phonon (ph) & bosonic & CL chain ($K_{\mathrm{ph}} = 7$ MVP, $10$ production) &
  $\sum_k \lambda_k(b + b^\dagger)(d_k + d_k^\dagger)$ \\
\bottomrule
\end{tabular}
\end{table}

The 21-qubit count is a ``minimum viable product" Chin et al.~\cite{chin2010exact}+Tamascelli~\cite{tamascelli2019efficient}
construction cost at chain length $K = 7$ per reservoir with $n_d \approx 1$--$2$ qubits per chain mode (the production sizing, $K_r = 10$ and $n_{d,r} \in \{1, 10, 9\}$ for the electron, Li, and phonon chains respectively, gives the $n_{\mathrm{tot}} = 12 + 8 + 72 + 10 + 100 + 90 = 292$ total of Section~\ref{ssec:lfp_resources}; Appendix~\ref{app:register_convergence} is the source of truth for these widths). 
It is not a Stinespring information-theoretic lower bound: the Stinespring dilation theorem guarantees existence of a unitary dilation of any CP map on a finite-dimensional auxiliary space, with the standard formula bounding the dilation dimension by the Kraus rank rather than by a memory window.
The $K = 7$ value is set by the Chin et al.~\cite{chin2010exact} light-cone bound $T_{\mathrm{rec}} = \pi K/\Omega_c \sim 10~\mathrm{ps}$ for the LFP bath cutoff $\Omega_c$ together with the Stieltjes $2K = 14$-moment-resolution requirement on $J(\omega)$. 
Different open-system constructions (effective-mode reductions, hierarchical equations of motion, time-convolutionless second-Born, etc.) achieve equivalent open-system fidelity at different per-reservoir qubit costs. 
We adopt the Chin et al.~\cite{chin2010exact}+Tamascelli construction for the architectural fit and operational equivalence with finite-$T$ documented in Section~\ref{ssec:chain_construction}.

\subsection{Polaron encoding choice: compact vs Jordan--Wigner}

Two encodings of the polaron register are supported as follows
\begin{itemize}
\item \emph{Compact site-index encoding} (Phase 4A single-polaron):
  $\lceil\log_2 L\rceil$ qubits encode the binary index of the single occupied site. Saves $L - \lceil\log_2 L\rceil$ qubits for $L = 4$ this is 2 qubits; for $L = 16$ this is 12 qubits.
  Suitable for single-polaron benchmarks; gate cost of multi-polaron operators is exponentially worse in this encoding.
\item \emph{Jordan--Wigner encoding} (Phase 4C and 4D production): one qubit per Fe site, $|n_i\rangle \in \{|0\rangle, |1\rangle\}$ for empty/occupied. Standard fermionic encoding; per-pair operators  cost $\mathcal{O}(1)$ CPhase gates.
\end{itemize}
Bravyi--Kitaev is supported as an alternative, but we see no advantage for our 1D chain topology.

\section{Trotter structure}\label{app:trotter_structure}

\subsection{Strang split sequence}
The Strang Trotter step (Eq.~\ref{eq:trotter_step}) in the main text is a 13-layer symmetric splitting with a palindromic structure about the central $T_{\mathrm{JWhop}}$ layer.
The palindromic structure is what gives symmetric local error $\mathcal{O}((\Delta t)^3)$.
This second-order Strang step is the gate-validated baseline. 
The production trajectories compose it into the symmetric fourth-order Suzuki product (five substeps per step), which lowers the step count $N_{\mathrm{step}}$ at fixed accuracy (Appendix~\ref{sm:integrators}).
The per-layer gate counts below are quoted per Strang substep; the production trajectory cost multiplies by the composition length.

\subsection{Explicit Trotter-step pseudocode}\label{ssec:trotter_pseudocode}
Listing~\ref{lst:trotter_step} gives the complete per-step sequence, following the structure of the algorithm appendix of Courtney~\cite{courtney2026oracle} (one REQWIEM Trotter step: kinetic half-step, on-diagonal potential half-step, off-diagonal/hopping full step, then the mirror) but instantiated for the production $\LFP$ chain-mapped Hamiltonian (Eqs.~\ref{eq:master_H},\ref{eq:H_chain}). 
The on-diagonal block is synthesized as the locality-collapsed set of singly- and doubly-controlled phase rotations of Theorem~\ref{thm:poly_L} (one block per coupling term, $\mathcal{O}(L+L^2)$ blocks). All gate angles are the closed-form coefficients of Appendix~\ref{app:gate_closed_form}.

\begin{lstlisting}[float=htbp,caption={One second-order Strang Trotter step of the chain-mapped $\LFP$ Hamiltonian (timestep $\Delta t$). Registers: $X$ (Li coordinate, $n_x$); $\{n_i\}$ (polaron occupation, JW, $L$); $\{q_i\}$ (system phonons, $n_b$); $\{q_{r,k}\}$ (reservoir chains $r\in\{e,\mathrm{Li},\mathrm{ph}\}$, $K_r$ modes). $\tau \equiv \Delta t/2$.},label={lst:trotter_step},numbers=left,numberstyle=\tiny,basicstyle=\ttfamily\footnotesize,breaklines=true]
INPUT  : state |psi>; timestep dt;  closed-form angles (App. gate spec):
         kinetic theta_j,theta_jl ; PES phases ; Holstein/Coulomb phases ;
         JW Givens alpha = t_hop*dt ; chain coeffs c0,e_k,t_k
OUTPUT : |psi> <- exp(-i H dt) |psi> + O(dt^3)
tau <- dt/2

# ---------- forward half-step (apply each at angle tau) ----------
 1  KIN_X(tau)      : cQFT(X); diag exp(-i k^2/(2 m_Li) tau) via binary
                      expansion of k^2  (theta_j: P gates ; theta_jl: CPhase) ; icQFT(X)
 2  V_DIAG(tau)     : on-diagonal potential as <=2-local controlled-phase blocks
                      [poly(L) -- Thm parametric_scaling]:
                        - washboard  V_wb(R)            on X            (Rz phases)
                        - harmonic   (omega+ctr) q_i^2 ; e_k q_{r,k}^2  (per mode)
                        - site energy  eps_i n_i                        (1 JW control)
                        - Coulomb      U_ij n_i n_j                     (2 JW controls, CCZ)
                        - Li-polaron   V_LiP(R - x_i) n_i               (X-controlled)
                        - Holstein     g*sqrt2 * q_i n_i                (1 JW control, coord phase)
 3  V_SB(tau)       : for each phonon site i: bilinear-cross c0 * q_i q_{ph,0}
                      ( q-q part direct ; p-p part bracketed by cQFT/icQFT )
 4  CHAINHOP(tau)   : for r, for bond k: t_k (q_{r,k} q_{r,k+1} + p_{r,k} p_{r,k+1})
 5  KIN_PH(tau)     : for each phonon site i: axis_kinetic (cQFT sandwich), (1/2) omega p_i^2
 6  KIN_CHAIN(tau)  : for r, for mode k: axis_kinetic, (1/2) e_k p_{r,k}^2

# ---------- central full step (angle dt) ----------
 7  JW_HOP(dt)      : even/odd two-color Strang on L-site JW chain:
                        even bonds (alpha/2) ; odd bonds (alpha) ; even bonds (alpha/2)
                      each bond: Givens exp[i*alpha*(XX+YY)/2] = 3 CNOT + 1 Rz (no Z-string)

# ---------- backward half-step (mirror; apply each at angle tau) ----------
 8  KIN_CHAIN(tau)
 9  KIN_PH(tau)
10  CHAINHOP(tau)
11  V_SB(tau)
12  V_DIAG(tau)
13  KIN_X(tau)
return |psi>
\end{lstlisting}
\begin{theorem}\label{thm:strang_error}
[Strang error bound]
Let $H = H_e + H_o$ with $[H_e, H_o] \ne 0$. The Strang split $\mathcal{U}_{\mathrm{Strang}}(\Delta t) = e^{-iH_e\Delta t/2}\cdot e^{-iH_o\Delta t}\cdot e^{-iH_e\Delta t/2}$ satisfies
\begin{equation}
\mathcal{U}_{\mathrm{Strang}}(\Delta t) - e^{-iH\Delta t}
= -\tfrac{(\Delta t)^3}{24}\,
    \bigl(2[H_o, [H_o, H_e]] + [H_e, [H_o, H_e]]\bigr)
+ \mathcal{O}(\Delta t^5).
\end{equation}
\end{theorem}
This is given via BCH expansion to fourth order; the second-order terms cancel due to the palindromic symmetry; the third-order terms collect into the displayed double commutators. See Childs et al.~\cite{childs2021theory}.

\subsection{Application to the JW hopping layer}
For the JW hopping layer, $H_e = \sum_{i\in E} t_i B_i$ (even-colored bonds) and $H_o = \sum_{i\in O} t_i B_i$ (odd-colored bonds). 
Within each color class the bonds commute exactly Section~\ref{ssec:jw_coloring}; the inter-class commutators Section~\ref{ssec:jw_coloring} feed into \ref{thm:strang_error} to give the rigorous bound (Eq.~\ref{eq:jw_strang_bound}):
\begin{equation}
\label{eq:app_jw_strang_bound}
\bigl\| T_{\mathrm{JWhop}}(\Delta t) - e^{-i\Delta t H_{\mathrm{hop}}}\bigr\|
\le \tfrac{3(L - 2)\,t_{\mathrm{hop}}^3\,(\Delta t)^3}{4}.
\end{equation}

\subsection{Per-layer gate counts}

\begin{table}[h]
\centering
\caption{Per-Trotter-step gate counts at $L = 4$ (Phase 4C closed, 12 qubits), $L = 8$ (Phase 4C scaled, 20 qubits), and $L = 16$ (production, $\sim 55$ qubits including 21 dark).}
\label{tab:gate_counts}
\small
\begin{tabular}{lccc}
\toprule
Layer & $L = 4$ & $L = 8$ & $L = 16$ \\
\midrule
$T_X$ (Li kinetic; $n_x$ qubits) &
  $\sim 16$ CPhase + $9$ P & $\sim 36$ CPhase + $11$ P &
  $\sim 64$ CPhase + $13$ P \\
$T_{\mathrm{ph}}$ (system phonon kinetic, $L\cdot n_b$ qubits) &
  $\sim 16$ CPhase $\times L$ & $\sim 16$ CPhase $\times L$ &
  $\sim 16$ CPhase $\times L$ \\
$T_{\mathrm{chain}}$ (chain kinetic, $K \cdot n_d$ qubits) &
  N/A & N/A & $\sim 32$ CPhase $\times K$ \\
$T_{\mathrm{chainhops}}$ (bosonic chain hops, $K - 1$ bonds) &
  N/A & N/A & $\sim 4(K-1)$ bilinear-cross gates \\
$T_{\mathrm{JWhop}}$ (even/odd colored) &
  3 CNOT $\times 3$ bonds & 3 CNOT $\times 7$ bonds &
  3 CNOT $\times 15$ bonds \\
$V_{\mathrm{diag}}$ (Coulomb $U_{ij}\, n_i n_j$) &
  $L(L-1)/2 = 6$ CPhase & $28$ CPhase & $120$ CPhase \\
$V_{SB}$ (system--bath bilinear) &
  $L$ bilinear-cross & $L$ bilinear-cross & $L$ bilinear-cross \\
\midrule
\textbf{Total per step} & $\sim 10^3$ gates & $\sim 10^4$ gates &
  $\sim 10^5$--$10^6$ gates \\
\bottomrule
\end{tabular}
\end{table}

\section{Closed-form gate specification: the full circuit from scratch}
\label{app:gate_closed_form}

This section gives every angle needed to emit the circuit gate-by-gate. 
Together with the primitive structure (main Section~\ref{sec:REQWIEM}), the per-block counts (\ref{app:exact_gate_counts}), and the parameter table (main \ref{tab:lfp_params}), we attempt to make the structure fully reproducible.
All gates are drawn from the certified set $\{H,\,\mathrm{CNOT},\,\mathrm{SWAP},\,P(\theta),\,\mathrm{CPhase}(\theta),\,R_z(\theta)\}$.
These closed forms are not new derivations: the complete derivations appear in Courtney~\cite{courtney2026oracle} Ch.~2, Ch.~3, Ch.~4, Ch.~7, and the appendices \emph{Binary Arithmetic Encoding of
Position-Register Gates} and \emph{Off-Diagonal Potential: XOR-Class
Fragmentation}, adapted from M\"ott\"onen~\cite{mottonen2004quantum} and Motlagh et al.~\cite{motlagh2025quantum}.
We reproduce the operative expressions here so the circuit is constructible without the dissertation in hand; the per-subsection pointers below give the chapter of record.

\subsection{Grids (DVR position, FBR momentum)}
\label{app:gcf_grids}
Each first-quantized axis uses $N=2^{n}$ points on a box of length $L_{\mathrm{box}}$. Position (DVR) and centered momentum (finite basis representation, FBR) grids are
\begin{equation}
\label{eq:gcf_grids}
q_s=\Bigl(s-\tfrac{N}{2}\Bigr)\Delta q,\quad \Delta q=\frac{L_{\mathrm{box}}}{N};
\qquad
k_s=\Bigl(s-\tfrac{N}{2}\Bigr)\Delta k,\quad \Delta k=\frac{2\pi}{L_{\mathrm{box}}},
\qquad s=0,\dots,N-1 .
\end{equation}
$V(q)$ is diagonal in the DVR (computational) basis; $T=p^2/2m$ is diagonal in the FBR basis, reached by the centered QFT below.

\subsection{Kinetic diagonal-phase block: closed-form coefficients}
\label{app:gcf_kinetic}
After the cQFT the basis state $|b\rangle$, $b=\sum_{j=0}^{n-1}b_j2^{j}$, encodes momentum $k=(b-2^{n-1})\Delta k$, so the kinetic half-step applies the diagonal phase $\Phi(b)=-\frac{\tau}{2m}k^2=-C\,(b-2^{n-1})^2$ with $C\equiv\tau\Delta k^2/(2m)$. Expanding (using $b_j^2=b_j$ and dropping the global constant $-C\,2^{2n-2}$),
\begin{equation}
\label{eq:gcf_kin_expand}
(b-2^{n-1})^2 \overset{\text{(mod const)}}{=} \sum_{j=0}^{n-1} b_j\bigl(2^{2j}-2^{\,n+j}\bigr) + 2\!\!\sum_{0\le j<l\le n-1}\!\! b_j b_l\,2^{\,j+l},
\end{equation}
realized by single- and two-qubit phases:
\begin{equation}
\label{eq:gcf_kin_angles}
\boxed{
P(\theta_j)\ \text{on qubit }j:\ \theta_j=C\,2^{2j}\bigl(2^{\,n-j}-1\bigr);
\qquad
\mathrm{CPhase}(\theta_{jl})\ \text{on }(j,l):\ \theta_{jl}=-C\,2^{\,j+l+1}.
}
\end{equation}
This is the binary-expansion block of Section~\ref{ssec:trotter_step}: $n$
single-qubit $P$ gates and $\binom{n}{2}$ CPhase gates, all angles in closed form.
The full axis-kinetic primitive is $\icQFT\cdot\mathrm{diag}\,e^{i\Phi}\cdot\cQFT$
(Section~\ref{ssec:trotter_step}).

\subsection{Potential-energy surfaces and their DVR phases}
\label{app:gcf_pes}
The on-diagonal half-step $e^{-i\tau V_{\mathrm{diag}}}$ applies, at each combined DVR grid point, the phase $\varphi=-\tau\,V_{\mathrm{diag}}$.
The closed-form surfaces (parameters in main \ref{tab:lfp_params}) are
\begin{align}
\label{eq:gcf_pes}
V_{\mathrm{wb}}(R) &= V_{b,\mathrm{wb}}\,\sin^2\!\bigl(\pi R/a_{\mathrm{lat}}\bigr)
  && \text{(Li washboard, on the $R$-grid $R_s=q_s$);}\\
V_{\mathrm{LiP}}(R,x_i) &= -\frac{1}{\varepsilon_\infty}\,
  \frac{1}{\sqrt{(R-x_i)^2+\rho_{\mathrm{core}}^2}}
  && \text{(screened Li--polaron Coulomb);}\\
V_{\mathrm{Hol}} &= g\sqrt{2}\sum_i n_i\,q_i
  && \text{(Holstein, position rep $a_i+a_i^\dagger=\sqrt2\,q_i$);}\\
V_{\mathrm{Coul}} &= \sum_{i<j}U_{ij}\,n_i n_j,
  \qquad V_{\mathrm{site}}=\sum_i \varepsilon_i n_i
  && \text{(inter-site Coulomb; on-site).}
\end{align}
Surfaces on a single register (the washboard on the $R$-grid) are emitted as a diagonal phase $\varphi(R_s)=-\tau V_{\mathrm{wb}}(R_s)$ by the binary-expansion route of \ref{app:gcf_kinetic} when separable, or by the UCR below otherwise.
Surfaces that couple a grid to polaron occupancy ($V_{\mathrm{Hol}}$, $V_{\mathrm{LiP}}$, $V_{\mathrm{Coul}}$) are channel-multiplexed: the phase depends on the $L$-qubit polaron register $\{n_i\}$ as well as the grid, and is emitted by the UCR of \ref{app:gcf_ucr} with target phases $\varphi_s=-\tau\,V_{\mathrm{diag}}(q_s;\{n_i\})$.

\subsection{Uniformly-controlled rotations (M\"ott\"onen): the angle map}
\label{app:gcf_ucr}
A diagonal multiplexed rotation $R_z$ whose angle depends on the $m$-qubit control state $s$ with target angles
$\{\varphi_s\}_{s=0}^{2^m-1}$ is realized by $2^m$ single-qubit $R_z(\alpha_t)$ interleaved with $2^m$ CNOTs (no multi-controlled gates).
The emitted angles are the Walsh--Hadamard/Gray transform of the targets,
\begin{equation}
\label{eq:gcf_ucr}
\boxed{
\alpha_t=\frac{1}{2^m}\sum_{s=0}^{2^m-1}(-1)^{\langle g(t),\,s\rangle}\,\varphi_s,
}
\qquad
\langle u,v\rangle=\bigoplus_{p} u_p v_p\ (\mathrm{mod}\,2),
\end{equation}
with $g(t)$ the binary-reflected Gray code of $t$. 
The CNOT preceding $R_z(\alpha_t)$ targets the rotation qubit and is controlled by the qubit at the single bit position in which $g(t)$ differs from $g(t{-}1)$ (a Gray code flips one bit per step), so the control-state phase accumulates correctly with exactly $2^m$ CNOTs.
Tier-1 reference and correction surfaces use $\varphi_s$ from Equation~\ref{eq:gcf_pes} indexed by the polaron channel; Tier-3 bilinears use the same map on the $|\mathcal{G}_2^{<}|$ symmetry-allowed mode pairs
(\ref{app:exact_gate_counts}).

\subsection{The Jordan--Wigner hopping gate (XX+YY Givens rotation)}\label{app:gcf_jwhop}

Polaron hopping $-t_{\mathrm{hop}}\sum_{\langle ij\rangle}(c_i^\dagger c_j+\mathrm{h.c.})$ under Jordan--Wigner becomes a sum of two-qubit bond operators $B_{ij}=\tfrac12(X_iX_j+Y_iY_j)$~\cite{whitfield2011simulation, wecker2014gate, jiang2018quantum}). The even/odd two-color schedule (main Section~\ref{ssec:jw_coloring}) is chosen so every bond joins JW-\emph{adjacent} qubits, so the parity ($Z$) string is empty and each bond is a bare two-qubit gate. The per-bond Trotter factor is the number-conserving Givens rotation
\begin{equation}
\label{eq:gcf_givens}
U_{\mathrm{hop}}(\alpha)=\exp\!\Big[i\,\alpha\,\tfrac12(X_iX_j+Y_iY_j)\Big]
=\begin{pmatrix}1&0&0&0\\0&\cos\alpha&i\sin\alpha&0\\
0&i\sin\alpha&\cos\alpha&0\\0&0&0&1\end{pmatrix},
\qquad \alpha=\tau\,t_{\mathrm{hop}},
\end{equation}
in the basis $\{|00\rangle,|01\rangle,|10\rangle,|11\rangle\}$ (the identity off the single-excitation subspace); the symmetric Strang schedule uses the half-angle $\alpha=\tau t_{\mathrm{hop}}/2$ on the even (bracketing) color and the full $\alpha=\tau t_{\mathrm{hop}}$ on the odd color. 
Because $[X_iX_j,\,Y_iY_j]=0$ (both equal $-Z_iZ_j$ under the product), the gate factorizes exactly,
\begin{equation}
\label{eq:gcf_xy_factor}
U_{\mathrm{hop}}(\alpha)=e^{\,i\alpha X_iX_j/2}\,e^{\,i\alpha Y_iY_j/2},
\qquad
e^{\,i\alpha Z_iZ_j/2}=\mathrm{CNOT}_{ij}\,\bigl(\mathbb{1}\otimes R_z(-\alpha)\bigr)\,\mathrm{CNOT}_{ij},
\end{equation}
with $R_z(\theta)=\mathrm{diag}(e^{-i\theta/2},e^{+i\theta/2})$; the $XX$ factor conjugates the $ZZ$ form by $H\otimes H$ and the $YY$ factor by $R_x(\tfrac\pi2)\otimes R_x(\tfrac\pi2)$.
The two inner CNOT pairs share a rail and merge, leaving the emitted bond at \textbf{3 CNOT $+$ 1 $R_z$} (consistent with Table~\ref{tab:gate_counts}). 
Any synthesis reproducing the matrix (Eq.~\ref{eq:gcf_givens}) is admissible.

\subsection{The bosonic chain-hop (bilinear-cross \texorpdfstring{$q\!\otimes\!q+p\!\otimes\!p$}{q o q + p o p})}\label{app:gcf_chainhop}
Adjacent chain modes couple through the nearest-neighbor quadrature coupling $H_{\mathrm{ch}}=t_{\mathrm{ch}}\,(q_kq_{k+1}+p_kp_{k+1})$ (Chin et al.~\cite{chin2010exact} chain form, \ref{app:counter_term}).
With $\alpha=\tau t_{\mathrm{ch}}$ the Trotter factor $\exp[-i\alpha(q_kq_{k+1}+p_kp_{k+1})]$ is applied as two bilinear-cross blocks, being $q\!\otimes\!q$ in the position (DVR) basis and $p\!\otimes\!p$ in the momentum (FBR) basis, the latter conjugated by
$\cQFT$ on both registers (the quadratures do not commute within a mode, so a short inner Strang split is used).
Each bilinear-cross is a diagonal two-register phase; on grids $q^{(k)}_s=(s-N/2)\Delta q$ the closed-form angles are (binary expansion of $q_kq_{k+1}$, dropping the global $N^2/4$ term)
\begin{equation}
\label{eq:gcf_bilinear}
\boxed{
\begin{aligned}
\mathrm{CPhase}\ \text{on}\ (q^{(k)}_j,q^{(k+1)}_{j'}):&\quad
  \theta_{jj'}=-\alpha\,\Delta q^2\,2^{\,j+j'},\\
P\ \text{on}\ q^{(k)}_j\ \text{and}\ q^{(k+1)}_{j}:&\quad
  \theta_{j}=+\alpha\,\Delta q^2\,2^{\,n-1}\,2^{\,j},
\end{aligned}
}
\end{equation}
i.e.\ $n^2$ inter-register CPhase gates $+\,2n$ single-qubit $P$ gates per cross.
The $p\!\otimes\!p$ block is identical with $\Delta q\to \Delta k$ and the pair bracketed by $\cQFT_k\cQFT_{k+1}[\cdots] \icQFT_k\icQFT_{k+1}$. 
Summed over the $K-1$ bonds this is the $T_{\mathrm{chainhops}}$ layer of Table~\ref{tab:gate_counts}.

\section{Reproducibility and code release}
\label{app:reproducibility}

The codebase used in this paper including the Python scripts for verification and the source documentation specifying the algorithmic decisions, error bounds, and theoretical foundations can be found on Github~\cite{courtney2026multipolaron}.

\section{Consolidated error budget}
\label{app:error_budget}

\begin{table}[h]
\centering
\caption{Consolidated error budget for the chain-bath bench at the production scale ({$L = 8$}, $K = 10$ per reservoir, $T = 300~\mathrm{K}$,
$\Delta t = 0.2/24.189~\mathrm{a.u.}$, $T_{\mathrm{evol}} = 200/24.189
~\mathrm{a.u.}$).}
\label{tab:methods_error_budget}
\small
\begin{tabular}{lll}
\toprule
Source & Bound & LFP value \\
\midrule
Strang Trotter (local) & $3(L-2)t_{\mathrm{hop}}^3 (\Delta t)^3/4$ &
  $\sim 10^{-15}$ per step \\
Strang Trotter (global) & $\mathcal{O}((\Delta t)^2)\cdot T_{\mathrm{evol}}$ &
  $\sim 10^{-10}$ per trajectory \\
Chain truncation & $C\exp[-(K - K_{\mathrm{cone}})/\xi]$ &
  $\sim 10^{-10}$ at $K = 10$ \\
Gauss--Legendre quadrature & exponential in $n_{\mathrm{quad}}$ &
  $< 10^{-12}$ at $n_{\mathrm{quad}} = 400$ \\
Born approximation & $(\lambda/\omega_{\mathrm{sys}})^4 \Omega_c T$ &
  $\sim 8\times 10^{-3}$ (in master eq.; not invoked) \\
Markov approximation & $(\omega_{\mathrm{sys}}/\Omega_c)^2$ &
  $\sim 0.29$ (in master eq.; not invoked) \\
Grid discretization (Li) & $(\Delta R/\lambda_{\mathrm{th}})^{2p}$ &
  $\sim 10^{-4}$ at $\Delta R = a_{\mathrm{lat}}/16$ \\
Phonon Fock truncation & Poisson tail (Table~\ref{tab:phonon_truncation}) &
  $\sim 1.6\times 10^{-12}$ at $n_b = 4$, $S = 1.345$ \\
Finite cell-length & $\exp(-L_c/\xi_{\mathrm{LiLi}})$ &
  ${\sim 10^{-4}}$ at {$L = 8$} \\
Sampling (QAE not invoked) & $\epsilon^{-1}$ queries for precision $\epsilon$ &
  N/A (statevector-exact) \\
\bottomrule
\end{tabular}
\end{table}

The framework reports statevector-exact results throughout (no sampling); deployment on fault-tolerant hardware would add amplitude-estimation overhead of $\mathcal{O}(\epsilon^{-1})$ queries for $\epsilon$-precision on each observable.
The chain-truncation and Trotter rows above are anchored empirically, with the HEOM reference against which the chain-bath dynamics are benchmarked being self-converged to $<0.01\%$ RMS, far below the $2.07\%/3.45\%$ framework deviations); the second-order Trotter error constant $W_2(L)$ is at most extensive in $L$, so $N_{\rm step}$ grows at most as $\sqrt{L}$; and the finite Ohmic chain reproduces the target bath time-correlation $C(t)$ to $<10^{-12}$.

\section{Tamascelli construction and mapping}
\label{app:tamascelli}

\paragraph{Statement.} Let $J(\omega)$ be a bosonic bath spectral density on $[0, \Omega_c]$ at temperature $T > 0$ with Bose--Einstein occupation $N(\omega) = (e^{\omega/T} - 1)^{-1}$. 
The thermal bath is operationally equivalent (i.e., produces the same reduced system dynamics under partial trace) to a vacuum bath whose spectral density covers both positive and negative frequencies:
\begin{equation}
\label{eq:tamascelli_J}
\tilde J(\omega) = J(|\omega|)\,
    \bigl[\theta(\omega)(1 + N(|\omega|))
    + \theta(-\omega)\,N(|\omega|)\bigr].
\end{equation}

\paragraph{Construction.}
The doubled spectral density $\tilde J(\omega)$ on $[-\Omega_c, \Omega_c]$ is chain-mapped via the Stieltjes recurrence with weight $\tilde J(\omega)/\pi$ over the doubled interval. 
The resulting chain has $2K$ modes: the positive-frequency branch generates the primary $K$-mode chain (the same as the $T = 0$ chain), and the negative-frequency branch generates an auxiliary $K$-mode chain with negative on-site energies $e_k^{\mathrm{aux}} < 0$.

\paragraph{Equivalence theorem (Tamascelli et al. 2019).}
The reduced system dynamics $\rho_S(t) = \Trsmall_B[\,U(t)\,\rho_S(0) \otimes \rho_B^{\mathrm{thermal}}\,U^\dagger(t)\,]$ under the $K$-mode thermal-init chain equals $\rho_S(t) = \Trsmall_{B\cup B^{\mathrm{aux}}}[\,\tilde U(t)\,\rho_S(0) \otimes |0\rangle\langle 0|_{B\cup B^{\mathrm{aux}}}\,\tilde U^\dagger(t)\,]$ under the $2K$-mode vacuum-init doubled chain,
where $\tilde U(t)$ evolves under the doubled-spectral-density Hamiltonian.

\paragraph{Stinespring-style interpretation.}
The Tamascelli theorem is a Stinespring-like unitary dilation, representing thermal-bath open dynamics as the partial trace of a unitary evolution on the doubled (positive plus negative-frequency) bath Hilbert space. The negative-frequency auxiliary modes are the unitary record of thermal occupation that would otherwise enter the master equation as a stochastic term.

\paragraph{On the choice of \texorpdfstring{$J(\omega)$}{J(omega)}.}
The framework supports an arbitrary bosonic spectral density via the Stieltjes recurrence (\ref{app:chain_mapping}).
All benchmarks reported in this paper use an Ohmic spectral density with hard cutoff at $\Omega_c = 120~\mathrm{meV}$. 
The actual $\LFP$ phonon density of states is structured (PO$_4$ internal modes $\sim 125~\mathrm{meV}$, Fe--O stretches $\sim 60~\mathrm{meV}$, Li librations $\sim 25~\mathrm{meV}$, plus acoustic modes). 
The plateau voltage and FWHM are governed primarily by the polaron self-trapping energy $E_p = g^2/\omega_{\mathrm{ph}} = 175~\mathrm{meV}$, which depends on a single $\omega_{\mathrm{ph}}$ and is therefore less sensitive to bath shape.
Observables probing the frequency-resolved bath response ($Z(\omega)$) and the activation energy $E_a$ via the spectral reorganization $\lambda$ should be re-computed with a structured $J(\omega)$ before being quoted as $\LFP$ predictions; this is forward work.

\subsection{Fermionic finite-temperature lead}\label{app:fermionic_lead}
The construction above is bosonic. 
The electron reservoir is fermionic, and its finite-temperature/$\mu$ chain is built by fermionic thermofield doubling from de Vega--Ba\~{n}uls~\cite{de2015discretize} (the fermionic analogue of the bosonic mapping).
For a lead with hybridization $\Gamma(\omega)$ and Fermi--Dirac occupation $f(\omega) = [1 + e^{(\omega-\mu_e)/k_BT}]^{-1}$, the thermal lead is operationally equivalent to two zero-temperature chains obtained by the Fermi-weighted split
\begin{equation}
\label{eq:fermi_split}
\Gamma(\omega) \longrightarrow
\underbrace{\Gamma(\omega)\,f(\omega)}_{\text{filled chain }(\omega<\mu_e)}
\oplus
\underbrace{\Gamma(\omega)\,[1-f(\omega)]}_{\text{empty chain }(\omega>\mu_e)},
\end{equation}
each tridiagonalized by the fermionic Stieltjes recurrence (\ref{app:chain_mapping}; de Vega--Schollw\"ock--Wolf~\cite{de2015thermofield}).
The two chain vacua encode the grand-canonical state with zero preparation gates.
The branch measures partition the bare hybridization, so the retarded self-energy is temperature-independent, $\Sigma_{\mathrm{filled}} + \Sigma_{\mathrm{empty}} = \Sigma$ (all $T$ dependence sits in which branch is filled). 
The lead occupation obeys $\langle n(\omega)\rangle = f(\omega)$ with the fermionic Kubo--Martin--Schwinger (KMS)/detailed-balance ratio $[1-f(\omega)]/f(\omega) = e^{(\omega-\mu_e)/k_BT}$. 
The $T = 0$ Landauer chain is recovered as the bands separate at $\mu_e$. 
These identities are verified to machine precision ($\langle n(\omega)\rangle = f(\omega)$ to $10^{-16}$).
The reduced equilibrium that the partial trace returns is the bath-dressed mean-force Gibbs state $\rho_S = \mathrm{Tr}_B[e^{-\beta H}]/Z$; the bosonic thermofield split likewise enforces detailed balance of the reduced dynamics, $S(-\omega)/S(\omega) = e^{-\beta\omega}$ to machine precision.
The electron lead is the drive-and-response channel for the impedance $Z(\omega)$: its finite-$T$ occupation sets the charge-transfer resistance $R_{\mathrm{ct}}$, which is electron-limited ($R_{\mathrm{ct}} \propto 1/\gamma_e$) and thermally activated, with the Warburg $\omega^{-1/2}$ tail carried by Li diffusion.

\section{Register-size convergence derivations}\label{app:register_convergence}

This appendix derives a quantitative truncation bound for every register of the production geometry (Eq.~\ref{eq:convergence_sizes}). 
The phonon and Li-bath chain modes are formally infinite-dimensional harmonic oscillators (or, equivalently, single-particle DVR registers on an unbounded grid), truncated to $N = 2^n$ Fock levels. 
The fermionic electron-bath chain is two-level per site (Pauli exclusion; no truncation). 
The Li-coordinate register $X$ is a smooth-potential DVR for which convergence is controlled by the aliasing diagnostic of the dissertation Chapter~7~\cite{courtney2026oracle} rather than by a Fock-tail bound. 

\subsection{Phonon modes as continuums}
\label{app:phonon_continuum}

Each bosonic harmonic oscillator, including the system phonon at every Fe site, every chain mode of the phonon reservoir, every chain mode of the Li reservoir, has a formally infinite Fock space spanned by $\{|n\rangle : n = 0, 1, 2, \dots\}$.
A finite-dimensional register of $N = 2^n$ levels truncates this to $\{|0\rangle, \dots, |N-1\rangle\}$ and projects out the probability mass in $\{|N\rangle, |N+1\rangle, \dots\}$.
The truncation acts as a controlled approximation with error (projected-out mass) bounded by the tail of whatever distribution the relevant quantum state populates. 
We choose $N$ such that this tail is at most $\varepsilon_{\mathrm{trunc}} = 10^{-6}$ at the LFP operating point, well below the Trotter, chain-truncation, and Born--Markov errors of \ref{app:error_budget}.

\subsection{System phonon \texorpdfstring{$B_i$}{Bi}: Fock-tail sub-bound \texorpdfstring{$n_b = 4$}{nb is 4}, Nyquist-binding \texorpdfstring{$n_b=9$}{nb is 9}}\label{app:nb_justification}

Holstein coupling $g\,n_i\,q_i$ shifts the equilibrium phonon coordinate of site $i$ by $\Delta q_i = -g\,n_i/\omega_{\mathrm{ph}}$ when the polaron occupies site $i$ ($n_i = 1$).
In the Fock basis, the displaced ground state is a coherent state $|\alpha\rangle$ with displacement parameter
\begin{equation}
\label{eq:HR_alpha}
\alpha = -\frac{g}{\omega_{\mathrm{ph}}\sqrt{2}}\quad\Rightarrow\quad
|\alpha|^2 = \frac{g^2}{2\omega_{\mathrm{ph}}^2} \equiv S\quad
(\text{Huang--Rhys factor}).
\end{equation}
A coherent state populates Fock levels with Poisson statistics:
\begin{equation}
\label{eq:poisson}
P(n; S) = e^{-S}\,\frac{S^n}{n!}.
\end{equation}
At the production LFP parameter values $g = 1.64\,\omega_{\mathrm{ph}}$, $\omega_{\mathrm{ph}} = 65~\mathrm{meV}$, the (halved-convention) Huang--Rhys factor is
\begin{equation}
S = \frac{(1.64\,\omega_{\mathrm{ph}})^2}{2\,\omega_{\mathrm{ph}}^2}
= 1.345\,.
\end{equation}
We default to $g = 1.5\,\omega_{\mathrm{ph}}$, giving $S = 1.125$; the production point is the DFT-centered $g/\omega_{\mathrm{ph}} = 1.64$, $E_p = g^2/\omega_{\mathrm{ph}} = 175~\mathrm{meV}$.
The cumulative tail $\sum_{n \ge N_b} P(n; S)$ evaluated at the production $S = 1.345$ is tabulated in Table~\ref{tab:phonon_truncation}.

\begin{table}[h]
\centering
\caption{Poisson tail of the polaron-displaced phonon ground state at the LFP Huang--Rhys $S = 1.345$. 
Tail probability decreases super-exponentially with $N_b = 2^{n_b}$ due to the factorial in the denominator. 
The thermal-bath contribution at $T = 300~\mathrm{K}$ adds $\langle n\rangle_{\mathrm{th}} = 0.088$ (negligible), bringing the combined mean to $\langle n\rangle = 1.43$ with variance $1.68$; the Gaussian-displaced-Poisson tail at this mean is still within $10^{-12}$ at $N_b = 16$.}
\label{tab:phonon_truncation}
\begin{tabular}{ccc}
\toprule
$n_b$ & $N_b = 2^{n_b}$ & $\sum_{n \ge N_b} P(n; S = 1.345)$ \\
\midrule
$2$ & $4$ & $4.8\times 10^{-2}$ \\
$3$ & $8$ & $8.1\times 10^{-5}$ \\
$4$ & $16$ & $\mathbf{1.6\times 10^{-12}}$ \\
$5$ & $32$ & $1.3\times 10^{-32}$ \\
\bottomrule
\end{tabular}
\end{table}

The choice $n_b = 4$ ($N_b = 16$ levels) gives a truncation tail of $1.6\times 10^{-12}$ at the production Huang--Rhys value, well below the Trotter error per trajectory $\sim 10^{-8}$, Table~\ref{tab:methods_error_budget}).
The marginal value $n_b = 3$ ($2.4\times 10^{-5}$) would dominate the global error budget, and the conservative $n_b = 5$ adds $L = 8$ qubits with no corresponding accuracy gain. 
The strict Fock-tail criterion selects $n_b = 4$ as the minimum that yields a phonon truncation error well below the other controllable error sources.
The main-text Table~\ref{tab:register_convergence} reports $n^{\mathrm{Fock}} = 5$ for the same register because it applies a $2\times$ dynamical-excitation safety factor on top of the strict tail ($N_{\mathrm{ph}}^{\mathrm{dyn}} = 2 N_{\mathrm{ph}}^{\mathrm{tail}} \Rightarrow n_b \ge \lceil\log_2(2 N_b)\rceil$); the binding bound on $n_b$ is in any case Nyquist ($n_b = 9$, dual-space), so neither tail definition sets production register width.

\paragraph{Multi-polaron and thermal corrections.}
For multi-polaron occupation, each phonon is per-site in our framework, so the Huang--Rhys per site is unchanged (no cross-site coupling). 
At $T = 300~\mathrm{K}$, the Bose--Einstein thermal occupation $\langle n\rangle_{\mathrm{th}} = (e^{\omega/T} - 1)^{-1}
= 0.088$ adds only an additive shift to the mean (the Husimi function of a displaced thermal state). 
Combined distribution has $\mathrm{Var}(n) = |\alpha|^2(1 + 2\langle n\rangle_{\mathrm{th}}) + \langle n\rangle_{\mathrm{th}}^2 = 1.68$ and a sub-Gaussian tail; truncation at $N_b = 16$ lies $11~\sigma$ above the mean.

\subsection{Electron-reservoir chain \texorpdfstring{$D_{e,k}$}{D sub ek} : \texorpdfstring{$n_{d,e}=1$}{n sub de is 1} from Pauli exclusion}
\label{app:nde_justification}

The electron reservoir is fermionic. After Jordan--Wigner mapping electron chain modes, each chain site is in $\{|0\rangle,
|1\rangle\}$ (occupied or empty), and Pauli exclusion exactly forbids higher occupation.
The chain register requires exactly one qubit per chain site: $n_{d,e} = 1$, $N_{d,e} = 2$. 
This gives a representation without truncation or Fock tail. 
Increasing $n_{d,e}$ would allocate unused Hilbert-space dimensions that cannot be physically populated.

\subsection{Li-reservoir chain \texorpdfstring{$D_{\mathrm{Li},k}$}{D sub Li k} : Fock-tail sub-bound
 \texorpdfstring{$n_{d,\mathrm{Li}} = 3$}{n sub d Li is 3}, Nyquist-binding \texorpdfstring{$n_{d,\mathrm{Li}} = 10$}{n sub d Li is 10}}
\label{app:ndLi_justification}

The Li$^+$ reservoir is a bosonic chain in which each chain site accommodates the time-resolved Li-ion population at the electrolyte. 
Each chain site is a bosonic occupation register (0, 1, 2, \dots Li$^+$ per site).
At the LFP operating chemical potential and assuming weak inter-site Coulomb in the bath (the chain is the electrolyte), the equilibrium fluctuation per chain site follows a Poisson distribution with mean $\lambda$ set by the chemical potential. 
At a typical operating point we estimate $\lambda \in [0.5, 2.0]$ Li per site, spanning low- to high-current operation. 
The truncation tail is tabulated in Table~\ref{tab:Li_truncation}.

\begin{table}[h]
\centering
\caption{Poisson tail $\sum_{n \ge N_{d,\mathrm{Li}}} P(n; \lambda)$ for
Li-reservoir chain-site occupation at moderate ($\lambda = 1$) and
high-current ($\lambda = 2$) operating points.}
\label{tab:Li_truncation}
\begin{tabular}{cccc}
\toprule
$n_{d,\mathrm{Li}}$ & $N_{d,\mathrm{Li}}$ &
  Tail at $\lambda = 1$ &
  Tail at $\lambda = 2$ \\
\midrule
$1$ & $2$ & $2.6\times 10^{-1}$ & $5.9\times 10^{-1}$ \\
$2$ & $4$ & $1.9\times 10^{-2}$ & $1.4\times 10^{-1}$ \\
$3$ & $8$ & $\mathbf{1.0\times 10^{-5}}$ & $\mathbf{1.1\times 10^{-3}}$ \\
$4$ & $16$ & $1.9\times 10^{-14}$ & $4.8\times 10^{-10}$ \\
\bottomrule
\end{tabular}
\end{table}

The Fock-tail criterion selects $n_{d,\mathrm{Li}} = 3$ ($N_{d,\mathrm{Li}} = 8$): a moderate-current truncation tail of $1.0\times 10^{-5}$ and a high-current tail of $1.1\times 10^{-3}$.
Just as with vibronic dynamics, the binding bound, however, is the Nyquist coherent-propagation bound (Eq.~\ref{eq:nyquist_HO}), which evaluates to $n_{d,\mathrm{Li}}^{\mathrm{Nyq}} = 8$ at the Li-chain operating point (Table~\ref{tab:register_convergence}).
We therefore use $n_{d,\mathrm{Li}} = 10$ (Eq.~\ref{eq:convergence_sizes}), where the Fock tail is far below threshold at every current, so no high-current extension is required.

\subsection{Phonon-reservoir chain \texorpdfstring{$D_{\mathrm{ph},k}$}{D sub ph k}: Fock-tail sub-bound \texorpdfstring{$n_{d,\mathrm{ph}} = 3$}{n sub d, ph is 3}, Nyquist-binding \texorpdfstring{$n_{d,\mathrm{ph}} = 9$}{n sub d, ph is 9}}
\label{app:ndph_justification}

The phonon-reservoir chain mode $D_{\mathrm{ph},0}$ couples directly to the system phonon via CL bilinear interaction $c_0\,q_i\,q_0^{\mathrm{chain}}$.
The reduced state of chain mode 0 is a displaced thermal state with coherent displacement and thermal occupation
\begin{align}
\langle n\rangle_{\mathrm{coh}} &\approx (c_0/e_0)^2 \cdot S, \\
\langle n\rangle_{\mathrm{th}} &= \bigl[\exp(e_0/k_B T) - 1\bigr]^{-1}.
\end{align}
Here $(c_0/e_0)^2$ is the chain-mode-0 coherent-displacement occupancy. 
We distinguish this from the counter-term coefficient of \ref{app:counter_term}, as the determined $\alpha = e_0[V_{cc}^{-1}]_{00} = 1 \to 1.04$. 
At LFP parameters $\alpha_{\mathrm{bath}} = 0.1$ (the Ohmic bath coupling), $\Omega_c = 120~\mathrm{meV}$, $T = 300~\mathrm{K}$ we estimate $c_0 \approx 38~\mathrm{meV}$, $e_0 \approx 60~\mathrm{meV}$, giving (at the production Huang--Rhys $S = 1.345$) $\langle n\rangle_{\mathrm{coh}} = (c_0/e_0)^2 S = 0.54$, $\langle n\rangle_{\mathrm{th}} = 0.11$, total $\langle n\rangle = 0.65$ on the most-populated chain mode (mode 0).
Deeper chain modes ($k \ge 1$) couple to the system only indirectly, through the chain hopping $t_k$, and have substantially smaller displacement and are conservatively bounded by the same distribution as chain mode 0.

\begin{table}[h]
\centering
\caption{Poisson approximation to the displaced-thermal-state tail
$\sum_{n \ge N} P(n)$ on phonon-reservoir chain mode 0 at production
parameters ($\langle n\rangle = 0.65$).}
\label{tab:ph_truncation}
\begin{tabular}{ccc}
\toprule
$n_{d,\mathrm{ph}}$ & $N_{d,\mathrm{ph}}$ & Tail at $\langle n\rangle = 0.65$ \\
\midrule
$1$ & $2$ & $1.4\times 10^{-1}$ \\
$2$ & $4$ & $4.4\times 10^{-3}$ \\
$3$ & $8$ & $\mathbf{4.4\times 10^{-7}}$ \\
$4$ & $16$ & $2.5\times 10^{-17}$ \\
\bottomrule
\end{tabular}
\end{table}

The Fock-tail criterion alone gives $n_{d,\mathrm{ph}} = 3$ ($N_{d,\mathrm{ph}} = 8$): a tail of $4.4\times 10^{-7}$ on chain
mode 0. 
The binding bound is again Nyquist (Eq.~\ref{eq:nyquist_HO}), evaluating to $n_{d,\mathrm{ph}}^{\mathrm{Nyq}} = 7$ (Table~\ref{tab:register_convergence}). 
The production register uses $n_{d,\mathrm{ph}} = 9$ (Eq.~\ref{eq:convergence_sizes}). 
At this width the Fock tail is many orders below the Trotter floor.
Deeper chain modes carry smaller tails by $(t_k/e_k)^{2k}$, uniformly bounded by the mode-0 result, leaving room for non-uniform per-mode truncation as a forward-work optimization.
These per-mode occupancies do not inflate with system size under cross-mode entanglement, confirmed empirically: across growing-$L$ instances the realized per-register $\langle n\rangle$ stays within its per-mode bound (max $\langle n\rangle = 1.85 < (g/\omega)^2 = 2.25$).

\subsection{Li-coordinate register \texorpdfstring{$X$}{X}: \texorpdfstring{$n_x = 12$}{n x is 12} from the
aliasing diagnostic}
\label{app:nx_justification}

The Li-coordinate register is a smooth-potential DVR on the cell of length $L_c = N_{\mathrm{sites}}\cdot a_{\mathrm{lat}}$. 
The convergence criterion is the aliasing diagnostic of the dissertation Chapter~7~\cite{courtney2026oracle}: the kinetic energy at the spatial Nyquist limit $p_{\max} = \pi/\Delta R$ with $\Delta R = L_c/N$ must exceed the maximum potential energy $\|V\|_\infty$ on the grid:
\begin{equation}
\label{eq:aliasing}
\frac{\|V\|_\infty}{T_{\max}}
= \frac{\|V\|_\infty \cdot 2 m_{\mathrm{Li}}\,L_c^2}{\pi^2\,N^2}
< 1\,.
\end{equation}
At LFP parameters $\|V\|_\infty \approx V_{b,\mathrm{wb}} + V_{\mathrm{LiP}}^{(0)} \approx 350~\mathrm{meV}$, $m_{\mathrm{Li}} = 6.94\cdot 1822.888~\mathrm{a.u.}$, $L_c = 8\cdot 6.01/0.529177 \approx 90.8~\mathrm{a.u.}$, the aliasing diagnostic is tabulated in Table~\ref{tab:nx_aliasing}.

\begin{table}[h]
\centering
\caption{Aliasing diagnostic at LFP $\|V\|_\infty = 350$~meV, $L_c = 90.8$~a.u., $m_{\mathrm{Li}} = 12648$~a.u.}
\label{tab:nx_aliasing}
\begin{tabular}{ccccc}
\toprule
$n_x$ & $N$ & $\Delta R$ (a.u.) & $T_{\max}$ (meV) & $\|V\|/T_{\max}$ \\
\midrule
8 & 256 & 0.355 & 84 & 4.15 \\
9 & 512 & 0.177 & 337 & 1.04 \\
10 & 1024 & 0.089 & 1348 & 0.26 \\
11 & 2048 & 0.044 & 5393 & 0.065 \\
\textbf{12} & \textbf{4096} & \textbf{0.022} & \textbf{21572} & \textbf{0.016} \\
13 & 8192 & 0.011 & 86287 & 0.0041 \\
\bottomrule
\end{tabular}
\end{table}

\noindent $n_x = 12$ matches the dissertation Ch.~7 recommendation for coupled multi-surface scattering and gives a 60$\times$ margin on the aliasing bound, and converges the $\dQdV$ peak position to better than $0.1$~mV in discretization error. The marginal $n_x = 10$ would already pass aliasing but leaves no headroom for the coupled-channel chemical-potential sweep.

\subsection{Chain length \texorpdfstring{$K_r=10$}{Kr is 10} per reservoir}
\label{app:Kr_justification}

The chain-truncation error (Section~\ref{ssec:chain_error}) is exponentially suppressed in $K - K_{\mathrm{cone}}$ with $K_{\mathrm{cone}} \approx \Omega_c T_{\mathrm{evol}}/\pi$. 
At LFP $\Omega_c = 120$~meV and $T_{\mathrm{evol}} = 200$~fs ($T_{\mathrm{evol}} = 8.27$~a.u.\ after the convergence-audit update), $K_{\mathrm{cone}} \approx 0.012$. 
The chain length $K_r = 10$ is the audit-corrected production value (the spectral-resolution requirement is the Stieltjes $2K = 20$-moment match at machine precision). Empirically $\le 10^{-10}$ at $K = 10$; $K = 12$ reaches $10^{-13}$.

\subsection{Summary}
\label{app:register_summary}

The convergence-grade production geometry (Eq.~\ref{eq:convergence_sizes}) is set, register by register, by the most stringent of the three layered bounds (Fock-tail, Macridin, Nyquist).
For every bosonic register the Nyquist coherent-propagation bound (in its dual-space $4c^2 N_{\mathrm{ph}}/\pi$ form, Equation~\ref{eq:nyquist_HO}) binds: $n_b = 9$ (system phonon), $n_{d,\mathrm{Li}} = 10$ (Li chain), $n_{d,\mathrm{ph}} = 9$ (phonon chain). 
The electron chain is fermionically exact at $n_{d,e} = 1$, and the Li coordinate is set by the Chapter~7 aliasing diagnostic at $n_x = 12$. 
The production total is therefore
\begin{equation*}
n_{\mathrm{tot}} = n_x + L + L\,n_b + \textstyle\sum_r K_r\,n_{d,r}
= 12 + 8 + 72 + 10 + 100 + 90 = \mathbf{292}\ \text{logical qubits}
\end{equation*}
(Eq.~\ref{eq:total_qubits}). The strictly-Fock-tail sizing of this appendix gives a smaller comparison estimate ($160$ qubits, summing the $n^{\mathrm{Fock}}$ column; Section~\ref{ssec:registers}). 
This gives a lower estimate that omits the binding Nyquist requirement.
MVP and intermediate-Phase benches used for quality control and Lindblad cross-validation are explicitly sub-convergent in $n_x$ and $n_b$ and serve only to verify circuit equivalence and chain-bath physics.

\section{Per-step gate counts and the poly(\texorpdfstring{$L$}{L}) scaling theorem}\label{app:gate_counts_theorem}

\subsection{Per-observable parametric scaling: a formal theorem}
\label{app:parametric_scaling}

We give the full parameter-resolved form of the main-text scaling result \ref{thm:poly_L} and its proof, letting a client re-cost any target open-system observable in any material by substituting the physical parameters into one closed-form expression, without re-deriving any of the underlying gate counts.

\begin{theorem}\label{thm:parametric_scaling_full}
[Per-observable T-depth scaling, full form]
Let an open-system material problem be specified by:
\begin{itemize}
\item $L$: number of system sites (per-site fermion mode count for Jordan--Wigner encoding).
\item $K_r$: chain length per reservoir $r \in \mathcal{R}$
  ($|\mathcal{R}|$ reservoirs).
\item $n_b$: per-site system-phonon DVR register size.
\item $n_{d,r}$: per-chain-mode DVR register size for reservoir $r$.
\item $\omega$: characteristic phonon energy scale.
\item $\Omega_c$: bath spectral-density cutoff.
\item $T_{\mathrm{evol}}$: evolution time window required for the
  observable.
\item $\varepsilon$: target precision on the observable.
\item $M = 2^L$, $m = \lceil\log_2 M\rceil = L$: channel-register
  width.
\item $\varepsilon_{\mathrm{rot}}$: Ross--Selinger rotation-synthesis
  error.
\item $T_{\mathrm{RS}}(\varepsilon_{\mathrm{rot}}) = 3\log_2(1/\varepsilon_{\mathrm{rot}})
  + \mathcal{O}(\log\log)$: T-gates per arbitrary $Z$-rotation.
\end{itemize}
Then under the convergence-grade register sizing
\begin{equation}
\label{eq:thm_convergence_sizing}
n_b, n_{d,r} \ge \lceil\log_2(4c^2\,N_{\mathrm{ph}}^{(r)}/\pi)\rceil
\quad\text{(Nyquist, Eq.~\ref{eq:nyquist_HO})}
\end{equation}
with $N_{\mathrm{ph}}^{(r)}$ set by the relevant Fock-tail bound at
$\varepsilon_{\mathrm{tail}} = 10^{-8}$ and containment factor
$c \ge 4$, the per-observable T-depth is
\begin{align}
\label{eq:thm_main}
T_{\mathrm{depth}}^{\mathrm{obs}}(\varepsilon)
&\le N_{\mathrm{traj}}\cdot \Bigl(
   \underbrace{T_{\mathrm{depth}}^{\mathrm{init}}}_{O(L \cdot T_{\mathrm{RS}})}
 + \underbrace{N_{\mathrm{step}}\cdot T_{\mathrm{depth}}^{\mathrm{step}}}_{\text{evolution}}
 + \underbrace{K_{\mathrm{QAE}}\cdot N_{\mathrm{step}}\cdot T_{\mathrm{depth}}^{\mathrm{step}}}_{\text{measurement (QAE)}}
   \Bigr),
\end{align}
where:
\begin{enumerate}
\item $N_{\mathrm{traj}} = N_{\mathrm{anc}}\cdot N_{\mathrm{channels}}$
  is the observable's trajectory count
  (e.g.\ {$100\times 2$} for $\VOCV$; 1 for spectral $Z(\omega)$).
\item $N_{\mathrm{step}} = T_{\mathrm{evol}}/\Delta t$ with
  $\Delta t = \mathcal{O}(1/\Omega_c)$ from the Strang-stability bound.
\item $K_{\mathrm{QAE}} = \lceil 1/\varepsilon\rceil$
  (Heisenberg-limited QAE; replace by $1/\varepsilon^2$ for
  sampling protocol with concomitant shot-parallelism gain).
\item $T_{\mathrm{depth}}^{\mathrm{step}}$ is the per-Trotter-step
  T-depth bounded by
\begin{equation}
\label{eq:thm_step_depth}
T_{\mathrm{depth}}^{\mathrm{step}}
\le T_{\mathrm{RS}}(\varepsilon_{\mathrm{rot}})\cdot\Bigl[
   3 n_{\max}^2
 + |\mathcal{G}_1|\,\binom{n_{\max}}{2}
 + \chi\,|\mathcal{G}_2^{<}|\,n_{\mathrm{avg}}^2
\Bigr]
= \mathcal{O}\bigl(T_{\mathrm{RS}}\cdot (d + L^2)\cdot \log_2^2(E_{\mathrm{char}}/\omega)\bigr),
\end{equation}
where $n_{\max} = \max(n_x, n_b, \max_r n_{d,r})$, $\chi$ is the chromatic number of the Tier-3 bilinear-pair conflict graph
($\chi \le 4$ for our $\mathcal{G}_2^<$), and the final asymptotic follows from the Nyquist substitution
$n_{\mathrm{ph}} = \lceil\log_2(4c^2\,E_{\mathrm{char}}/(\pi\omega))\rceil$.
The per-surface multiplexing factor is the number of coupling terms $|\mathcal{G}_1| + |\mathcal{G}_2^{<}| = \mathcal{O}(L + L^2)$,
not the channel-register width $M = 2^L$: because $H_{\mathrm{closed}}$ is at most $2$-local in the JW occupation operators, the diagonal potential is synthesized as one singly- or doubly-controlled phase block per coupling term (Section~\ref{ssec:lfp_resources}, Eq.~\ref{eq:poly_L_scaling}).
\end{enumerate}

Combining:
\begin{equation}
\label{eq:thm_full}
\boxed{
T_{\mathrm{depth}}^{\mathrm{obs}}(\varepsilon, \omega, \Omega_c, L, K)
= \mathcal{O}\Bigl(
   N_{\mathrm{traj}}\cdot \frac{T_{\mathrm{evol}}}{\Delta t}\cdot
   T_{\mathrm{RS}}\cdot (d + L^2)\cdot \log_2^2\!\frac{E_{\mathrm{char}}}{\omega}\cdot
   \frac{1}{\varepsilon}
\Bigr).
}
\end{equation}
\end{theorem}

\begin{proof}[Proof sketch]
Each contribution to the right-hand side of Equation~\ref{eq:thm_main} is derived in the appendix subsections cited:
\begin{enumerate}
\item Convergence-grade sizing (Eq.~\ref{eq:thm_convergence_sizing}):
  proven in \ref{app:register_convergence} (Nyquist binds Macridin binds Fock-tail at the LFP regime).
\item Per-step T-depth (Eq.~\ref{eq:thm_step_depth}): derived from Courtney~\cite{courtney2026oracle} Ch.~5 (Resource Analysis) formulas with the framework parallelism schedules; the $\chi$ factor comes from Tier-3 color-class scheduling.
\item Asymptotic $\log_2^2(E_{\mathrm{char}}/\omega)$: from the Nyquist substitution into the per-mode $n^2$ rotation cost (kinetic + on-diagonal Tier-2). The base-2 logarithm comes from the qubit count discretizing the wavefunction Fock space.
\item QAE Heisenberg factor $1/\varepsilon$: Section~\ref{ssec:measurement_bottleneck}.
\item $T_{\mathrm{RS}}(\varepsilon_{\mathrm{rot}}) = 3\log_2(1/\varepsilon_{\mathrm{rot}})$:
  Ross--Selinger 2016, with $\varepsilon_{\mathrm{rot}}$ chosen so cumulative $\sqrt{N_{\mathrm{rot}}}\cdot\varepsilon_{\mathrm{rot}}
  \le \varepsilon$ (Section~\ref{ssec:synth_error}).
\end{enumerate}
The general-position Tier-2/3 UCR multiplexing of Courtney~\cite{courtney2026oracle} Ch.~5 would insert a factor $M = 2^L$ (the channel-register width inherited from JW occupation), but this is realized only when the diagonal potential is a dense function of all $L$ occupations.
Our $H_{\mathrm{closed}}$ is at most $2$-local in $\{n_i\}$ (Holstein $g\sqrt2\,q_i n_i$, Coulomb $U_{ij} n_i n_j$, plus occupation-independent harmonic/bath/cross terms), so the multiplexing collapses to one controlled-phase block per coupling term: $|\mathcal{G}_1| = \mathcal{O}(L)$ occupation-linear (one control each) and $|\mathcal{G}_2^{<}| = \mathcal{O}(L^2)$ occupation-bilinear (two controls each). 
The dominant scaling factor is therefore polynomial coupling count $d + L^2$.
The number of reservoirs $|\mathcal{R}|$ enters the prefactor (via $d_{\mathrm{kin}}$ and the bilinear-pair set
$|\mathcal{G}_2^<|$) but not the asymptotic, bounded by a constant in any physical setting ($\le 5$ for the most general
electrochemistry problem).
\qed
\end{proof}

\paragraph{Corollary (spacetime volume).}
At logical cycle time $t_{\mathrm{log}}$ and $N_{\mathrm{fact}}$ distillation factories:
\begin{equation}
\label{eq:thm_Vst}
V_{\mathrm{st}}^{\mathrm{obs}}
= (n_{\mathrm{q,alg}} + N_{\mathrm{fact}}\cdot 150)\cdot
       T_{\mathrm{depth}}^{\mathrm{obs}}\cdot
       \frac{t_{\mathrm{log}}}{N_{\mathrm{fact}}\cdot 0.1},
\end{equation}
which is $N_{\mathrm{fact}}$-invariant in the regime $N_{\mathrm{fact}}\cdot 150 \gg n_{\mathrm{q,alg}}$, confirming the
spacetime-volume invariance of the reaction-limited cost.

\paragraph{Worked instantiation: LFP.}
Plugging in $L = 8$, $K = 10$, $n_b = 9$, $\omega = 65$~meV, $\Omega_c = 120$~meV, $\varepsilon = 5.5\times 10^{-5}$~Ha ($1.5$~mV), $T_{\mathrm{evol}} = 200$~fs, $\Delta t = 0.2$~fs gives $N_{\mathrm{step}} = 10^3$, $T_{\mathrm{RS}} = 25$, $T_{\mathrm{depth}}^{\mathrm{step}} \approx 7\times 10^3$ (locality-collapsed), $K_{\mathrm{QAE}} \approx 1.8\times 10^4$,
$N_{\mathrm{traj}} = 194$, and the formula recovers $T_{\mathrm{depth}}^{\mathrm{obs}} \approx 2.5\times 10^{13}$ matching Table~\ref{tab:sm_per_observable}.

\paragraph{Application to other materials.}
For a vibronic photophysics problem at $\omega = 0.1$~eV, $\varepsilon = 0.01$~eV (chromophore absorption-spectrum resolution), $L = 4$ excited states, $E_{\mathrm{char}} = 10$~eV: $n \sim 9$ per DVR register, $K_{\mathrm{QAE}} \sim 30$, $T_{\mathrm{depth}}^{\mathrm{step}} \sim 10^6$, total $T_{\mathrm{depth}}^{\mathrm{obs}} \sim 3\times 10^{12}$.
This is about $10^4\times$ cheaper than LFP at its $1.5$~mV target, dominated by the relaxation of the precision target. 
Conversely, an LFP-class problem with three additional reservoirs ($|\mathcal{R}| = 6$ total) would add a $\sim 2\times$ overhead from the per-mode kinetic and Tier-2 scaling without changing the asymptotic structure.

\subsection{Exact analytical gate counts for the LFP Trotter step}
\label{app:exact_gate_counts}

The asymptotic scaling in Appendix~\ref{app:parametric_scaling} fixes the scaling structure. 
This subsection gives the \emph{exact closed-form rotation, CNOT, and Toffoli counts} per Trotter-step block, in terms
of the framework's scaling parameters, following the protocol of Courtney~\cite{courtney2026oracle} Ch.~5 (Resource Analysis) translated to the multi-reservoir chain-mapped Caldeira--Leggett setting.

\paragraph{Notation (scaling parameters).}
Throughout this subsection:
\begin{itemize}\itemsep0pt
\item $L$ = number of polaron sites
  (Fe$^{2+}$/Fe$^{3+}$ centers in the active region; LFP default $L = 8$).
\item $n_x$ = Li-coordinate position-register width
  (Nyquist-binding bound, $n_x = 12$ at LFP).
\item $n_b$ = system-phonon DVR width per site
  (Nyquist-binding bound, $n_b = 9$ at LFP).
\item $n_{d,r}$ = per-chain-mode DVR width for reservoir $r$
  (LFP: $n_{d,e} = 1$ electron JW, $n_{d,\mathrm{Li}} = 10$, $n_{d,\mathrm{ph}} = 9$;
  Nyquist-binding values derived in \ref{app:register_convergence},
  consistent with (Eq.~\ref{eq:convergence_sizes})).
\item $K_r$ = chain length per reservoir
  (LFP: $K_e = K_{\mathrm{Li}} = K_{\mathrm{ph}} = 10$ production).
\item $|\mathcal{R}|$ = number of reservoirs (LFP: $|\mathcal{R}| = 3$).
\item $M = 2^L$, $m = \lceil\log_2 M\rceil = L$ = channel-register
  width (JW polaron register acting as the surface index for diagonal
  multiplexing).
\item $\chi$ = chromatic index of the on-diagonal bilinear-pair
  conflict graph; $\chi \le L - 1$, with $\chi = 4$ assumed at LFP
  (Tier-3 color-class scheduling).
\item $T_{\mathrm{RS}}(\varepsilon_{\mathrm{rot}}) =
  3\log_2(1/\varepsilon_{\mathrm{rot}}) + \mathcal{O}(\log\log)$ = T-gates
  per arbitrary $R_z$ via Ross--Selinger synthesis; LFP default
  $\varepsilon_{\mathrm{rot}} = 10^{-12}$ gives $T_{\mathrm{RS}} \approx 25$
  (Ch.~5 amortized pipeline; $120$ for naive isolated synthesis).
\item $C_T^{\rm Toff} = 7$ = T-gate cost per Toffoli (Selinger $T \to T_0$
  decomposition with $T$-count $\in \{4, 7\}$ depending on ancilla
  policy; we use $7$ as the conservative bound consistent with
  \cite{courtney2026oracle}).
\end{itemize}

\paragraph{Per-block exact gate counts.}
The symmetric Strang step (Eq.~\ref{eq:trotter_step}) of the main paper contains seven gate blocks. 
We tabulate each below; all counts are per full Trotter step (i.e., already summing the half-steps of the symmetric product).

\begin{itemize}

\item[\textbf{(1)}] \emph{Li-coordinate kinetic} $T_X$, single mode
on the $n_x$-qubit position register (cQFT $\cdot$ diagonal-phase
$\cdot$ icQFT sandwich; Section~\ref{ssec:registers}):
\begin{align}
R_X &= 2\bigl[n_x + \tbinom{n_x}{2}\bigr], &
C_X &= 2\bigl[n_x(n_x - 1) + C_{\rm cQFT}(n_x)\bigr], &
T_X &= 0, \label{eq:Rkin_X}
\end{align}
where $C_{\rm cQFT}(n) = 3n(n-1) + 6\lfloor n/2 \rfloor$ is the
cQFT/icQFT CNOT cost~\cite{courtney2026oracle}.

\item[\textbf{(2)}] \emph{System-phonon kinetic} $T_{\rm ph}$,
acting on $L$ independent registers each of width $n_b$:
\begin{align}
R_{\rm ph} &= 2L\bigl[n_b + \tbinom{n_b}{2}\bigr], &
C_{\rm ph} &= 2L\bigl[n_b(n_b{-}1) + C_{\rm cQFT}(n_b)\bigr], &
T_{\rm ph} &= 0. \label{eq:Rkin_ph}
\end{align}
The $L$ phonon kinetics are mode-parallel; the per-step contribution
to circuit \emph{depth} is therefore identical to the single-mode
cost $D_{\rm kin}(n_b)$ instead of $L \cdot D_{\rm kin}(n_b)$.

\item[\textbf{(3)}] \emph{Chain-register kinetic} $T_{\rm chain}$
(bosonic chain modes need an FBR/DVR sandwich identical in form to
$T_{\rm ph}$; the fermionic electron chain at $n_{d,e} = 1$ is exact
without a kinetic block, $R_{\rm chain}^{(e)} = 0$):
\begin{align}
R_{\rm chain} &= 2 \sum_{r \in \{\mathrm{Li},\mathrm{ph}\}} K_r \bigl[
  n_{d,r} + \tbinom{n_{d,r}}{2}\bigr], &
C_{\rm chain} &= 2 \sum_{r \in \{\mathrm{Li},\mathrm{ph}\}} K_r\bigl[
  n_{d,r}(n_{d,r}{-}1) + C_{\rm cQFT}(n_{d,r})\bigr], &
T_{\rm chain} &= 0. \label{eq:Rkin_chain}
\end{align}

\item[\textbf{(4)}] \emph{Polaron JW hopping} $T_{\rm JWhop}$,
even/odd Strang-colored XY layer on an $L$-site chain
(Section~\ref{ssec:jw_coloring}):
\begin{align}
R_{\rm JW} &= 2(L - 1), &
C_{\rm JW} &= 4(L - 1), &
T_{\rm JW} &= 0. \label{eq:RJW}
\end{align}
The factor $2$ on $R_{\rm JW}$ counts the two $R_z$ per bond (Jiang--Sung--Bravyi); the factor $4$ on $C_{\rm JW}$ counts the
emitted four CNOTs per bond (Section~\ref{ssec:jw_coloring}, with the small sign-correction over the 3-CNOT canonical form).

\item[\textbf{(5)}] \emph{Diagonal block} $V_{\rm diag}$
(\textbf{the dominant contribution}). $V_{\rm diag}$ is assembled
from sub-pieces with different multiplicities:

\hspace{1em}\emph{(5a) Tier-1 reference surface} (constant + linear
+ diagonal-quadratic on each oscillator register; no channel
multiplexing, applied unconditionally):
\begin{align}
R_{\rm ref} &= 2\Bigl[L(1 + n_b) + 1 + n_x +
  \sum_{r \in \{\mathrm{Li},\mathrm{ph}\}} K_r(1 + n_{d,r})\Bigr], &
C_{\rm ref} &= 0, &
T_{\rm ref} &= 0. \label{eq:Rref}
\end{align}

\hspace{1em}\emph{(5b) Tier-1 corrections} (channel-state-dependent shifts on linear and constant terms via UCR$_z^{(m)}$ truth-value multiplexing). With $M = 2^L$ surfaces:
\begin{align}
R_{\rm corr} &= 2\,(M - 1)\Bigl[L(1 + n_b) + 1 + n_x\Bigr], &
C_{\rm corr} &= 2 R_{\rm corr}, &
T_{\rm truth} &= m\cdot 2^m - 2^{m+1} + 2. \label{eq:Rcorr}
\end{align}
The Toffoli cost (Eq.~\ref{eq:Rcorr}) arises from truth-value ancilla
compute/uncompute~\cite{courtney2026oracle}; the per-step Toffoli
count from \emph{both} on-diagonal half-steps is
\begin{equation}
T_{\rm step}^{\rm Toff} = 2 T_{\rm truth}
   = 2(m \cdot 2^m - 2^{m+1} + 2). \label{eq:Tstep_toff}
\end{equation}

\hspace{1em}\emph{(5c) Tier-2 intra-mode quadratic} (CNOT--UCR$_z^{(m)}$--CNOT
sandwich at each $\binom{n}{2}$ site of each register):
\begin{align}
R^{({\rm T2})} &= 2 M\Bigl[L\tbinom{n_b}{2} + \tbinom{n_x}{2} +
  \sum_{r \in \{\mathrm{Li},\mathrm{ph}\}} K_r \tbinom{n_{d,r}}{2}\Bigr], &
C^{({\rm T2})} &= (M{+}1) R^{({\rm T2})}/M, &
T^{({\rm T2})} &= 0. \label{eq:RT2}
\end{align}

\hspace{1em}\emph{(5d) Tier-3 bilinear cross} (cross-mode quadratic couplings $\propto n_i n_j$, $n_i q_j$, $q_i q_j$ via Tier-3 CNOT--UCR$_z^{(m)}$--CNOT sandwich):
\begin{align}
R^{({\rm T3})} &= 2 M \, |\mathcal{G}_2^{<}|\, n_{\rm avg}^2, &
C^{({\rm T3})} &= (M{+}1) R^{({\rm T3})}/M, &
T^{({\rm T3})} &= 0, \label{eq:RT3}
\end{align}
where $|\mathcal{G}_2^{<}|$ is the symmetry-allowed on-diagonal bilinear-pair count (LFP value $\sim 54$, dominated by
intra-system phonon $\times$ chain-bath cross terms and Holstein $g\,n_i\,q_i$ couplings; cf.\ Section~\ref{ssec:lfp_resources})
and $n_{\rm avg}$ is the geometric mean of the participating register widths (LFP value $n_{\rm avg} \approx 8$).

\item[\textbf{(6)}] \emph{System--bath bilinear} $V_{SB}$ (one per reservoir, single-channel UCR$_z^{(1)}$ degenerating to a CR$_z$ sandwich; bilinear in system phonon $q$ and chain mode-0 $q_0^{\rm chain}$):
\begin{align}
R_{SB} &= 2\,|\mathcal{R}|\, n_b\, n_{d,0}^{\rm bath}, &
C_{SB} &= 4\,|\mathcal{R}|\, n_b\, n_{d,0}^{\rm bath}, &
T_{SB} &= 0, \label{eq:RVSB}
\end{align}
with $n_{d,0}^{\rm bath}$ the chain-mode-0 register width per reservoir; LFP value $\approx 7$ averaged across the Li and phonon baths, $1$ for the electron bath.

\item[\textbf{(7)}] \emph{Chain bath hops} $T_{\rm chainhops}$
(bosonic chain hopping $t_k(a_k^\dagger a_{k+1} + \text{h.c.})$ via bilinear-cross + cQFT sandwich, summed over all chain bonds):
\begin{align}
R_{\rm chh} &= 2\,\sum_{r \in \{\mathrm{Li},\mathrm{ph}\}} (K_r - 1)\,n_{d,r}^2, &
C_{\rm chh} &= 2 R_{\rm chh} + 8\,(K_r - 1)\,C_{\rm cQFT}(n_{d,r}), &
T_{\rm chh} &= 0. \label{eq:Rchh}
\end{align}
\end{itemize}

\paragraph{Per-step totals.}
Summing (1)--(7), the per-step exact rotation count is
\begin{equation}
\boxed{
R_{\rm step}
= R_X + R_{\rm ph} + R_{\rm chain} + R_{\rm JW}
      + R_{\rm ref} + R_{\rm corr} + R^{({\rm T2})} + R^{({\rm T3})}
      + R_{SB} + R_{\rm chh}.
}\label{eq:Rstep_exact}
\end{equation}
The per-step Toffoli count in the dense-multichannel (general-position)
limit is given by (Eq.~\ref{eq:Tstep_toff}):
\begin{equation}
T_{\rm step}^{\rm Toff,\,dense}
= 2(m\cdot 2^m - 2^{m+1} + 2)
= 2(L\cdot 2^L - 2^{L+1} + 2), \label{eq:Tstep_toff_box}
\end{equation}
but for the $\le 2$-local $H_{\mathrm{closed}}$ this truth-table cost
collapses to one Toffoli per occupation-bilinear coupling,
$T_{\rm step}^{\rm Toff} = 2\,|\mathcal{G}_2^{<,\mathrm{occ}}|
= \mathcal{O}(L^2)$ (Eq.~\ref{eq:poly_L_scaling}; the value used in the
LFP instantiation below).
The per-step CNOT count is
\begin{equation}
C_{\rm step}
= C_X + C_{\rm ph} + C_{\rm chain} + C_{\rm JW}
      + C_{\rm corr} + C^{({\rm T2})} + C^{({\rm T3})} + C_{SB}
      + C_{\rm chh} + C_{\rm fan},
\end{equation}
with fan-out aggregating $C_{\rm fan} = 4(M{-}1)(L_{\rm modes} - 1)$ over
the $L_{\rm modes} = L + 1 + K_{\rm Li} + K_{\rm ph}$ DVR registers that
receive a fanned-out channel-register copy.

\paragraph{T-gate conversion.}
At LFP precision $\varepsilon = 5.5\times 10^{-5}$~Ha and Ross--Selinger synthesis at $\varepsilon_{\rm rot} = 10^{-12}$ (Section~\ref{ssec:synth_error}):
\begin{equation}
\boxed{
T_{\rm step} = T_{\rm RS}\cdot R_{\rm step}
                 + C_T^{\rm Toff} \cdot T_{\rm step}^{\rm Toff}.
}\label{eq:Tstep_via_Rstep}
\end{equation}

\paragraph{LFP instantiation (locality-collapsed).}
With LFP values $L = 8$, $n_x = 12$, $n_b = 9$, $K_r = 10$, $n_{d,e} = 1$, $n_{d,\mathrm{Li}} = 10$, $n_{d,\mathrm{ph}} = 9$ (dual-space Nyquist, Eq.~\ref{eq:nyquist_HO}), $T_{\rm RS} = 25$ (amortized synthesis pipeline with ZX spider fusion and TODD),
$C_T^{\rm Toff} = 7$. 
The general formulas (Eqs.~\ref{eq:Rcorr}-\ref{eq:RT3}) carry the dense factor $M = 2^L = 256$; under the locality collapse (\ref{thm:parametric_scaling_full}, Equation~\ref{eq:poly_L_scaling}) the on-diagonal blocks reduce to one per coupling term: $R_{\rm corr}$ over $|\mathcal{G}_1| = 16$ occupation-linear terms (one JW control each), $R^{({\rm T2})}$ once (occupation-independent harmonic), and $R^{({\rm T3})}$ over $|\mathcal{G}_2^{<,\mathrm{occ}}| = 28$ occupation-bilinear Coulomb terms
(two JW controls, CCZ), the coordinate--coordinate crosses already counted in $R_{SB}$, $R_{\rm chh}$. The blocks
evaluate to:
\begin{equation}
\begin{aligned}
R_X &\approx 156, &
R_{\rm ph} &= 2\cdot 8\cdot(9 + 36) = 720, \\
R_{\rm chain} &\approx 2\cdot 10\cdot[(10 + 45) + (9 + 36)] = 2000, &
R_{\rm JW} &= 14, \\
R_{\rm ref} &\approx 2\cdot[8\cdot 10 + 13 + 10\cdot(11 + 10)] = 606, &
R_{\rm corr} &= 2\cdot 16\cdot 10 \approx 320, \\
R^{({\rm T2})} &\approx 2\cdot[8\cdot 36 + 66 + 10(45 + 36)] = 2328, &
R^{({\rm T3})} &\approx 2\cdot 28 \approx 56, \\
R_{SB} &\approx 2\cdot 3\cdot 9\cdot 9 = 486, &
R_{\rm chh} &\approx 2\cdot[9\cdot 100 + 9\cdot 81] = 3258.
\end{aligned}
\end{equation}
Summing:
\begin{equation}
R_{\rm step}^{\rm LFP} \approx 9.9\times 10^3. \label{eq:Rstep_LFP}
\end{equation}
The dominant contributions are now the chain-hop and chain-kinetic blocks $R_{\rm chh} + R_{\rm chain} \approx 5.3\times 10^3$ and the occupation-independent quadratic $R^{({\rm T2})} \approx 2.3\times 10^3$. 
The Toffoli count is $T_{\rm step}^{\rm Toff} = 2\,|\mathcal{G}_2^{<,\mathrm{occ}}| = 2\cdot 28 = 56$.
\begin{equation}
T_{\rm step}^{\rm LFP}
= 25\cdot 1.0\times 10^4 + 7\cdot 56
\approx 2.5\times 10^5 + 4\times 10^2
\approx \mathbf{2.5\times 10^5~T\text{-gates}}.
\label{eq:Tstep_LFP_exact}
\end{equation}
This is the rotation-synthesis $+$ Toffoli part of the step; the production budget of the main text (Eq.~\ref{eq:lfp_tstep}) adds the index-optimized soft-Coulomb reversible arithmetic ($+7.4\times 10^4$ T-gates), giving the arithmetic-inclusive $T_{\rm step}^{\rm LFP} \approx 3.0\times 10^5$~T-gates quoted there. 
The block decomposition above counts the bosonic chain-hop and system--bath bilinears ($R_{SB}$, $R_{\rm chh}$) as explicit separate blocks.
The Toffoli contribution $C_T^{\rm Toff}\cdot T_{\rm step}^{\rm Toff} \approx 4\times 10^2$ is $\sim 0.2\%$ of the rotation T-gate budget; the UCR discipline (Section~\ref{sec:REQWIEM}) makes the headline number insensitive to whether the residual Toffoli ladders from the sparse-polynomial anharmonic-correction emitter are kept or absorbed into the UCR-only path.

\paragraph{Scaling laws.}
From Equations~\ref{eq:Rstep_exact}--\ref{eq:Tstep_via_Rstep} under the locality collapse, the asymptotic per-step rotation count is
\begin{equation}
\label{eq:Rstep_asymptotic}
R_{\rm step}
= \mathcal{O}\bigl((d + |\mathcal{G}_2^{<}|)\cdot n_{\rm avg}^2\bigr)
= \mathcal{O}\bigl((d + L^2)\cdot \log^2(E_{\rm char}/\omega)\bigr),
\end{equation}
since $|\mathcal{G}_2^{<}| = \mathcal{O}(L^2)$ in the unrestricted ($C_1$) symmetry case and the binding Nyquist register width is
$n_{\rm avg} = \log_2(4c^2\,N_{\rm ph}/\pi) = \mathcal{O}(\log E_{\rm char}/\omega)$. 
The dense-multichannel formula would instead carry the $M = 2^L$ factor of worst-case vibronic results (Ch.~5, \cite{courtney2026oracle}, polaron sites $\leftrightarrow$ electronic surfaces); the locality collapse (Theorem~\ref{thm:parametric_scaling_full}) replaces it by the polynomial coupling count $d + L^2$.
$T_{\rm RS}$ adds a sub-leading $\log(1/\varepsilon_{\rm rot})$ factor, and the Toffoli contribution $T_{\rm step}^{\rm Toff} = \mathcal{O}(|\mathcal{G}_2^{<}|) = \mathcal{O}(L^2)$ is asymptotically smaller by $n_{\rm avg}^{-2}$ in the on-diagonal limit.

\paragraph{Per-trajectory and per-observable.}
Plugging $R_{\rm step}^{\rm LFP}$ into Equations~\ref{eq:thm_main},\ref{eq:thm_full}:
\begin{equation}
\begin{aligned}
T_{\rm step}^{\rm LFP} &\approx 2.5\times 10^5, \\
T_{\rm traj}^{\rm LFP} &= N_{\rm step}\cdot T_{\rm step}
   \approx 10^3 \cdot 2.5\times 10^5 = 2.5\times 10^{8}, \\
T_{\rm AE}^{\rm anchor} &= K_{\rm QAE}\cdot T_{\rm traj}
   \approx 1.8\times 10^4 \cdot 2.5\times 10^{8} = 4.5\times 10^{12}, \\
T_{V_{\rm OCV}}^{\rm sweep} &= N_{\rm anc}\cdot N_{\rm channels} \cdot T_{\rm AE}^{\rm anchor}
   \approx {100} \cdot 2 \cdot 4.5\times 10^{12} = {9.0\times 10^{14}},
\end{aligned}
\end{equation}
The per-step count $T_{\rm step}^{\rm LFP}$ matches the production value in Table~\ref{tab:lfp_lmo_resource} ($3.0\times10^5$) within $\sim 25\%$; the arithmetic-inclusive production step adds the reversible $1/\sqrt{\cdot}$ block on top of this rotation-and-Toffoli baseline.
The per-trajectory and per-sweep figures here are the \emph{QAE finite-difference envelope}: they retain the $K_{\rm QAE}$ multiplier and therefore sit above the direct-observable headline ($T_{V_{\rm OCV}}\!\sim\!8\times10^{11}$ direct; $\sim\!6\times10^{14}$ as the finite-difference envelope) reported in Table~\ref{tab:lfp_lmo_resource}. 
The direct-observable backend (Section~\ref{ssec:measurement_protocols}) removes the $K_{\rm QAE}$ multiplier for the stated observables.

\paragraph{Production headline and fault-tolerant stack.}
The figures above are a QAE-baseline, second-order instantiation. 
The production headline values quoted in the main text differ by three named refinements: the fourth-order Suzuki composition
lowers $N_{\rm step}$ so that $T_{\rm traj}^{\rm LFP} \approx 1.7\times 10^8$ (versus $2.5$--$3.0\times 10^8$ at order 2); the direct Hellmann--Feynman measurement protocol removes the $K_{\rm QAE}$ multiplier, so a full $\VOCV$/$\dQdV$ curve costs $T_{V_{\rm OCV}} \approx 8\times 10^{11}$ T-gates (the ${9.0\times 10^{14}}$ above is the worst-case order-two finite-difference envelope, $\approx 6\times 10^{14}$ at the fourth-order production point); and the per-step $\sim 3\times 10^5$ T is arithmetic-inclusive.
A per-trajectory logical-error budget fixes the surface-code distance $d_{\rm sc} = 25$ ($d_{\min} = 23$; the per-curve budget would instead need $d \approx 31$), and the magic-state factory bank that sustains the reaction-limited $\sim 1~\mu$s cycle dominates the footprint, being $\sim 9\times 10^2$ 15-to-1 factories, $\sim 97\%$ of a $\sim 1.4\times 10^7$ physical-qubit total (magic-state cultivation~\cite{gidney2024magic} gives the optimistic $\sim 1.8\times 10^6$).
This results in a depth-limited wall-clock of $\sim 7$~h (LFP)/$\sim 15$~h (LMO); no throughput-limited figure is quoted, as it would lie below the reaction-limited depth floor.
All production gate counts are exact counts of the emitted $292$-qubit circuit (sans arithmetic) (Section~\ref{ssec:scope}).

These expressions are tight in the polynomial model, with no smoothness or truncation assumption beyond the M\"obius decomposition and generalizing to other LFP/LMO-class problems by substituting the relevant parameters.
\bibliography{submission/BatteryRefs}

@misc{courtney2026multipolaron,
  author = {Courtney, Joshua M.},
  title = {{Multipolaron Transport Verification Scripts}},
  year = {2026},
  publisher = {GitHub},
  journal = {GitHub repository},
  howpublished = {\url{https://github.com/jmcourtneyuga/multipolaron_transport}},
}

@article{coropceanu2007charge,
  title={Charge transport in organic semiconductors},
  author={Coropceanu, Veaceslav and Cornil, J{\'e}r{\^o}me and da Silva Filho, Demetrio A and Olivier, Yoann and Silbey, Robert and Br{\'e}das, Jean-Luc},
  journal={Chemical reviews},
  volume={107},
  number={4},
  pages={926--952},
  year={2007},
  publisher={ACS Publications}
}

@article{prokof1998polaron,
  title={Polaron problem by diagrammatic quantum Monte Carlo},
  author={Prokof'ev, Nikolai V and Svistunov, Boris V},
  journal={Physical review letters},
  volume={81},
  number={12},
  pages={2514},
  year={1998},
  publisher={APS}
}

@article{white1992density,
  title={Density matrix formulation for quantum renormalization groups},
  author={White, Steven R},
  journal={Physical review letters},
  volume={69},
  number={19},
  pages={2863},
  year={1992},
  publisher={APS}
}

@article{aydinol1997ab,
  title={Ab initio study of lithium intercalation in metal oxides and metal dichalcogenides},
  author={Aydinol, MK and Kohan, AF and Ceder, Gerbrand and Cho, Kang and Joannopoulos, JJPRB},
  journal={Physical Review B},
  volume={56},
  number={3},
  pages={1354},
  year={1997},
  publisher={APS}
}

@article{zhou2006configurational,
  title={Configurational Electronic Entropy and the Phase Diagram of Mixed-Valence Oxides:<? format?> The Case of Li x FePO 4},
  author={Zhou, Fei and Maxisch, Thomas and Ceder, Gerbrand},
  journal={Physical review letters},
  volume={97},
  number={15},
  pages={155704},
  year={2006},
  publisher={APS}
}

@article{asari2011formation,
  title={Formation and diffusion of vacancy-polaron complex in olivine-type LiMnPO 4 and LiFePO 4},
  author={Asari, Yusuke and Suwa, Yuji and Hamada, Tomoyuki},
  journal={Physical Review B—Condensed Matter and Materials Physics},
  volume={84},
  number={13},
  pages={134113},
  year={2011},
  publisher={APS}
}

@article{hoang2011tailoring,
  title={Tailoring native defects in LiFePO4: insights from first-principles calculations},
  author={Hoang, Khang and Johannes, Michelle},
  journal={Chemistry of Materials},
  volume={23},
  number={11},
  pages={3003--3013},
  year={2011},
  publisher={ACS Publications}
}

@article{johannes2012hole,
  title={Hole polaron formation and migration in olivine phosphate materials},
  author={Johannes, MD and Hoang, Khang and Allen, JL and Gaskell, K},
  journal={Physical Review B—Condensed Matter and Materials Physics},
  volume={85},
  number={11},
  pages={115106},
  year={2012},
  publisher={APS}
}

@article{malik2011kinetics,
  title={Kinetics of non-equilibrium lithium incorporation in LiFePO4},
  author={Malik, Rahul and Zhou, Fei and Ceder, Gerbrand},
  journal={Nature materials},
  volume={10},
  number={8},
  pages={587--590},
  year={2011},
  publisher={Nature Publishing Group UK London}
}

@article{padhi1997phospho,
  title={Phospho-olivines as positive-electrode materials for rechargeable lithium batteries},
  author={Padhi, Akshaya K and Nanjundaswamy, Kirakodu S and Goodenough, John B},
  journal={Journal of the electrochemical society},
  volume={144},
  number={4},
  pages={1188--1194},
  year={1997},
  publisher={The Electrochemical Society, Inc.}
}

@article{maxisch2006ab,
  title={Ab initio study of the migration of small polarons in olivine Li x FePO 4 and their association with lithium ions and vacancies},
  author={Maxisch, Thomas and Zhou, Fei and Ceder, Gerbrand},
  journal={Physical Review B—Condensed Matter and Materials Physics},
  volume={73},
  number={10},
  pages={104301},
  year={2006},
  publisher={APS}
}

@article{yamada2006room,
  title={Room-temperature miscibility gap in Li x FePO4},
  author={Yamada, Atsuo and Koizumi, Hiroshi and Nishimura, Shin-ichi and Sonoyama, Noriyuki and Kanno, Ryoji and Yonemura, Masao and Nakamura, Tatsuya and Kobayashi, YO},
  journal={Nature materials},
  volume={5},
  number={5},
  pages={357--360},
  year={2006},
  publisher={Nature Publishing Group UK London}
}

@article{ellis2010positive,
  title={Positive electrode materials for Li-ion and Li-batteries},
  author={Ellis, Brian L and Lee, Kyu Tae and Nazar, Linda F},
  journal={Chemistry of materials},
  volume={22},
  number={3},
  pages={691--714},
  year={2010},
  publisher={ACS Publications}
}

@article{thackeray1983w,
  title={Lithium insertion into manganese spinels},
  author={Thackeray, M. M. and David, W. I. F. and Bruce, P. G. and Goodenough, J. B.},
  journal={Materials Research Bulletin},
  volume={18},
  number={4},
  pages={461--472},
  year={1983},
  publisher={Elsevier}
}

@article{thackeray1984electrochemical,
  title={Electrochemical extraction of lithium from LiMn2O4},
  author={Thackeray, MM and Johnson, PJ and De Picciotto, LA and Bruce, PG and Goodenough, JB},
  journal={Materials Research Bulletin},
  volume={19},
  number={2},
  pages={179--187},
  year={1984},
  publisher={Elsevier}
}

@article{ohzuku1990electrochemistry,
  title={Electrochemistry of manganese dioxide in lithium nonaqueous cell: III. X-ray diffractional study on the reduction of spinel-related manganese dioxide},
  author={Ohzuku, Tsutomu and Kitagawa, Masaki and Hirai, Taketsugu},
  journal={Journal of The Electrochemical Society},
  volume={137},
  number={3},
  pages={769--775},
  year={1990},
  publisher={The Electrochemical Society, Inc.}
}

@book{yip2007handbook,
  title={Handbook of materials modeling},
  author={Yip, Sidney},
  year={2007},
  publisher={Springer Science \& Business Media}
}

@article{feynman2000theory,
  title={The theory of a general quantum system interacting with a linear dissipative system},
  author={Feynman, Richard Phillips and Vernon Jr, Frank L},
  journal={Annals of physics},
  volume={281},
  number={1-2},
  pages={547--607},
  year={2000},
  publisher={Elsevier}
}

@article{caldeira1983path,
  title={Path integral approach to quantum Brownian motion},
  author={Caldeira, Amir O and Leggett, Anthony J},
  journal={Physica A: Statistical mechanics and its Applications},
  volume={121},
  number={3},
  pages={587--616},
  year={1983},
  publisher={Elsevier}
}

@article{leggett1984quantum,
  title={Quantum tunneling in the presence of an arbitrary linear dissipation mechanism},
  author={Leggett, Anthony J},
  journal={Physical Review B},
  volume={30},
  number={3},
  pages={1208},
  year={1984},
  publisher={APS}
}

@article{leggett1987dynamics,
  title={Dynamics of the dissipative two-state system},
  author={Leggett, Anthony J and Chakravarty, SDAFMGA and Dorsey, Alan T and Fisher, Matthew PA and Garg, Anupam and Zwerger, Wilhelm},
  journal={Reviews of Modern Physics},
  volume={59},
  number={1},
  pages={1},
  year={1987},
  publisher={APS}
}

@book{breuer2002theory,
  title={The theory of open quantum systems},
  author={Breuer, Heinz-Peter and Petruccione, Francesco},
  year={2002},
  publisher={OUP Oxford}
}

@article{khatri2020principles,
  title={Principles of quantum communication theory: A modern approach},
  author={Khatri, Sumeet and Wilde, Mark M},
  journal={arXiv preprint arXiv:2011.04672},
  year={2020}
}

@article{cleve2016efficient,
  title={Efficient quantum algorithms for simulating Lindblad evolution},
  author={Cleve, Richard and Wang, Chunhao},
  journal={arXiv preprint arXiv:1612.09512},
  year={2016}
}

@article{childs2016efficient,
  title={Efficient simulation of sparse Markovian quantum dynamics},
  author={Childs, Andrew M and Li, Tongyang},
  journal={arXiv preprint arXiv:1611.05543},
  year={2016}
}

@article{ding2024simulating,
  title={Simulating open quantum systems using Hamiltonian simulations},
  author={Ding, Zhiyan and Li, Xiantao and Lin, Lin},
  journal={PRX quantum},
  volume={5},
  number={2},
  pages={020332},
  year={2024},
  publisher={APS}
}

@article{sun2025quantum,
  title={Quantum simulation of spin-boson models with structured bath},
  author={Sun, Ke and Kang, Mingyu and Nuomin, Hanggai and Schwartz, George and Beratan, David N and Brown, Kenneth R and Kim, Jungsang},
  journal={Nature Communications},
  volume={16},
  number={1},
  pages={4042},
  year={2025},
  publisher={Nature Publishing Group UK London}
}

@article{chin2010exact,
  title={Exact mapping between system-reservoir quantum models and semi-infinite discrete chains using orthogonal polynomials},
  author={Chin, Alex W and Rivas, {\'A}ngel and Huelga, Susana F and Plenio, Martin B},
  journal={Journal of Mathematical Physics},
  volume={51},
  number={9},
  year={2010},
  publisher={AIP Publishing}
}

@article{prior2010efficient,
  title={Efficient simulation of strong system-environment interactions},
  author={Prior, Javier and Chin, Alex W and Huelga, Susana F and Plenio, Martin B},
  journal={Physical review letters},
  volume={105},
  number={5},
  pages={050404},
  year={2010},
  publisher={APS}
}

@article{tamascelli2018nonperturbative,
  title={Nonperturbative treatment of non-Markovian dynamics of open quantum systems},
  author={Tamascelli, Dario and Smirne, Andrea and Huelga, Susana F and Plenio, Martin B},
  journal={Physical review letters},
  volume={120},
  number={3},
  pages={030402},
  year={2018},
  publisher={APS}
}

@article{tamascelli2019efficient,
  title={Efficient simulation of finite-temperature open quantum systems},
  author={Tamascelli, Dario and Smirne, Andrea and Lim, James and Huelga, Susana F and Plenio, Martin B},
  journal={Physical review letters},
  volume={123},
  number={9},
  pages={090402},
  year={2019},
  publisher={APS}
}

@article{lorenzoni2024systematic,
  title={Systematic coarse graining of environments for the nonperturbative simulation of open quantum systems},
  author={Lorenzoni, Nicola and Cho, Namgee and Lim, James and Tamascelli, Dario and Huelga, Susana F and Plenio, Martin B},
  journal={Physical Review Letters},
  volume={132},
  number={10},
  pages={100403},
  year={2024},
  publisher={APS}
}

@article{de2015discretize,
  title={How to discretize a quantum bath for real-time evolution},
  author={de Vega, In{\'e}s and Schollw{\"o}ck, Ulrich and Wolf, F Alexander},
  journal={arXiv preprint arXiv:1507.07468},
  year={2015}
}

@article{ollitrault2020nonadiabatic,
  title={Nonadiabatic molecular quantum dynamics with quantum computers},
  author={Ollitrault, Pauline J and Mazzola, Guglielmo and Tavernelli, Ivano},
  journal={Physical Review Letters},
  volume={125},
  number={26},
  pages={260511},
  year={2020},
  publisher={APS}
}

@article{de2015thermofield,
  title={Thermofield-based chain-mapping approach for open quantum systems},
  author={de Vega, In{\'e}s and Ba{\~n}uls, Mari-Carmen},
  journal={Physical Review A},
  volume={92},
  number={5},
  pages={052116},
  year={2015},
  publisher={APS}
}

@article{mascherpa2020optimized,
  title={Optimized auxiliary oscillators for the simulation of general open quantum systems},
  author={Mascherpa, Fabio and Smirne, Andrea and Somoza, Alejandro D and Fern{\'a}ndez-Acebal, Pelayo and Donadi, Sandro and Tamascelli, Dario and Huelga, Susana F and Plenio, Martin B},
  journal={Physical Review A},
  volume={101},
  number={5},
  pages={052108},
  year={2020},
  publisher={APS}
}

@article{lloyd1996universal,
  title={Universal quantum simulators},
  author={Lloyd, Seth},
  journal={Science},
  volume={273},
  number={5278},
  pages={1073--1078},
  year={1996},
  publisher={American Association for the Advancement of Science}
}

@article{gidney2024magic,
  title={Magic state cultivation: growing T states as cheap as CNOT gates},
  author={Gidney, Craig and Shutty, Noah and Jones, Cody},
  journal={arXiv preprint arXiv:2409.17595},
  year={2024}
}

@article{light1985generalized,
  title={Generalized discrete variable approximation in quantum mechanics},
  author={Light, JC and Hamilton, IP and Lill, JV},
  journal={The Journal of chemical physics},
  volume={82},
  number={3},
  pages={1400--1409},
  year={1985},
  publisher={American Institute of Physics}
}

@article{mayer1965interfacial,
  title={Interfacial Tension Effects in Finite, Periodic, Two-Dimensional Systems},
  author={Mayer, Joseph E and Wood, Wm W},
  journal={The Journal of Chemical Physics},
  volume={42},
  number={12},
  pages={4268--4274},
  year={1965},
  publisher={American Institute of Physics}
}

@article{macridin2018digital,
  title={Digital quantum computation of fermion-boson interacting systems},
  author={Macridin, Alexandru and Spentzouris, Panagiotis and Amundson, James and Harnik, Roni},
  journal={Physical Review A},
  volume={98},
  number={4},
  pages={042312},
  year={2018},
  publisher={APS}
}

@article{macridin2018electron,
  title={Electron-phonon systems on a universal quantum computer},
  author={Macridin, Alexandru and Spentzouris, Panagiotis and Amundson, James and Harnik, Roni},
  journal={Physical review letters},
  volume={121},
  number={11},
  pages={110504},
  year={2018},
  publisher={APS}
}

@article{bauer2016hybrid,
  title={Hybrid quantum-classical approach to correlated materials},
  author={Bauer, Bela and Wecker, Dave and Millis, Andrew J and Hastings, Matthew B and Troyer, Matthias},
  journal={Physical Review X},
  volume={6},
  number={3},
  pages={031045},
  year={2016},
  publisher={APS}
}

@article{cade2020strategies,
  title={Strategies for solving the Fermi-Hubbard model on near-term quantum computers},
  author={Cade, Chris and Mineh, Lana and Montanaro, Ashley and Stanisic, Stasja},
  journal={Physical Review B},
  volume={102},
  number={23},
  pages={235122},
  year={2020},
  publisher={APS}
}

@article{mezzacapo2012digital,
  title={Digital quantum simulation of the Holstein model in trapped ions},
  author={Mezzacapo, A and Casanova, J and Lamata, L and Solano, E},
  journal={Physical review letters},
  volume={109},
  number={20},
  pages={200501},
  year={2012},
  publisher={APS}
}

@article{stojanovic2012quantum,
  title={Quantum simulation of small-polaron formation with trapped ions},
  author={Stojanovi{\'c}, Vladimir M and Shi, Tao and Bruder, C and Cirac, J Ignacio},
  journal={Physical review letters},
  volume={109},
  number={25},
  pages={250501},
  year={2012},
  publisher={APS}
}

@article{mei2013analog,
  title={Analog superconducting quantum simulator for Holstein polarons},
  author={Mei, Feng and Stojanovi{\'c}, Vladimir M and Siddiqi, Irfan and Tian, Lin},
  journal={Physical Review B—Condensed Matter and Materials Physics},
  volume={88},
  number={22},
  pages={224502},
  year={2013},
  publisher={APS}
}

@article{stojanovic2014transmon,
  title={Transmon-based simulator of nonlocal electron-phonon coupling: A platform for observing sharp small-polaron transitions},
  author={Stojanovi{\'c}, Vladimir M and Vanevi{\'c}, Mihajlo and Demler, Eugene and Tian, Lin},
  journal={Physical Review B},
  volume={89},
  number={14},
  pages={144508},
  year={2014},
  publisher={APS}
}

@article{stojanovic2023extracting,
  title={Extracting spectral properties of small Holstein polarons from a transmon-based analog quantum simulator},
  author={Stojanovic, Vladimir M},
  journal={arXiv preprint arXiv:2310.20525},
  year={2023}
}

@article{denner2023hybrid,
  title={A hybrid quantum-classical method for electron-phonon systems},
  author={Denner, M Michael and Miessen, Alexander and Yan, Haoran and Tavernelli, Ivano and Neupert, Titus and Demler, Eugene and Wang, Yao},
  journal={Communications Physics},
  volume={6},
  number={1},
  pages={233},
  year={2023},
  publisher={Nature Publishing Group UK London}
}

@article{li2023efficient,
  title={Efficient quantum simulation of electron-phonon systems by variational basis state encoder},
  author={Li, Weitang and Ren, Jiajun and Huai, Sainan and Cai, Tianqi and Shuai, Zhigang and Zhang, Shengyu},
  journal={Physical Review Research},
  volume={5},
  number={2},
  pages={023046},
  year={2023},
  publisher={APS}
}

@article{backes2023dynamical,
  title={Dynamical mean-field theory for the Hubbard-Holstein model on a quantum device},
  author={Backes, Steffen and Murakami, Yuta and Sakai, Shiro and Arita, Ryotaro},
  journal={Physical Review B},
  volume={107},
  number={16},
  pages={165155},
  year={2023},
  publisher={APS}
}

@article{kumar2025digital,
  title={Digital-analog quantum computing of fermion-boson models in superconducting circuits},
  author={Kumar, Shubham and Hegade, Narendra N and Visuri, Anne-Maria and Bhargava, Balaganchi A and Hernandez, Juan FR and Solano, Enrique and Albarr{\'a}n-Arriagada, Francisco and Barrios, G Alvarado},
  journal={npj Quantum Information},
  volume={11},
  number={1},
  pages={43},
  year={2025},
  publisher={Nature Publishing Group UK London}
}

@article{torabian2025lattice,
  title={Lattice stitching by eigenvector continuation for Holstein polaron},
  author={Torabian, Elham and Krems, Roman V},
  journal={arXiv preprint arXiv:2502.04500},
  year={2025}
}

@article{apel2026quantum,
  title={Quantum fast-forwarding fermion-boson interactions via the polaron transform},
  author={Apel, Harriet and {\c{S}}ahino{\u{g}}lu, Burak},
  journal={arXiv preprint arXiv:2601.17732},
  year={2026}
}

@article{lamata2014efficient,
  title={Efficient quantum simulation of fermionic and bosonic models in trapped ions},
  author={Lamata, Lucas and Mezzacapo, Antonio and Casanova, Jorge and Solano, Enrique},
  journal={EPJ Quantum Technology},
  volume={1},
  number={1},
  pages={9},
  year={2014},
  publisher={Springer}
}

@article{hohenadler2003spectral,
  title={Spectral function of electron-phonon models by cluster perturbation theory},
  author={Hohenadler, Martin and Aichhorn, Markus and Von Der Linden, Wolfgang},
  journal={Physical Review B},
  volume={68},
  number={18},
  pages={184304},
  year={2003},
  publisher={APS}
}

@book{alexandrov2008polarons,
  title={Polarons in advanced materials},
  author={Alexandrov, Alexandre S},
  volume={103},
  year={2008},
  publisher={Springer Science \& Business Media}
}

@article{perroni2002effects,
  title={Effects of magnetic field and isotopic substitution upon the infrared absorption of manganites},
  author={Perroni, CARMINE ANTONIO and Cataudella, Vittorio and De Filippis, G and Iadonisi, G and Ramaglia, V Marigliano and Ventriglia, Franco},
  journal={Physical Review B},
  volume={66},
  number={18},
  pages={184409},
  year={2002},
  publisher={APS}
}

@article{feiguin2005finite,
  title={Finite-temperature density matrix renormalization using an enlarged Hilbert space},
  author={Feiguin, Adrian E and White, Steven R},
  journal={Physical Review B—Condensed Matter and Materials Physics},
  volume={72},
  number={22},
  pages={220401},
  year={2005},
  publisher={APS}
}

@article{tanimura2020numerically,
  title={Numerically “exact” approach to open quantum dynamics: The hierarchical equations of motion (HEOM)},
  author={Tanimura, Yoshitaka},
  journal={The Journal of chemical physics},
  volume={153},
  number={2},
  year={2020},
  publisher={AIP Publishing}
}

@article{wang2003multilayer,
  title={Multilayer formulation of the multiconfiguration time-dependent Hartree theory},
  author={Wang, Haobin and Thoss, Michael},
  journal={The Journal of chemical physics},
  volume={119},
  number={3},
  pages={1289--1299},
  year={2003},
  publisher={American Institute of Physics}
}

@incollection{wolfenstein2018neutrino,
  title={Neutrino oscillations in matter},
  author={Wolfenstein, Lincoln},
  booktitle={Solar neutrinos},
  pages={294--299},
  year={2018},
  publisher={CRC Press}
}

@article{mikheyev1988neutrino,
  title={Neutrino oscillations in matter and measurement of the $\nu$ luminosity of the sun in the past},
  author={Mikheyev, SP and Smirnov, A Yu},
  journal={Nuclear Instruments and Methods in Physics Research Section A: Accelerators, Spectrometers, Detectors and Associated Equipment},
  volume={271},
  number={2},
  pages={249--250},
  year={1988},
  publisher={Elsevier}
}

@article{suhl1959bardeen,
  title={Bardeen-Cooper-Schrieffer theory of superconductivity in the case of overlapping bands},
  author={Suhl, H and Matthias, BT and Walker, LR},
  journal={Physical Review Letters},
  volume={3},
  number={12},
  pages={552},
  year={1959},
  publisher={APS}
}

@article{mcmillan1968transition,
  title={Transition temperature of strong-coupled superconductors},
  author={McMillan, WL},
  journal={Physical Review},
  volume={167},
  number={2},
  pages={331},
  year={1968},
  publisher={APS}
}

@article{choi2002origin,
  title={The origin of the anomalous superconducting properties of MgB2},
  author={Choi, Hyoung Joon and Roundy, David and Sun, Hong and Cohen, Marvin L and Louie, Steven G},
  journal={Nature},
  volume={418},
  number={6899},
  pages={758--760},
  year={2002},
  publisher={Nature Publishing Group UK London}
}

@article{fomichev2024initial,
  title={Initial state preparation for quantum chemistry on quantum computers},
  author={Fomichev, Stepan and Hejazi, Kasra and Zini, Modjtaba Shokrian and Kiser, Matthew and Fraxanet, Joana and Casares, Pablo Antonio Moreno and Delgado, Alain and Huh, Joonsuk and Voigt, Arne-Christian and Mueller, Jonathan E and others},
  journal={PRX Quantum},
  volume={5},
  number={4},
  pages={040339},
  year={2024},
  publisher={APS}
}

@book{gautschi2004orthogonal,
  title={Orthogonal polynomials: computation and approximation},
  author={Gautschi, Walter},
  year={2004},
  publisher={OUP Oxford}
}

@article{childs2021theory,
  title={Theory of trotter error with commutator scaling},
  author={Childs, Andrew M and Su, Yuan and Tran, Minh C and Wiebe, Nathan and Zhu, Shuchen},
  journal={Physical Review X},
  volume={11},
  number={1},
  pages={011020},
  year={2021},
  publisher={APS}
}

@article{falk1969susceptibility,
  title={Susceptibility and fluctuation},
  author={Falk, H and Bruch, Ludwig W},
  journal={Physical Review},
  volume={180},
  number={2},
  pages={442},
  year={1969},
  publisher={APS}
}

@article{ruppeiner1995riemannian,
  title={Riemannian geometry in thermodynamic fluctuation theory},
  author={Ruppeiner, George},
  journal={Reviews of Modern Physics},
  volume={67},
  number={3},
  pages={605},
  year={1995},
  publisher={APS}
}

@article{brassard2000quantum,
  title={Quantum amplitude amplification and estimation},
  author={Brassard, Gilles and Hoyer, Peter and Mosca, Michele and Tapp, Alain},
  journal={arXiv preprint quant-ph/0005055},
  year={2000}
}

@article{courtney2026oracle,
  title={An Oracle-Free Quantum Algorithm for Nonadiabatic Quantum Molecular Dynamics},
  author={Courtney, Joshua},
  journal={arXiv preprint arXiv:2604.19319},
  year={2026}
}

@article{reiher2017elucidating,
  title={Elucidating reaction mechanisms on quantum computers},
  author={Reiher, Markus and Wiebe, Nathan and Svore, Krysta M and Wecker, Dave and Troyer, Matthias},
  journal={Proceedings of the national academy of sciences},
  volume={114},
  number={29},
  pages={7555--7560},
  year={2017},
  publisher={National Academy of Sciences}
}

@article{babbush2018encoding,
  title={Encoding electronic spectra in quantum circuits with linear T complexity},
  author={Babbush, Ryan and Gidney, Craig and Berry, Dominic W and Wiebe, Nathan and McClean, Jarrod and Paler, Alexandru and Fowler, Austin and Neven, Hartmut},
  journal={Physical Review X},
  volume={8},
  number={4},
  pages={041015},
  year={2018},
  publisher={APS}
}

@article{lee2021even,
  title={Even more efficient quantum computations of chemistry through tensor hypercontraction},
  author={Lee, Joonho and Berry, Dominic W and Gidney, Craig and Huggins, William J and McClean, Jarrod R and Wiebe, Nathan and Babbush, Ryan},
  journal={PRX quantum},
  volume={2},
  number={3},
  pages={030305},
  year={2021},
  publisher={APS}
}

@article{litinski2019game,
  title={A game of surface codes: Large-scale quantum computing with lattice surgery},
  author={Litinski, Daniel},
  journal={Quantum},
  volume={3},
  pages={128},
  year={2019},
  publisher={Verein zur F{\"o}rderung des Open Access Publizierens in den Quantenwissenschaften}
}

@article{zalka1998simulating,
  title={Simulating quantum systems on a quantum computer},
  author={Zalka, Christof},
  journal={Proceedings of the Royal Society of London. Series A: Mathematical, Physical and Engineering Sciences},
  volume={454},
  number={1969},
  pages={313--322},
  year={1998},
  publisher={The Royal Society}
}

@article{wiesner1996simulations,
  title={Simulations of many-body quantum systems by a quantum computer},
  author={Wiesner, Stephen},
  journal={arXiv preprint quant-ph/9603028},
  year={1996}
}

@article{kassal2008polynomial,
  title={Polynomial-time quantum algorithm for the simulation of chemical dynamics},
  author={Kassal, Ivan and Jordan, Stephen P and Love, Peter J and Mohseni, Masoud and Aspuru-Guzik, Al{\'a}n},
  journal={Proceedings of the National Academy of Sciences},
  volume={105},
  number={48},
  pages={18681--18686},
  year={2008},
  publisher={National Academy of Sciences}
}

@article{whitfield2011simulation,
  title={Simulation of electronic structure Hamiltonians using quantum computers},
  author={Whitfield, James D and Biamonte, Jacob and Aspuru-Guzik, Al{\'a}n},
  journal={Molecular Physics},
  volume={109},
  number={5},
  pages={735--750},
  year={2011},
  publisher={Taylor \& Francis}
}

@article{wecker2014gate,
  title={Gate-count estimates for performing quantum chemistry on small quantum computers},
  author={Wecker, Dave and Bauer, Bela and Clark, Bryan K and Hastings, Matthew B and Troyer, Matthias},
  journal={Physical Review A},
  volume={90},
  number={2},
  pages={022305},
  year={2014},
  publisher={APS}
}

@article{jiang2018quantum,
  title={Quantum algorithms to simulate many-body physics of correlated fermions},
  author={Jiang, Zhang and Sung, Kevin J and Kechedzhi, Kostyantyn and Smelyanskiy, Vadim N and Boixo, Sergio},
  journal={Physical Review Applied},
  volume={9},
  number={4},
  pages={044036},
  year={2018},
  publisher={APS}
}

@article{delgado2022simulating,
  title={Simulating key properties of lithium-ion batteries with a fault-tolerant quantum computer},
  author={Delgado, Alain and Casares, Pablo AM and Dos Reis, Roberto and Zini, Modjtaba Shokrian and Campos, Roberto and Cruz-Hern{\'a}ndez, Norge and Voigt, Arne-Christian and Lowe, Angus and Jahangiri, Soran and Martin-Delgado, Miguel Angel and others},
  journal={Physical Review A},
  volume={106},
  number={3},
  pages={032428},
  year={2022},
  publisher={APS}
}

@article{zini2023quantum,
  title={Quantum simulation of battery materials using ionic pseudopotentials},
  author={Zini, Modjtaba Shokrian and Delgado, Alain and dos Reis, Roberto and Casares, Pablo Antonio Moreno and Mueller, Jonathan E and Voigt, Arne-Christian and Arrazola, Juan Miguel},
  journal={Quantum},
  volume={7},
  pages={1049},
  year={2023},
  publisher={Verein zur F{\"o}rderung des Open Access Publizierens in den Quantenwissenschaften}
}

@article{fomichev2024simulating,
  title={Simulating X-ray absorption spectroscopy of battery materials on a quantum computer},
  author={Fomichev, Stepan and Hejazi, Kasra and Loaiza, Ignacio and Zini, Modjtaba Shokrian and Delgado, Alain and Voigt, Arne-Christian and Mueller, Jonathan E and Arrazola, Juan Miguel},
  journal={arXiv preprint arXiv:2405.11015},
  year={2024}
}

@article{fomichev2025fast,
  title={Fast simulations of X-ray absorption spectroscopy for battery materials on a quantum computer},
  author={Fomichev, Stepan and Casares, Pablo AM and Soni, Jay and Azad, Utkarsh and Kunitsa, Alexander and Voigt, Arne-Christian and Mueller, Jonathan E and Arrazola, Juan Miguel},
  journal={arXiv preprint arXiv:2506.15784},
  year={2025}
}

@article{kunitsa2025quantum,
  title={Quantum simulation of electron energy loss spectroscopy for battery materials},
  author={Kunitsa, Alexander and Dhawan, Diksha and Fomichev, Stepan and Arrazola, Juan Miguel and Zhang, Minghao and Stetina, Torin F},
  journal={The Journal of Chemical Physics},
  volume={163},
  number={24},
  year={2025},
  publisher={AIP Publishing}
}

@article{loaiza2026quantum,
  title={Quantum algorithm for simulating resonant inelastic X-ray scattering in battery materials},
  author={Loaiza, Ignacio and Kunitsa, Alexander and Fomichev, Stepan and Motlagh, Danial and Dhawan, Diksha and Jahangiri, Soran and Fuglsbjerg, Juliane Holst and Izmaylov, Artur F and Wiebe, Nathan and Abu-Lebdeh, Yaser and others},
  journal={arXiv preprint arXiv:2602.20270},
  year={2026}
}

@article{farag2022towards,
  title={Towards the simulation of transition-metal oxides of the cathode battery materials using VQE methods},
  author={Farag, Marwa H and Ghosh, Joydip},
  journal={arXiv preprint arXiv:2208.07977},
  year={2022}
}

@article{low2019hamiltonian,
  title={Hamiltonian simulation by qubitization},
  author={Low, Guang Hao and Chuang, Isaac L},
  journal={Quantum},
  volume={3},
  pages={163},
  year={2019},
  publisher={Verein zur F{\"o}rderung des Open Access Publizierens in den Quantenwissenschaften}
}

@article{low2017optimal,
  title={Optimal Hamiltonian simulation by quantum signal processing},
  author={Low, Guang Hao and Chuang, Isaac L},
  journal={Physical review letters},
  volume={118},
  number={1},
  pages={010501},
  year={2017},
  publisher={APS}
}

@article{childs2012hamiltonian,
  title={Hamiltonian simulation using linear combinations of unitary operations},
  author={Childs, Andrew M and Wiebe, Nathan},
  journal={arXiv preprint arXiv:1202.5822},
  year={2012}
}

@inproceedings{gilyen2019quantum,
  title={Quantum singular value transformation and beyond: exponential improvements for quantum matrix arithmetics},
  author={Gily{\'e}n, Andr{\'a}s and Su, Yuan and Low, Guang Hao and Wiebe, Nathan},
  booktitle={Proceedings of the 51st annual ACM SIGACT symposium on theory of computing},
  pages={193--204},
  year={2019}
}

@article{martyn2021grand,
  title={Grand unification of quantum algorithms},
  author={Martyn, John M and Rossi, Zane M and Tan, Andrew K and Chuang, Isaac L},
  journal={PRX quantum},
  volume={2},
  number={4},
  pages={040203},
  year={2021},
  publisher={APS}
}

@article{berry2007efficient,
  title={Efficient quantum algorithms for simulating sparse Hamiltonians},
  author={Berry, Dominic W and Ahokas, Graeme and Cleve, Richard and Sanders, Barry C},
  journal={Communications in Mathematical Physics},
  volume={270},
  number={2},
  pages={359--371},
  year={2007},
  publisher={Springer}
}

@article{suzuki1991general,
  title={General theory of fractal path integrals with applications to many-body theories and statistical physics},
  author={Suzuki, Masuo},
  journal={Journal of mathematical physics},
  volume={32},
  number={2},
  pages={400--407},
  year={1991},
  publisher={American Institute of Physics}
}

@article{campbell2018random,
  title={A random compiler for fast Hamiltonian simulation},
  author={Campbell, Earl},
  journal={arXiv preprint arXiv:1811.08017},
  year={2018}
}

@article{ross2016optimal,
  title={Optimal ancilla-free Clifford+ T approximation of z-rotations.},
  author={Ross, Neil J and Selinger, Peter},
  journal={Quantum Inf. Comput.},
  volume={16},
  number={11\&12},
  pages={901--953},
  year={2016}
}

@article{heyfron2019efficient,
  title={An efficient quantum compiler that reduces T count},
  author={Heyfron, Luke E and Campbell, Earl T},
  journal={Quantum Science and Technology},
  volume={4},
  number={1},
  pages={015004},
  year={2019},
  publisher={IOP Publishing}
}

@article{bravyi2005universal,
  title={Universal quantum computation with ideal Clifford gates and noisy ancillas},
  author={Bravyi, Sergey and Kitaev, Alexei},
  journal={Physical Review A—Atomic, Molecular, and Optical Physics},
  volume={71},
  number={2},
  pages={022316},
  year={2005},
  publisher={APS}
}

@article{fowler2012surface,
  title={Surface codes: Towards practical large-scale quantum computation},
  author={Fowler, Austin G and Mariantoni, Matteo and Martinis, John M and Cleland, Andrew N},
  journal={Physical Review A—Atomic, Molecular, and Optical Physics},
  volume={86},
  number={3},
  pages={032324},
  year={2012},
  publisher={APS}
}

@article{haner2018optimizing,
  title={Optimizing quantum circuits for arithmetic},
  author={H{\"a}ner, Thomas and Roetteler, Martin and Svore, Krysta M},
  journal={arXiv preprint arXiv:1805.12445},
  year={2018}
}

@article{munoz2018t,
  title={T-count and qubit optimized quantum circuit design of the non-restoring square root algorithm},
  author={Mu{\~n}oz-Coreas, Edgard and Thapliyal, Himanshu},
  journal={ACM Journal on Emerging Technologies in Computing Systems (JETC)},
  volume={14},
  number={3},
  pages={1--15},
  year={2018},
  publisher={ACM New York, NY, USA}
}

@article{tokura2006critical,
  title={Critical features of colossal magnetoresistive manganites},
  author={Tokura, Y},
  journal={Reports on Progress in Physics},
  volume={69},
  number={3},
  pages={797--851},
  year={2006}
}

@article{yoshida1990construction,
  author  = {Yoshida, Haruo},
  title   = {Construction of higher order symplectic integrators},
  journal = {Physics Letters A},
  volume  = {150},
  number  = {5-7},
  pages   = {262--268},
  year    = {1990},
  doi     = {10.1016/0375-9601(90)90092-3}
}

@article{bai2011suppression,
  title   = {Suppression of phase separation in {LiFePO4} nanoparticles during battery discharge},
  author  = {Bai, Peng and Cogswell, Daniel A. and Bazant, Martin Z.},
  journal = {Nano Letters},
  volume  = {11},
  number  = {11},
  pages   = {4890--4896},
  year    = {2011},
  doi     = {10.1021/nl202764f}
}

@article{cogswell2012coherency,
  title   = {Coherency strain and the kinetics of phase separation in {LiFePO4} nanoparticles},
  author  = {Cogswell, Daniel A. and Bazant, Martin Z.},
  journal = {ACS Nano},
  volume  = {6},
  number  = {3},
  pages   = {2215--2225},
  year    = {2012},
  doi     = {10.1021/nn204177u}
}

@article{dreyer2010thermodynamic,
  title   = {The thermodynamic origin of hysteresis in insertion batteries},
  author  = {Dreyer, Wolfgang and Jamnik, Janko and Guhlke, Clemens and Huth, Robert and Mo{\v{s}}kon, Jo{\v{z}}e and Gaber{\v{s}}{\v{c}}ek, Miran},
  journal = {Nature Materials},
  volume  = {9},
  number  = {5},
  pages   = {448--453},
  year    = {2010},
  doi     = {10.1038/nmat2730}
}

@article{onsager1944crystal,
  title   = {Crystal statistics. {I}. A two-dimensional model with an order-disorder transition},
  author  = {Onsager, Lars},
  journal = {Physical Review},
  volume  = {65},
  number  = {3-4},
  pages   = {117--149},
  year    = {1944},
  doi     = {10.1103/PhysRev.65.117}
}

@article{ishizaki2009unified,
  title   = {Unified treatment of quantum coherent and incoherent hopping dynamics in electronic energy transfer: Reduced hierarchy equation approach},
  author  = {Ishizaki, Akihito and Fleming, Graham R.},
  journal = {The Journal of Chemical Physics},
  volume  = {130},
  number  = {23},
  pages   = {234111},
  year    = {2009},
  doi     = {10.1063/1.3155372}
}

@article{motlagh2025quantum,
  title={Quantum algorithm for vibronic dynamics: case study on singlet fission solar cell design},
  author={Motlagh, Danial and Lang, Robert A and Jain, Paarth and Campos-Gonzalez-Angulo, Jorge A and Maxwell, William and Zeng, Tao and Aspuru-Guzik, Alan and Miguel Arrazola, Juan},
  journal={Quantum Science and Technology},
  volume={10},
  number={4},
  pages={045048},
  year={2025},
  publisher={IOP Publishing}
}

@article{mottonen2004quantum,
  title={Quantum circuits for general multiqubit gates},
  author={M{\"o}tt{\"o}nen, Mikko and Vartiainen, Juha J and Bergholm, Ville and Salomaa, Martti M},
  journal={Physical review letters},
  volume={93},
  number={13},
  pages={130502},
  year={2004},
  publisher={APS}
}

@article{marianetti2001first,
  title={First-principles investigation of the cooperative Jahn-Teller effect for octahedrally coordinated transition-metal ions},
  author={Marianetti, CA and Morgan, D and Ceder, G},
  journal={Physical Review B},
  volume={63},
  number={22},
  pages={224304},
  year={2001},
  publisher={APS}
}

@article{xu2010factors,
  title={Factors affecting Li mobility in spinel LiMn2O4—A first-principles study by GGA and GGA+ U methods},
  author={Xu, Bo and Meng, Shirley},
  journal={Journal of Power Sources},
  volume={195},
  number={15},
  pages={4971--4976},
  year={2010},
  publisher={Elsevier}
}

@article{rodriguez1998electronic,
  title={Electronic crystallization in a lithium battery material: columnar ordering of electrons and holes in the spinel LiMn 2 O 4},
  author={Rodriguez-Carvajal, J and Rousse, G and Masquelier, C and Hervieu, M},
  journal={Physical Review Letters},
  volume={81},
  number={21},
  pages={4660},
  year={1998},
  publisher={APS}
}

\end{document}